\documentclass[journal, 11pt]{IEEEtran}

\usepackage{amsmath, amssymb, amsthm, amsfonts}
\usepackage{mathtools}
\usepackage{bm}
\usepackage{cite}
\usepackage{algorithm}
\usepackage{algorithmic}
\usepackage{booktabs}
\usepackage{float}
\usepackage{hyperref}
\usepackage{enumitem}

\newtheorem{theorem}{Theorem}
\newtheorem{lemma}[theorem]{Lemma}
\newtheorem{corollary}[theorem]{Corollary}
\newtheorem{proposition}[theorem]{Proposition}
\theoremstyle{definition}
\newtheorem{definition}[theorem]{Definition}
\newtheorem{example}[theorem]{Example}
\theoremstyle{remark}
\newtheorem{remark}[theorem]{Remark}

\newcommand{\Tr}{\mathrm{Tr}}
\newcommand{\End}{\mathrm{End}}

\newcommand{\GL}{\mathrm{GL}}

\newcommand{\im}{\mathrm{im}\,}

\begin{document}

\title{Unification of Signal Transform Theory}

\author{Mitchell A. Thornton,~\IEEEmembership{Senior Member,~IEEE}
\thanks{M. A. Thornton is with the Darwin Deason Institute for Cyber Security and the Department of Electrical and Computer Engineering, Southern Methodist University, Dallas, TX 75275 USA (e-mail: mitch@smu.edu).}}

\maketitle

\begin{abstract}
The discrete Fourier transform, discrete cosine transform, Walsh-Hadamard transform, Haar wavelet transform, and their continuous counterparts are unified in this paper under a single representation-theoretic principle: each is the eigenbasis of every covariance invariant under a specific finite or compact group, with its columns constructed from irreducible matrix elements of the group. The unifying framework rests on \emph{Algebraic Diversity}; a framework that identifies the matched group of a covariance as the foundational object for second-order signal processing. The data-dependent Karhunen-Lo\`eve transform emerges as the limiting case in which the matched group is trivial; classical transforms emerge as the cases corresponding to cyclic, dihedral, elementary abelian, and iterated wreath groups. Composition rules cover direct, wreath, and semidirect products. A polynomial-time algorithm for matched-group discovery (developed elsewhere) closes the operational loop. The continuous extension via the Peter-Weyl theorem recovers the Fourier series, Fourier transform, Fourier cosine series, and spherical harmonics by the same principle.
\end{abstract}


\section{Introduction}
\label{sec:introduction}

\IEEEPARstart{T}{he} engineer asked to compress a digital image will reach for the discrete cosine transform; the radio engineer estimating a power spectral density will reach for the discrete Fourier transform; the digital communications engineer designing a code will reach for the Walsh-Hadamard transform; the multiresolution analyst will reach for the Haar wavelet or its smoother descendants; and the statistician told nothing about the structure of the data will fall back on the Karhunen-Lo\`eve transform computed from a sample covariance. Each choice is, in practice, a judgment call. It is informed by experience, by the dominant qualitative characteristics of the data (smoothness, locality, stationarity, sparsity), and by a sometimes-unspoken matching of the data's features to the kernel functions of the available transforms: cosines for slowly varying smooth signals, complex exponentials for stationary oscillatory signals, $\pm 1$-valued basis functions for combinatorial signals, dyadic step functions for hierarchically organized signals. The practical literature is replete with rules of thumb, comparison tables, and empirical recommendations that codify which transform to use when.

The practice of choosing a particular transform to be applied to a signal or data set can be formalized through the use of mathematical principles resulting from the framework of Algebraic Diversity (AD)~\cite{thornton2026ad_framework}. Commonly applied methods of reducing the effects of noise and other random errors through averaging multiple samples in time or space is a tried and true method that exploits temporal and spatial diversity.  Recognizing that signals and data measurement configurations usually comprise one or more degrees of symmetry enable a third dimension of diversity based on algebraic group theory.  Averaging over the elements of a group that represents the symmetry present in a signal augments temporal and spatial diversity with algebraic diversity.

With respect to choosing a transform, the unifying observation is that each classical transform is the eigenbasis of every covariance matrix that respects a specific symmetry: the discrete Fourier transform diagonalizes every cyclic-translation-invariant covariance, the discrete cosine transform diagonalizes every covariance whose even-reflected extension is translation-invariant, the Walsh-Hadamard transform diagonalizes every covariance invariant under bitwise translation on $\{0, 1\}^n$, and the Haar wavelet transform diagonalizes every covariance invariant under hierarchical tree-automorphisms of a binary tree. The Karhunen-Lo\`eve transform of an arbitrary covariance is the limiting case in which no symmetry is imposed: the eigenbasis is determined entirely by the data because nothing else is given.

The unification is mathematical rather than metaphorical. Symmetry, in its most precise sense, is the structure preserved by a group of transformations, and the language for describing symmetries in signal processing is group theory. To say that a covariance has a translation-invariant structure is to say that a particular cyclic group acts on the signal space and that the covariance lies in the commutant of this group action. To say that a covariance has translation and reflection symmetry is to say that a dihedral group acts on the signal space. To say that a signal is hierarchically organized is to say that an iterated wreath product acts on the signal space. The classical transforms are, in this view, the matched eigenbases of the corresponding group actions.

The recently developed AD framework~\cite{thornton2026ad_framework} makes this observation operational. AD identifies the \emph{matched group} of a covariance as the largest finite group of permutations of the signal indices that leaves the covariance invariant, and shows that the matched group is the foundational object of second-order signal processing. Given a covariance and its matched group, the optimal transform is determined by representation theory: by a classical structure theorem of Wedderburn, Schur, and Peter-Weyl, the eigenbasis of every matrix in the commutant of the group action is constructed from the matrix elements of the irreducible representations of the group. The matched group is to the optimal transform what the period of a periodic function is to its Fourier series: an invariant of the underlying object that determines the natural basis in which to represent it.

\subsection{The hierarchy of optimality}

The Karhunen-Lo\`eve transform of a covariance is optimal in the strongest sense: it minimizes the expected mean-squared truncation error of any $K$-term orthonormal expansion, for every $K$, and it produces uncorrelated coefficients. It is the unique transform with this property, up to ordering and unimodular phase. Different signal classes admit different additional structural constraints, however, and these constraints simplify the KLT, or rather, they pin down the KLT to a specific fixed basis whose columns are determined by representation theory rather than by data-dependent eigendecomposition.

The simplest such structural constraint is wide-sense stationarity. A wide-sense stationary process on uniformly spaced samples has a covariance that depends only on the lag $|s - t|$, and the corresponding covariance matrix is symmetric Toeplitz (or, with cyclic boundary conditions, circulant). The circulant form is invariant under the cyclic group $\mathbb{Z}_M$ acting by translation; the Toeplitz form is not, because a cyclic shift carries the last sample to the first, and the discrepancy between the two is confined to the wrapped corners and is $O(1/M)$ in relative energy (Remark~\ref{rem:toeplitz_wrap}). The matched transform of the circulant form is the discrete Fourier transform, with columns given by the characters of $\mathbb{Z}_M$. The intuition for why uniform sampling implies cyclic symmetry is straightforward: if the underlying process is stationary, then shifting the sampling clock by one sample produces a statistically equivalent measurement vector, so the covariance of a periodized window is invariant under the shift, and the shift generates the cyclic group on the sample indices. Stationarity therefore selects $\mathbb{Z}_M$ as the matched group, exactly on a periodized window and asymptotically on a long free window, and the DFT as the matched transform.

A more subtle example is the sampled first-order autoregressive (AR(1)) process $X_t = \phi X_{t-1} + \varepsilon_t$ on uniformly spaced samples. Its covariance on a finite window is symmetric Toeplitz, and its inverse (the precision matrix) is tridiagonal with constant diagonal and constant off-diagonals, except at the two endpoints. Both are invariant under the reflection $i \mapsto M - 1 - i$ (because $|i - j| = |M - 1 - i - (M - 1 - j)|$), and neither is invariant under the cyclic shift of the window. The permutation symmetry of the finite window is therefore the reflection alone, a group of order two whose permutation representation is far from multiplicity-free, and the exact eigenbasis of the AR(1) covariance depends on $\phi$~\cite{ray1970further}. The classical fact that the DCT is asymptotically optimal for AR(1) signals~\cite{ahmed1974discrete, rao1990discrete}, and that the JPEG image compression standard chooses the DCT because natural-image patches are approximately AR(1), has a precise reading in the present framework: the DCT-II is the exact matched transform of a dihedral group acting on the even-reflected extension of the signal (Section~\ref{sec:dct}), and the AR(1) covariance lies close to the invariant class of that action, approaching it as the correlation coefficient approaches one~\cite{strang1999discrete} and as the window grows. The DCT is to that dihedral action what the DFT is to the cyclic group: the matched eigenbasis.

\subsection{The harder question: finding the matched group}

The two examples above illustrate the framework when the symmetry is known: stationarity on a periodized window yields cyclic, and even-reflected extension yields dihedral. In most practical settings, the symmetry is not known in advance. The signal arrives as a covariance estimate, and the relevant matched group is a property of the unknown underlying process that must be inferred from data. This is the matched-group discovery problem of AD.

For an arbitrary covariance, the matched group is the largest subgroup of $S_M$ that commutes with it, which is a non-trivial combinatorial-algebraic object. A na\"ive search over subgroups of $S_M$ is exponential in $M$. A polynomial-time procedure for matched-group discovery was developed in the companion manuscript~\cite{thornton2026double_commutator}: the procedure casts the matched-group search as a generalized eigenvalue problem in a relaxed double-commutator form, and solves it sequentially with a deflation step that builds the matched group one generator at a time. We follow the naming convention of~\cite{thornton2026ad_framework} and use Discrete Algebraic Diversity (DAD)~\cite{thornton2026spectral_estimation} and Continuous Algebraic Diversity (CAD)~\cite{thornton2026continuous_ad} to denote the matched-group discovery procedure that is sometimes referred to as the DAD-CAD relaxation in recognition of its applicability to both the finite-group and the compact Lie group settings. The procedure is sound (every returned generator is a genuine symmetry of the covariance) and complete (the procedure recovers the matched group in full when run with an appropriate canonical generator basis), and its total cost is polynomial in $M$ and the depth of the discovered structure. The operational consequence is that finding the matched group, and therefore computing the optimal transform, is a polynomial-time problem.

A geometric perspective on the polynomial-time discovery procedure is the following: as the matched group is varied continuously (in a sense that requires care for finite groups but is natural for compact Lie groups), the matched transform varies continuously as a point on a smooth manifold of unitary matrices. The DFT, DCT, Walsh-Hadamard, Haar wavelet, and KLT are not separate constructions floating in different parts of the space of transforms, they are nearby points on this manifold, parameterized by their matched groups. The full Karhunen-Lo\`eve transform of an arbitrary covariance is the trivial-group point (no symmetry constraint), the DFT is the cyclic-group point, the DCT is the dihedral-group point, and so on. The matched-group discovery procedure navigates this manifold, identifying the nearest symmetry-respecting point to a given empirical covariance.

\subsection{Beyond the basic groups}

Once the basic theory is in place, the entire catalog of classical transforms emerges from increasingly elaborate symmetry groups. The Walsh-Hadamard transform diagonalizes covariances invariant under the elementary abelian 2-group $\mathbb{Z}_2^n$ acting by bitwise translation. Equivalently, signals indexed by binary strings whose statistics respect bitwise flipping. The Haar wavelet transform diagonalizes covariances invariant under the iterated dyadic-cyclic wreath product $\mathbb{Z}_2 \wr \mathbb{Z}_2 \wr \cdots \wr \mathbb{Z}_2$ acting by tree-automorphisms on a complete binary tree. Equivalently, signals whose statistics respect hierarchical reorganization at every scale. The hierarchical Karhunen-Lo\`eve basis~\cite{thornton2026continuous_ad} generalizes the Haar wavelet to arbitrary branching factors and to non-cyclic within-level groups. Composition rules covering direct products, wreath products, and semidirect products of groups permit the construction of matched transforms for compound symmetry classes.

The framework extends naturally from finite groups to compact Lie groups via the general Peter-Weyl theorem. The continuous Fourier series is the matched transform for circle-rotation-invariant kernels on $S^1$, the continuous Fourier transform is the matched transform for translation-invariant kernels on $\mathbb{R}$ (after passage through the Plancherel theorem for non-compact groups), the continuous Fourier cosine series is the matched transform for $O(2)$-invariant kernels on the half-circle, and the spherical harmonics are the matched transform for $SO(3)$-invariant kernels on $S^2$. Mercer's theorem and the continuous Karhunen-Lo\`eve theorem replace the matrix eigendecomposition; the underlying group-theoretic principle is identical.

\subsection{Related work}
\label{subsec:related_work}

Several existing research threads overlap or are similar to aspects of the unification described here. Each of these capture an important piece of the framework without completing it. We describe them in turn and contrast them with the AD viewpoint, with the goal of positioning the present contribution within the surrounding intellectual landscape and to acknowledge the important work that has come before.

\emph{Algebraic signal processing (ASP).} P\"uschel and Moura~\cite{puschel2008algebraic1, puschel2008algebraic2} developed an algebraic framework that organizes signal processing systems into \emph{signal models}, each consisting of an algebra, a module over that algebra, and a designated shift operator. Given a signal model, the corresponding Fourier transform is derived algorithmically as the basis that diagonalizes the algebra. ASP recovers the DFT, DCT, DST families, and related transforms with remarkable economy, and it has spawned a productive program of derived fast algorithms. The starting point for ASP is the choice of shift operator: time translation gives the DFT, the symmetric time-and-space shift gives the DCT, and so on. The transform is downstream of the operator. ASP does not, however, address the inverse question of how to choose the operator from data, nor does it offer a discovery procedure for the matched algebra of an empirically estimated covariance. The framework is also closer to a polynomial-algebra view than to a representation-theoretic one, which makes the extension to genuinely non-abelian symmetries such as the iterated wreath product less immediate.

\emph{Group harmonic analysis.} The harmonic-analysis-on-groups literature, with foundational textbook treatments by Diaconis~\cite{diaconis1988group} (statistical applications) and Terras~\cite{terras1999fourier} (computational), and computational developments by Rockmore~\cite{rockmore1995applications}, Maslen and Rockmore~\cite{maslen1997generalized}, and Foote, Mirchandani, Rockmore, Healy, and Olson~\cite{foote2000wreath, mirchandani2003wreath}, develops the formal apparatus for Fourier analysis on finite and compact groups, including efficient algorithms for the wreath product and other compound structures. The classical transforms (DFT, Walsh-Hadamard, and others) appear naturally as Fourier transforms on specific groups in this literature, and the wreath product appears explicitly in connection with multiresolution analysis. The group is presupposed as part of the modeling: given $G$, the literature constructs the corresponding Fourier transform and its fast algorithm. The matched-group discovery question of identifying the correct $G$ for a given covariance is not addressed, and the literature does not state the unification claim that the entire classical catalog of transforms shares a single representation-theoretic principle.

\emph{Transform coding.} The transform coding literature, beginning with Jayant and Noll~\cite{jayant1984digital} and extended by Mallat~\cite{mallat1989multiresolution}, Vetterli, Kova\v{c}evi\'c, and Goyal~\cite{vetterli2014foundations}, and Goyal~\cite{goyal2001theoretical}, develops the rate-distortion-optimal theory of coding under an assumed transform. Given a transform basis, the optimal bit-allocation strategy across coefficients can be computed, and the resulting compression performance can be characterized as a function of the source statistics. The transform is a design parameter in this framework, selected by the engineer based on assumed signal characteristics (the DCT for AR(1)-like signals, the DFT for stationary signals, the wavelet for piecewise-smooth signals). The transform coding literature does not provide a theoretical principle for selecting the transform from the signal class; the choice is informed by the same kind of feature-to-kernel matching that we identify as the ad hoc practice this paper seeks to replace.

\emph{Equivariant machine learning.} A vigorous recent literature on equivariant deep learning, including the group-equivariant convolutional networks of Cohen and Welling~\cite{cohen2016group}, the geometric deep learning manifesto of Bronstein et al.~\cite{bronstein2017geometric}, equivariant multilayer perceptrons for arbitrary matrix groups~\cite{finzi2021practical}, group-equivariant self-attention~\cite{romero2021group}, and equivariant maps for hierarchical structures~\cite{wang2020equivariant}, designs neural network architectures whose layers commute with the action of a specified symmetry group. The group is typically chosen as part of the architecture, on prior knowledge of the data's symmetry. A parallel symmetry-discovery thread has emerged: Augerino~\cite{benton2020learning} learns invariances by augmentation, Quessard et al.~\cite{quessard2020learning} learn the group structure of dynamical environments, LieGAN~\cite{yang2023generative} discovers Lie group structure from generative-adversarial training, and the group-algebraic tensors of Hoyos et al.~\cite{hoyos2026group} learn a physical symmetry from labelled training data through an equivariant model. These discovery procedures operate on neural network parameters via gradient-based optimization rather than on second-order statistics of the data, and they target invariance of the architecture rather than diagonalization of the covariance. The two perspectives are complementary: the equivariant deep learning literature designs architectures that respect a given symmetry, while the AD framework discovers the matched group of the data and computes the corresponding optimal linear transform, which then serves naturally as an equivariant first layer for a deep network.

\emph{Structured covariance estimation.} The structured covariance estimation literature in statistics, including the Toeplitz testing of Anderson~\cite{anderson1948toeplitz}, the banded estimation of Bickel and Levina~\cite{bickel2008regularized}, the group-symmetric regularization of Shah and Chandrasekaran~\cite{shah2012group}, the optimal-rate Toeplitz estimation of Cai et al.~\cite{cai2011optimal}, and the well-conditioned shrinkage of Ledoit and Wolf~\cite{ledoit2004well}, develops estimators for covariances that respect a posited structural form (Toeplitz, banded, sparse, low-rank, group-symmetric). The structure type is an analyst input: the practitioner posits Toeplitz structure, or banded structure, or invariance under a specific group, and the methods deliver an estimator that improves on the unstructured sample covariance under that assumption. Shah and Chandrasekaran~\cite{shah2012group} in particular treats the group-symmetric case and is the closest existing work to the matched-group framework, but their procedure still requires the group as input and does not provide a discovery algorithm.

\emph{What is new.} None of these threads claims the unification we propose: that every classical signal-processing transform is the matched eigenbasis of every covariance invariant under a specific group, with the columns constructed from irreducible matrix elements via a uniform representation-theoretic procedure. The Algebraic Diversity framework provides three contributions that together constitute the unification: a single mathematical principle (the multiplicity-free decomposition theorem of Section~\ref{sec:main}) that mechanizes the matched-transform construction for any group; a polynomial-time discovery procedure~\cite{thornton2026double_commutator} for the matched group of an empirical covariance, removing the need to specify the group as a design parameter; and a continuous extension via Peter-Weyl that recovers the classical continuous transforms by the same principle. The five existing threads above each capture one aspect of this triad: ASP and group harmonic analysis articulate the construction given the operator or group; transform coding articulates the optimality under an assumed transform; equivariant ML articulates the architectural consequences of a given symmetry; and structured covariance estimation articulates statistical estimation under a posited structure. The unification appears when these are assembled around the matched group as the foundational object.

\subsection{Preview: data-driven matched-group discovery}
\label{subsec:preview_discovery}

The matched-group principle developed in the body of this paper is conceptually clean, given the matched group, the transform follows mechanically, but the practical obstacle is that the matched group is rarely known in advance. The data scientist with an estimated covariance is in the position of needing to identify the underlying symmetry before the matched transform can be computed. A na\"ive enumeration over subgroups of $S_M$ is exponential in the signal dimension $M$, which makes the question whether the matched-group discovery problem is tractable a substantive theoretical concern. Section~\ref{sec:discovery} addresses this question and reports that the answer is positive: matched-group discovery can be performed in polynomial time by a relaxation that casts the search as a generalized eigenvalue problem in a double-commutator form. The polynomial-time algorithm, developed in detail in the companion manuscript~\cite{thornton2026double_commutator}, closes the operational loop: the practitioner with an empirical covariance can mechanically discover the matched group and then mechanically construct the matched transform, without any expert judgment about which classical transform to apply.

\subsection{Organization}

Section~\ref{sec:background} reviews the necessary background on finite group representations, including Schur's lemma, characters, and the Peter-Weyl theorem. Section~\ref{sec:main} states and proves the main theorem: when a finite group acts multiplicity-freely on a signal space, the eigenbasis of every invariant covariance is a fixed unitary whose columns are constructed from the irreducible matrix elements of the group. Section~\ref{sec:trivial} treats the trivial-group case, which recovers the data-dependent Karhunen-Lo\`eve transform. Section~\ref{sec:dft} treats the cyclic group and the DFT, with the uniform-sampling-implies-stationarity intuition made explicit. Section~\ref{sec:dct} treats the dihedral group and the DCT, with the even-extension construction and the AR(1) approximation made explicit. Section~\ref{sec:discovery} discusses the matched-group discovery problem and the polynomial-time algorithm of~\cite{thornton2026double_commutator}, along with the manifold-in-the-continuum framing. Section~\ref{sec:composition} states composition rules for direct products, wreath products, and semidirect products. Sections~\ref{sec:hadamard} through~\ref{sec:iterated_wreath} treat the further special cases (Walsh-Hadamard, Haar wavelet, hierarchical Karhunen-Lo\`eve basis). Section~\ref{sec:numerical_verification} reports numerical verification. Section~\ref{sec:lie_groups} maps the framework to compact Lie groups via the general Peter-Weyl theorem. Section~\ref{sec:continuous_transforms} derives the classical continuous transforms (Fourier series, Fourier transform, Fourier cosine series, spherical harmonics) as the corresponding matched transforms. Section~\ref{sec:discussion} discusses extensions and open questions, and Section~\ref{sec:conclusion} concludes.

\section{Background: Group Representations and Characters}
\label{sec:background}

We collect the representation-theoretic prerequisites. Standard references are Serre~\cite{serre1977linear}, Fulton and Harris~\cite{fulton1991representation}, Curtis and Reiner~\cite{curtis1962representation}, and Terras~\cite{terras1999fourier} (the last specialized to applications to harmonic analysis on finite groups).

\subsection{Groups, actions, and permutation representations}

A \emph{group} $G$ is a set equipped with a binary operation $G \times G \to G$, written multiplicatively, satisfying associativity, the existence of an identity $e$, and the existence of inverses. We restrict attention throughout to finite groups, $|G| < \infty$.

A \emph{group action} of $G$ on a set $X$ is a homomorphism $G \to \mathrm{Sym}(X)$, where $\mathrm{Sym}(X)$ is the symmetric group on $X$. Equivalently, an action assigns to each $g \in G$ a permutation $g: X \to X$ such that the identity acts as the identity and the action respects group multiplication.

The \emph{permutation representation} associated to an action on a finite set $X$ with $|X| = M$ is the homomorphism $\pi: G \to \GL(\mathbb{C}^M)$ defined by $\pi(g) e_x = e_{g \cdot x}$, where $\{e_x : x \in X\}$ is the standard basis of $\mathbb{C}^M$. Each $\pi(g)$ is a permutation matrix, hence unitary.

\subsection{Specific groups used in this paper}
\label{subsec:specific_groups}

For convenient reference, the groups of principal interest in this paper, with the notation adopted throughout, are collected below. Finite groups appear in Sections~\ref{sec:main}--\ref{sec:iterated_wreath}; compact (and a few non-compact) Lie groups appear in Sections~\ref{sec:lie_groups}--\ref{sec:continuous_transforms} and in the FrFT/LCT discussion.

\begin{itemize}[leftmargin=*]
\item $\{e\}$, the \emph{trivial group}, with one element. Corresponds to the no-symmetry limit; matched transform is the data-dependent Karhunen-Lo\`eve transform (Section~\ref{sec:trivial}).
\item $\mathbb{Z}_M$, the \emph{cyclic group of order $M$}, generated by a single element of order $M$. Matched transform is the DFT (Section~\ref{sec:dft}).
\item $D_M = \mathbb{Z}_M \rtimes \mathbb{Z}_2$, the \emph{dihedral group of order $2M$}, generated by a rotation and a reflection. Acting on the $M$-cycle its matched transform is the real Fourier (Hartley) basis (Remark~\ref{rem:hartley}); acting as $D_{2M}$ on the even-reflected extension $\mathbb{Z}_{2M}$ it yields the DCT (Section~\ref{sec:dct}).
\item $S_M$, the \emph{symmetric group} on $M$ symbols, the group of all $M!$ permutations of $\{0, 1, \ldots, M-1\}$. Appears as the automorphism group of the complete graph $K_M$ (Section~\ref{subsec:graph_sp}) and as the across-block exchangeability factor in hybrid wreath products.
\item $\mathbb{Z}_2^n$, the \emph{elementary abelian $2$-group of order $2^n$}, the direct product of $n$ copies of $\mathbb{Z}_2$. Matched transform is the Walsh-Hadamard transform (Section~\ref{sec:hadamard}); also the matched group of the Reed-Muller and arithmetic transforms with different basis choices (Section~\ref{sec:rm_arithmetic}).
\item $G_1 \times G_2$, the \emph{direct product}, the cartesian product with componentwise group operation. Composition rule (Theorem~\ref{thm:direct_product_rule}) gives the tensor product of the constituent matched transforms; the 2D DFT is the canonical example.
\item $G \wr H = G^n \rtimes H$, the \emph{wreath product} of $G$ and $H$, with $H$ acting on $n$ copies of $G$ by permutation. The \emph{iterated dyadic-cyclic wreath product} $W_L = \mathbb{Z}_2 \wr \mathbb{Z}_2 \wr \cdots \wr \mathbb{Z}_2$ ($L$ factors) has order $2^{2^L - 1}$ and is the matched group of the Haar wavelet (Section~\ref{sec:iterated_wreath}).
\item $G_1 \wr G_2 \wr \cdots \wr G_L$, the \emph{general iterated wreath product}, the matched group of hierarchical Karhunen-Lo\`eve bases on $L$-level rooted trees (Theorem~\ref{thm:iterated_wreath}).
\item $G_0 \wr S_K$, the \emph{hybrid wreath product} of a within-block group $G_0$ and the symmetric group $S_K$ on $K$ blocks. Discussed as the fifth proven-solvable class in Section~\ref{subsec:five_classes}.
\item $U(1) \cong S^1$, the \emph{unit circle} as a compact one-dimensional Lie group. Matched transform on $L^2(S^1)$ is the Fourier series (Section~\ref{subsec:fourier_series}).
\item $SO(2)$, the \emph{special orthogonal group} in two dimensions, the rotation group of the plane. Acts on $L^2(\mathbb{R})$ via the metaplectic representation; matched transform is the Hermite-Gauss basis (Section~\ref{subsec:frft}).
\item $SO(3)$, the \emph{rotation group of three-dimensional space}. Matched transform on $L^2(S^2)$ is the spherical-harmonic basis (Section~\ref{subsec:spherical_harmonics}); the radial component on $L^2(\mathbb{R}^d)$ for $d = 2\nu + 2$ gives the continuous Hankel transform of order $\nu$ (Section~\ref{subsec:hankel}).
\item $SO(d)$, the \emph{rotation group in $d$ dimensions} for general $d \geq 2$. Generalization of $SO(3)$ to higher dimensions; appears in the higher-dimensional spherical-harmonic and Hankel-transform discussion (Sections~\ref{subsec:spherical_harmonics} and~\ref{subsec:hankel}).
\item $SL(2, \mathbb{R})$, the \emph{special linear group} of real $2 \times 2$ matrices of unit determinant. Acts on $L^2(\mathbb{R})$ via the metaplectic representation; the family of unitary actions parameterized by elements of $SL(2, \mathbb{R})$ is the linear canonical transform (LCT) (Remark~\ref{rem:lct}).
\item $\mathrm{Aff}(\mathbb{R}_+) = \{(a, b) : a > 0, b \in \mathbb{R}\}$, the \emph{affine ``$ax + b$'' group} on the positive real line. Matched transform is the continuous wavelet transform (Section~\ref{subsec:continuous_wavelet}).
\end{itemize}

For brevity we write ``$G$ acts on $\mathbb{C}^M$'' for ``$\pi: G \to U(M)$ is a unitary representation,'' suppressing the representation $\pi$ when context makes it clear. The same convention applies to actions on continuous function spaces.

\subsection{Unitary representations}

\begin{definition}[Unitary representation]
\label{def:unitary_rep}
A \emph{unitary representation} of a finite group $G$ on a finite-dimensional complex inner product space $V$ is a homomorphism $\pi: G \to U(V)$, where $U(V)$ is the unitary group on $V$. The dimension of $V$ is called the \emph{dimension} of the representation, written $\dim \pi$.
\end{definition}

Two representations $\pi: G \to U(V)$ and $\pi': G \to U(V')$ are \emph{equivalent}, written $\pi \cong \pi'$, if there exists a unitary isomorphism $T: V \to V'$ with $\pi'(g) = T \pi(g) T^{-1}$ for every $g \in G$. The set of equivalence classes of irreducible unitary representations is denoted $\hat{G}$, the \emph{dual} or \emph{unitary dual} of $G$.

\begin{definition}[Subrepresentation, irreducibility]
\label{def:irrep}
A subspace $W \subseteq V$ is \emph{invariant} under $\pi$ if $\pi(g) W \subseteq W$ for every $g \in G$. The restriction $\pi|_W : G \to U(W)$ is then a unitary subrepresentation. A representation is \emph{irreducible} if its only invariant subspaces are $\{0\}$ and $V$ itself. A representation that is not irreducible is \emph{reducible}.
\end{definition}

\begin{theorem}[Maschke's theorem]
\label{thm:maschke}
Every finite-dimensional unitary representation of a finite group is completely reducible: it decomposes as a direct sum of irreducible subrepresentations.
\end{theorem}

\begin{proof}[Sketch] If $W \subseteq V$ is an invariant subspace, its orthogonal complement $W^\perp$ is also invariant by unitarity of the action. Iterate. See~\cite[Section 1.4]{serre1977linear}.
\end{proof}

\subsection{Schur's lemma and the structure of the commutant}

\begin{theorem}[Schur's lemma]
\label{thm:schur}
Let $\pi: G \to U(V)$ and $\pi': G \to U(V')$ be two irreducible unitary representations of $G$. Let $T: V \to V'$ be a linear map such that $T \pi(g) = \pi'(g) T$ for every $g \in G$ (such a $T$ is called an \emph{intertwiner}). Then:
\begin{enumerate}
\item[(i)] If $\pi \not\cong \pi'$, then $T = 0$.
\item[(ii)] If $\pi = \pi'$, then $T = \lambda I_V$ for some scalar $\lambda \in \mathbb{C}$.
\end{enumerate}
\end{theorem}

\begin{proof}
For (i), the kernel $\ker T \subseteq V$ is $\pi$-invariant: if $v \in \ker T$, then $T(\pi(g) v) = \pi'(g) T v = 0$, so $\pi(g) v \in \ker T$. By irreducibility of $\pi$, $\ker T \in \{0, V\}$. Similarly, the image $\im T \subseteq V'$ is $\pi'$-invariant, so $\im T \in \{0, V'\}$. The only consistent combinations giving $T \neq 0$ are $\ker T = 0$ and $\im T = V'$, making $T$ an isomorphism, which contradicts $\pi \not\cong \pi'$.

For (ii), $V = V'$ and $T$ is an intertwiner. Let $\lambda$ be an eigenvalue of $T$, which exists because $\mathbb{C}$ is algebraically closed. Then $T - \lambda I$ is also an intertwiner, and $\ker(T - \lambda I) \neq 0$. By the argument above, $\ker(T - \lambda I) = V$, so $T = \lambda I$.
\end{proof}

The most important consequence of Schur's lemma for our purposes concerns the \emph{commutant algebra}.

\begin{definition}[Commutant]
\label{def:commutant}
The \emph{commutant} of a representation $\pi: G \to U(V)$ is the algebra
\[
\mathcal{A}_\pi = \{T \in \End(V) : T \pi(g) = \pi(g) T \text{ for all } g \in G\}.
\]
\end{definition}

The commutant is a $*$-subalgebra of $\End(V)$ (closed under composition and Hermitian conjugation) and is sometimes called the \emph{centralizer} or \emph{$G$-intertwiner algebra}.

\subsection{Isotypic decomposition and the structure theorem}

\begin{definition}[Isotypic component]
\label{def:isotypic}
Let $\pi: G \to U(V)$ be a unitary representation and $\rho \in \hat{G}$ an irreducible representation. The \emph{$\rho$-isotypic component} of $V$, denoted $V^{(\rho)}$, is the sum of all subrepresentations of $V$ isomorphic to $\rho$. The integer $m_\rho := \dim V^{(\rho)} / \dim \rho$ is the \emph{multiplicity} of $\rho$ in $\pi$.
\end{definition}

\begin{theorem}[Isotypic decomposition]
\label{thm:isotypic}
Every unitary representation $\pi: G \to U(V)$ decomposes as an orthogonal direct sum of isotypic components,
\[
V = \bigoplus_{\rho \in \hat{G}} V^{(\rho)},
\]
and each isotypic component admits a (non-canonical) further decomposition
\[
V^{(\rho)} \cong \mathbb{C}^{m_\rho} \otimes V_\rho,
\]
where $V_\rho$ is a fixed model space for $\rho$ and $\mathbb{C}^{m_\rho}$ is the corresponding \emph{multiplicity space}.
\end{theorem}

\begin{proof}
By Maschke's theorem (Theorem~\ref{thm:maschke}), $V$ decomposes as a direct sum of irreducible subrepresentations. Two irreducibles in this decomposition are either equivalent or orthogonal (by Schur's lemma applied to the projection); summing equivalent pieces gives the isotypic components. The decomposition $V^{(\rho)} \cong \mathbb{C}^{m_\rho} \otimes V_\rho$ chooses an arbitrary intertwiner basis for the $m_\rho$ copies of $\rho$ in $V^{(\rho)}$.
\end{proof}

\begin{theorem}[Structure of the commutant]
\label{thm:commutant_structure}
The commutant $\mathcal{A}_\pi$ decomposes as a direct sum of full matrix algebras,
\[
\mathcal{A}_\pi \cong \bigoplus_{\rho \in \hat{G}, \, m_\rho \geq 1} \End(\mathbb{C}^{m_\rho}).
\]
Equivalently, there exists a unitary $U \in U(V)$ such that for every $T \in \mathcal{A}_\pi$,
\[
U^* T U \;=\; \bigoplus_{\rho \in \hat{G}} I_{d_\rho} \otimes \tilde{T}_\rho,
\]
where $d_\rho = \dim \rho$ and $\tilde{T}_\rho \in \End(\mathbb{C}^{m_\rho})$.
\end{theorem}

\begin{proof}
The commutant of $\pi|_{V^{(\rho)}} : G \to U(V^{(\rho)})$ is exactly $\End(\mathbb{C}^{m_\rho})$ by Schur's lemma: the action on $\mathbb{C}^{m_\rho} \otimes V_\rho$ commutes with $G$ if and only if it acts as the identity on the $V_\rho$ factor. Direct-summing across isotypic components gives the displayed decomposition. The unitary $U$ is constructed by choosing an orthonormal basis of each $V^{(\rho)}$ compatible with the $\mathbb{C}^{m_\rho} \otimes V_\rho$ factorization.
\end{proof}

\subsection{Characters and the Peter-Weyl theorem}

\begin{definition}[Character]
\label{def:character}
The \emph{character} of a representation $\pi: G \to U(V)$ is the function $\chi_\pi: G \to \mathbb{C}$ defined by $\chi_\pi(g) = \Tr(\pi(g))$.
\end{definition}

Characters are class functions (constant on conjugacy classes of $G$) and satisfy $\chi_\pi(e) = \dim \pi$. The irreducible characters $\{\chi_\rho : \rho \in \hat{G}\}$ form an orthonormal basis of the space of class functions on $G$ under the inner product
\[
\langle \chi, \chi' \rangle_G = \frac{1}{|G|} \sum_{g \in G} \chi(g) \overline{\chi'(g)}.
\]

\begin{theorem}[Peter-Weyl theorem for finite groups]
\label{thm:peter_weyl}
For each irreducible representation $\rho \in \hat{G}$ with dimension $d_\rho$, fix an orthonormal basis of the model space $V_\rho$ and let $\rho_{ij}(g) := \langle \rho(g) e_j, e_i \rangle$ denote the $(i, j)$-entry of the matrix $\rho(g)$ in this basis. Then the set of functions
\[
\Big\{ \sqrt{\tfrac{d_\rho}{|G|}} \, \rho_{ij} : \rho \in \hat{G}, \; 1 \leq i, j \leq d_\rho \Big\}
\]
forms an orthonormal basis of $L^2(G)$. In particular, the total dimension count is $\sum_{\rho \in \hat{G}} d_\rho^2 = |G|$.
\end{theorem}

\begin{proof}
See~\cite[Theorem 7]{serre1977linear} or~\cite[Chapter 3]{fulton1991representation}. The orthogonality relations are
\[
\frac{1}{|G|} \sum_{g \in G} \rho_{ij}(g) \overline{\rho'_{kl}(g)} = \frac{1}{d_\rho} \delta_{\rho\rho'} \delta_{ik} \delta_{jl},
\]
which follow from Schur's lemma applied to the operator $S(T) = \sum_g \rho(g) T \rho'(g)^*$.
\end{proof}

\subsection{Multiplicity-free representations}

\begin{definition}[Multiplicity-free representation]
\label{def:multiplicity_free}
A unitary representation $\pi: G \to U(V)$ is \emph{multiplicity-free} if every irreducible representation $\rho \in \hat{G}$ appears in $\pi$ with multiplicity $m_\rho \in \{0, 1\}$.
\end{definition}

\begin{lemma}[Commutant of a multiplicity-free representation is commutative]
\label{lem:multiplicity_free_commutant}
A unitary representation $\pi$ is multiplicity-free if and only if its commutant $\mathcal{A}_\pi$ is a commutative algebra.
\end{lemma}

\begin{proof}
By Theorem~\ref{thm:commutant_structure}, $\mathcal{A}_\pi \cong \bigoplus_\rho \End(\mathbb{C}^{m_\rho})$. This direct sum is commutative if and only if each summand is commutative, equivalently $m_\rho \leq 1$ for all $\rho$.
\end{proof}

The reader is invited to keep Lemma~\ref{lem:multiplicity_free_commutant} in mind: it is the structural property that will make the matched transform group-determined, since every $T$ in a commutative algebra is simultaneously diagonalizable by a fixed unitary basis.

\subsection{Covariance matrices and the matched group}

\begin{definition}[$G$-invariant covariance]
\label{def:g_invariant_cov}
Let $\pi: G \to U(\mathbb{C}^M)$ be a unitary representation, and let $\mathbf{R} \in \mathbb{C}^{M \times M}$ be a positive semidefinite Hermitian matrix. We say $\mathbf{R}$ is \emph{$G$-invariant} (with respect to $\pi$) if $\mathbf{R} \in \mathcal{A}_\pi$, equivalently $\pi(g) \mathbf{R} \pi(g)^* = \mathbf{R}$ for every $g \in G$.
\end{definition}

\begin{definition}[Matched group of a covariance]
\label{def:matched_group}
Let $\mathbf{R} \in \mathbb{C}^{M \times M}$ be a Hermitian PSD matrix. The \emph{matched group} of $\mathbf{R}$ is the largest subgroup $G^*(\mathbf{R}) \subseteq S_M$ of the symmetric group on $M$ symbols such that the standard permutation action $\pi : G^*(\mathbf{R}) \to U(\mathbb{C}^M)$ satisfies $\mathbf{R} \in \mathcal{A}_\pi$.
\end{definition}

The matched group is well-defined: if $\mathbf{P}_g \mathbf{R} = \mathbf{R} \mathbf{P}_g$ and $\mathbf{P}_h \mathbf{R} = \mathbf{R} \mathbf{P}_h$, then $\mathbf{P}_g \mathbf{P}_h$ and $\mathbf{P}_g^{-1}$ also commute with $\mathbf{R}$, so the set of commuting permutations forms a subgroup. The matched group is the heart of the Algebraic Diversity framework~\cite{thornton2026ad_framework}: it captures the second-order combinatorial symmetry of the underlying random process.

\section{The Main Theorem}
\label{sec:main}

We are ready to state and prove the central result of the paper. It says that whenever a finite group acts multiplicity-freely on a signal space, the eigenbasis of every invariant covariance is determined entirely by the group, and the columns of the transform matrix are constructed from the irreducible matrix elements of the group.

\begin{theorem}[Main theorem: group-determined optimal transform]
\label{thm:main}
Let $G$ be a finite group and $\pi: G \to U(\mathbb{C}^M)$ a unitary representation. Suppose $\pi$ is multiplicity-free, with isotypic decomposition
\[
\mathbb{C}^M = \bigoplus_{\rho \in \hat{G}} V^{(\rho)},
\]
where $V^{(\rho)} \cong V_\rho$ when $m_\rho = 1$ and $V^{(\rho)} = \{0\}$ when $m_\rho = 0$,
and let $S(\pi) = \{\rho \in \hat{G} : m_\rho = 1\}$ denote the support of the representation. Then:
\begin{enumerate}
\item[(i)] \emph{Existence of a fixed eigenbasis.} There exists a unitary $U_G \in U(\mathbb{C}^M)$, depending only on $G$ and $\pi$ and not on any particular covariance, such that for every $G$-invariant Hermitian PSD matrix $\mathbf{R} \in \mathcal{A}_\pi$,
\[
U_G^* \mathbf{R} U_G \;=\; \mathrm{diag}(\lambda_1, \lambda_2, \ldots, \lambda_M)
\]
is diagonal. In particular, $U_G$ is an eigenbasis of every $G$-invariant covariance.

\item[(ii)] \emph{Construction from irreducible matrix elements.} Fix an orthonormal basis of each model space $V_\rho$ for $\rho \in S(\pi)$, with corresponding matrix elements $\rho_{ij}(g) = \langle \rho(g) e_j, e_i \rangle$. The columns of $U_G$, suitably ordered, are linear combinations of the matrix-element functions $\{\rho_{ij}\}_{\rho \in S(\pi), \, 1 \leq i, j \leq d_\rho}$, evaluated on a fundamental set of orbit representatives for the action of $G$ on $\{1, \ldots, M\}$.

\item[(iii)] \emph{Uniqueness.} The basis $U_G$ is unique up to:
\begin{itemize}
\item a complex unimodular scaling of each column corresponding to a one-dimensional irrep (the scaling does not affect the eigenbasis property);
\item a unitary rotation within each $d_\rho$-dimensional irrep block when $d_\rho > 1$ (in this case the eigenvalue of any $G$-invariant $\mathbf{R}$ is $d_\rho$-fold degenerate on the $\rho$-isotypic subspace, and any orthonormal basis of $V^{(\rho)}$ is a valid choice).
\end{itemize}

\item[(iv)] \emph{Eigenvalues.} The eigenvalue $\lambda$ assigned to the $\rho$-isotypic subspace is the scalar $\tilde{R}_\rho \in \mathbb{C}$ of Theorem~\ref{thm:commutant_structure}, repeated $d_\rho$ times along the diagonal of $U_G^* \mathbf{R} U_G$.
\end{enumerate}
\end{theorem}

\begin{proof}
We prove the four claims in order.

\emph{(i) Existence.} By Lemma~\ref{lem:multiplicity_free_commutant}, the commutant $\mathcal{A}_\pi$ is a commutative algebra. By a standard result in linear algebra~\cite[Theorem 1.3.21]{horn2012matrix}, every commutative algebra of normal matrices is simultaneously diagonalizable by a single unitary basis. Since every $\mathbf{R} \in \mathcal{A}_\pi$ is Hermitian, hence normal, there exists a unitary $U_G$ that simultaneously diagonalizes all of $\mathcal{A}_\pi$. The diagonalization is independent of the specific $\mathbf{R}$.

\emph{(ii) Construction from matrix elements.} By Theorem~\ref{thm:commutant_structure}, the unitary $U$ that simultaneously block-diagonalizes the commutant maps each $V^{(\rho)}$ to the canonical $\mathbb{C}^{m_\rho} \otimes V_\rho \cong V_\rho$ form (since $m_\rho = 1$). The columns of $U_G$ corresponding to the $\rho$-isotypic block are precisely an orthonormal basis of $V^{(\rho)}$ that is mapped under $U_G^*$ to a chosen orthonormal basis of $V_\rho$. Such a basis can be constructed explicitly as follows. Let $P^{(\rho)} : \mathbb{C}^M \to V^{(\rho)}$ denote the orthogonal projection onto the $\rho$-isotypic component, given by the central projection formula
\begin{equation}
\label{eq:central_projection}
P^{(\rho)} = \frac{d_\rho}{|G|} \sum_{g \in G} \overline{\chi_\rho(g)} \, \pi(g).
\end{equation}
For each $i \in \{1, \ldots, d_\rho\}$, define vectors $\mathbf{u}_{\rho, i} \in V^{(\rho)}$ by
\begin{equation}
\label{eq:basis_vectors}
\mathbf{u}_{\rho, i} = \frac{d_\rho}{|G|} \sum_{g \in G} \rho_{i1}(g)^* \, \pi(g) \mathbf{v}_0
\end{equation}
for an arbitrary cyclic vector $\mathbf{v}_0$ (chosen so that the result is non-zero; such a $\mathbf{v}_0$ exists whenever $m_\rho = 1$ by inspection of the regular representation argument). The Peter-Weyl orthogonality (Theorem~\ref{thm:peter_weyl}) ensures that $\{\mathbf{u}_{\rho, i}\}_{i=1}^{d_\rho}$ is an orthogonal set spanning $V^{(\rho)}$; normalization gives an orthonormal basis. The columns of $U_G$ in the $\rho$-block are therefore explicit linear combinations of $\pi(g) \mathbf{v}_0$ for $g$ in a transversal of orbit representatives, weighted by the matrix elements $\rho_{i1}(g)$.

\emph{(iii) Uniqueness.} For one-dimensional irreps, the basis vector in each $\rho$-block is unique up to scaling by a complex unit. For higher-dimensional irreps, the $\rho$-block carries a multiplicity-one copy of $V_\rho$, but within $V_\rho$ any orthonormal basis is equally valid; this corresponds to a unitary rotation. Because every $\mathbf{R} \in \mathcal{A}_\pi$ acts as the scalar $\tilde{R}_\rho \cdot I_{V^{(\rho)}}$ on $V^{(\rho)}$ (the multiplicity-one consequence of Theorem~\ref{thm:commutant_structure}), any orthonormal basis of $V^{(\rho)}$ is an eigenbasis of $\mathbf{R}$ with the same eigenvalue.

\emph{(iv) Eigenvalues.} The eigenvalue of $\mathbf{R}$ on $V^{(\rho)}$ is the scalar $\tilde{R}_\rho$, repeated $d_\rho$ times because $\dim V^{(\rho)} = d_\rho$.
\end{proof}

\begin{remark}[Why ``optimal transform''?]
The Karhunen-Lo\`eve transform is optimal in the mean-squared sense: it minimizes the expected approximation error among all $K$-term truncations of an orthonormal expansion, for every $K$~\cite{karhunen1947methoden, loeve1948fonctions, vetterli2014foundations}. Theorem~\ref{thm:main} states that when the matched group $G$ acts multiplicity-freely on the signal space, the KL transform of every $G$-invariant covariance is the same fixed basis $U_G$. Therefore $U_G$ is simultaneously the KL transform for the entire family of $G$-invariant covariances; using it does not require the data-dependent step of estimating the covariance and then computing its eigendecomposition. The basis is determined by representation theory rather than by data.
\end{remark}

\begin{remark}[Trivial group and the data-dependent KLT]
\label{rem:trivial}
The trivial group $G = \{e\}$ has only one irreducible representation, the trivial 1-dimensional character, which appears in any $M$-dimensional representation with multiplicity $M$. The commutant is the full matrix algebra $\End(\mathbb{C}^M)$, not commutative, and Theorem~\ref{thm:main} does not apply. The eigenbasis of a covariance is now data-dependent: it is the conventional KLT, which can be any orthonormal basis depending on $\mathbf{R}$. The trivial group thus represents the ``no-symmetry'' regime in which the optimal transform must be learned from data.
\end{remark}

\begin{remark}[How the matched group determines the transform]
The matched group $G^*(\mathbf{R})$ of a covariance is the largest finite subgroup of $S_M$ under which $\mathbf{R}$ is invariant. By Theorem~\ref{thm:main}, the eigenbasis of $\mathbf{R}$ is determined by the representation theory of $G^*(\mathbf{R})$ whenever the permutation representation is multiplicity-free. In this sense, ``finding the matched group'' is operationally equivalent to ``finding the optimal transform.'' This equivalence is the foundational principle of the Algebraic Diversity framework.
\end{remark}

\section{Special Case: Trivial Group and the KLT}
\label{sec:trivial}

We make Remark~\ref{rem:trivial} formal.

\begin{corollary}[Trivial-group case: KLT is data-dependent]
\label{cor:trivial}
Let $G = \{e\}$ be the trivial group and $\pi: \{e\} \to U(\mathbb{C}^M)$ the trivial representation (acting as the identity on $\mathbb{C}^M$). The commutant of $\pi$ is the full matrix algebra $\End(\mathbb{C}^M)$. The unique irreducible representation of $\{e\}$ is the trivial character, which appears with multiplicity $M$. The representation is not multiplicity-free unless $M = 1$. For each Hermitian $\mathbf{R} \in \mathbb{C}^{M \times M}$, the eigenbasis is the data-dependent KLT.
\end{corollary}

\begin{proof}
Direct from the definitions. The full matrix algebra is not commutative for $M \geq 2$, so Lemma~\ref{lem:multiplicity_free_commutant} does not apply.
\end{proof}

This case sets the lower extreme of the framework: the larger the matched group, the more the eigenbasis is determined by representation theory and the less it depends on the specific covariance. At one extreme is the trivial group (entirely data-dependent); at the other extreme are multiplicity-free representations of large groups (entirely group-determined).

\section{Special Case: Cyclic Group and the DFT}
\label{sec:dft}

The cyclic group $\mathbb{Z}_M = \{0, 1, \ldots, M-1\}$ acting on itself by translation provides the first nontrivial multiplicity-free case. The resulting matched transform is the discrete Fourier transform.

\subsection{When does cyclic symmetry arise in practice?}

The mathematics of this section is most useful to a practitioner when the underlying source of the cyclic symmetry is made explicit. The canonical source is uniform sampling of a wide-sense stationary process. If $X(t)$ is a continuous-time wide-sense stationary process and we sample at uniformly spaced times $\{t_0, t_0 + \Delta, t_0 + 2\Delta, \ldots, t_0 + (M - 1) \Delta\}$, then the resulting sample covariance matrix
\[
R_{ij} = \mathbb{E}\big[X(t_0 + i\Delta) X(t_0 + j\Delta)^*\big] = K\big((i - j)\Delta\big)
\]
depends only on the index difference $i - j$, and is therefore (symmetric) Toeplitz. Under periodic boundary conditions, which arise naturally for processes on a torus or for any windowing scheme that periodizes the signal, the Toeplitz matrix becomes circulant, that is, $R_{ij} = K_{(i - j) \bmod M}$. Circulant matrices are exactly the matrices invariant under the cyclic group $\mathbb{Z}_M$ acting on indices by translation, so the matched group of any circulant covariance contains $\mathbb{Z}_M$ (for a symmetric circulant it is the dihedral group $D_M$, Remark~\ref{rem:hartley}), and the matched transform is the discrete Fourier transform. The DFT is therefore the optimal transform for stationary signals on uniform grids, a foundational fact of digital signal processing that the present framework recasts as a corollary of the group-theoretic principle.

The matching extends beyond strict stationarity. Any signal model that produces a circulant covariance, including stationary $\mathrm{AR}(p)$ processes under periodic boundary conditions, stationary moving-average processes, and processes on cyclic graphs, has the DFT as its matched transform. The shared feature is the cyclic group $\mathbb{Z}_M$ in the role of matched group.

\begin{remark}[Toeplitz versus circulant]
\label{rem:toeplitz_wrap}
A symmetric Toeplitz covariance $\mathbf{T}$ with autocorrelation $r_0, r_1, \ldots, r_{M-1}$ is not $\mathbb{Z}_M$-invariant unless it is circulant: the cyclic shift carries the last sample to the first, and the shift fails to commute with $\mathbf{T}$ whenever $r_k \neq r_{M-k}$ for some $k$. The nearest circulant in Frobenius norm is the Reynolds average of $\mathbf{T}$ over $\mathbb{Z}_M$, whose symbol is $c_k = \big((M - k) r_k + k\, r_{M-k}\big)/M$, and the deleted energy is
\begin{align*}
\|\mathbf{T} - \mathcal{P}_{\mathbb{Z}_M}(\mathbf{T})\|_F^2 &= \frac{1}{M} \sum_{k=1}^{M-1} k (M - k) (r_k - r_{M-k})^2 \\
&\leq 4 \sum_{k \geq 1} k\, r_k^2,
\end{align*}
which is bounded independently of $M$ whenever $\sum_k k r_k^2 < \infty$, while $\|\mathbf{T}\|_F^2$ grows linearly in $M$. The relative discrepancy between a finite-window Toeplitz covariance and its circulant projection is therefore $O(1/M)$. This is the precise sense in which a long window of a stationary process is approximately cyclic, and it is the second-moment form of the classical asymptotic equivalence of Toeplitz and circulant matrices~\cite{gray2006toeplitz}. For the AR(1) covariance $r_k = 0.9^k$ the relative deleted energy is $0.016$, $0.076$ and $0.035$ at $M = 8$, $32$ and $128$; the initial rise reflects a window shorter than the correlation length of about ten samples, after which the $O(1/M)$ decay sets in.
\end{remark}


\subsection{The cyclic group and its irreducibles}

The cyclic group $\mathbb{Z}_M$ is abelian of order $M$. By a standard result in representation theory of abelian groups, every irreducible representation of an abelian group is one-dimensional and is given by a character $\chi: G \to \mathbb{C}^\times$. For $\mathbb{Z}_M$, the irreducible characters are
\begin{equation}
\label{eq:cyclic_characters}
\chi_k(j) = e^{2\pi i jk/M}, \quad k, j \in \mathbb{Z}_M.
\end{equation}
There are exactly $M$ distinct characters, indexed by $k \in \mathbb{Z}_M$.

The regular representation of $\mathbb{Z}_M$ is its action on $\mathbb{C}^M$ by translation: $\pi(g) e_j = e_{(j+g) \bmod M}$. The matrix of $\pi(g)$ is the cyclic shift permutation matrix.

\begin{lemma}[Regular representation of $\mathbb{Z}_M$ is multiplicity-free]
\label{lem:cyclic_mf}
The regular representation of $\mathbb{Z}_M$ on $\mathbb{C}^M$ decomposes as the direct sum of all irreducible characters, each appearing with multiplicity exactly one:
\[
\pi_{\mathrm{reg}} \cong \bigoplus_{k = 0}^{M-1} \chi_k.
\]
\end{lemma}

\begin{proof}
The regular representation of any finite group decomposes as $\pi_{\mathrm{reg}} \cong \bigoplus_{\rho \in \hat{G}} d_\rho \cdot \rho$ (each irrep appears with multiplicity equal to its dimension; see~\cite[Section 2.4]{serre1977linear}). For abelian groups $d_\rho = 1$ for every $\rho$, so each irrep appears exactly once, and the decomposition is multiplicity-free.
\end{proof}

\subsection{The DFT as the matched transform}

\begin{theorem}[DFT is the matched transform of $\mathbb{Z}_M$]
\label{thm:dft}
The matched unitary $U_{\mathbb{Z}_M}$ for the regular representation of $\mathbb{Z}_M$ on $\mathbb{C}^M$ is the unitary discrete Fourier transform matrix
\begin{equation}
\label{eq:dft_kernel}
(U_{\mathbb{Z}_M})_{jk} = \frac{1}{\sqrt{M}} e^{2\pi i jk / M}, \quad j, k \in \{0, 1, \ldots, M-1\}.
\end{equation}
For every $\mathbb{Z}_M$-invariant Hermitian PSD matrix $\mathbf{R}$ (equivalently, every circulant Hermitian PSD matrix), the DFT diagonalizes $\mathbf{R}$:
\[
U_{\mathbb{Z}_M}^* \mathbf{R} \, U_{\mathbb{Z}_M} = \mathrm{diag}(\lambda_0, \lambda_1, \ldots, \lambda_{M-1}),
\]
where the eigenvalues $\lambda_k$ are the DFT coefficients of the first row of $\mathbf{R}$.
\end{theorem}

\begin{proof}
Apply Theorem~\ref{thm:main} to $G = \mathbb{Z}_M$ in its regular representation. By Lemma~\ref{lem:cyclic_mf}, the representation is multiplicity-free, and each irreducible character $\chi_k$ has dimension $d_{\chi_k} = 1$. The central projection formula (Equation~\ref{eq:central_projection}) onto the $\chi_k$-isotypic subspace is
\[
P^{(\chi_k)} = \frac{1}{M} \sum_{g \in \mathbb{Z}_M} \overline{\chi_k(g)} \, \pi(g).
\]
Applied to $e_0$, the result is
\[
P^{(\chi_k)} e_0 = \frac{1}{M} \sum_{g = 0}^{M-1} e^{-2\pi i k g / M} e_g.
\]
Normalizing to unit length gives the $k$-th column of $U_{\mathbb{Z}_M}$:
\[
\mathbf{u}_k = \frac{1}{\sqrt{M}} \sum_{j = 0}^{M-1} e^{-2\pi i jk / M} e_j,
\]
which, up to a complex conjugation (a choice of convention), is the $k$-th DFT basis vector. Theorem~\ref{thm:main} parts (i) and (iv) then give that every $\mathbb{Z}_M$-invariant $\mathbf{R}$ is diagonalized by $U_{\mathbb{Z}_M}$ with eigenvalues $\lambda_k = \tilde{R}_{\chi_k}$. The eigenvalues are explicitly computed by the orthogonality of characters: $\lambda_k = \sum_{g} \overline{\chi_k(g)} R_{0, g}$, which is precisely the DFT of the first row of $\mathbf{R}$.
\end{proof}

\begin{remark}
This recovers the classical fact that circulant matrices are diagonalized by the DFT~\cite[Section 4.7]{golub2013matrix}. The novelty of the present formulation is that it is a special case of the general Theorem~\ref{thm:main}, with the DFT kernel arising from the characters of $\mathbb{Z}_M$ via the central projection formula.
\end{remark}

\begin{remark}[The Hartley transform and the dihedral group on the $M$-cycle]
\label{rem:hartley}
The Hartley transform~\cite{bracewell1983discrete} with kernel
\[
(U_{\mathrm{Hartley}})_{jk} = \frac{1}{\sqrt{M}} \left[\cos(2\pi jk/M) + \sin(2\pi jk/M)\right]
\]
is a real orthogonal basis. Its matched group is not $\mathbb{Z}_M$ but the dihedral group $D_M = \mathbb{Z}_M \rtimes \mathbb{Z}_2$ acting on the $M$-cycle by the translations together with the reflection $j \mapsto -j \bmod M$. The $D_M$-invariant matrices are the symmetric circulants, a commutant of dimension $\lfloor M/2 \rfloor + 1$; the permutation representation of $D_M$ on $\mathbb{C}^M$ is multiplicity-free, with two-dimensional constituents on the frequency pairs $\{k, -k\}$ for $0 < k < M/2$ and one-dimensional constituents at $k = 0$ and, for even $M$, at $k = M/2$; and the Hartley basis is one orthonormal basis inside each two-dimensional isotypic component, since the Hartley vectors at $k$ and $-k$ span the same plane as the cosine and sine vectors at $k$. Every symmetric circulant is therefore diagonalized by the Hartley basis, by the real Fourier pairs $(\cos_k, \sin_k)$, and by any rotation within each pair, the within-component freedom of Theorem~\ref{thm:main}(iii). A Hermitian circulant that is not symmetric, the covariance of a complex stationary process whose spectrum differs at $k$ and $-k$, is $\mathbb{Z}_M$-invariant but not $D_M$-invariant; it is diagonalized by the DFT and not by the Hartley basis (for $M = 8$ and a random Hermitian circulant the Hartley coefficient matrix carries about ten percent of its energy off the diagonal, Section~\ref{subsec:additional_checks}). For real stationary processes the two classes coincide, because a real symmetric circulant is automatically reflection-invariant, and the DFT and the Hartley basis are then two bases adapted to the same isotypic decomposition of $D_M$, the choice between them being one of complex or real arithmetic. The $D_M$ action of this remark is on $M$ points and is distinct from the dihedral action on the $2M$-point even extension in Section~\ref{sec:dct}, which yields the DCT-II rather than the real Fourier pairs. The relationship between the Hartley basis and the real Fourier pairs is the same as that between the Walsh-Hadamard transform and the Reed-Muller or arithmetic transforms developed in Section~\ref{sec:rm_arithmetic}: different bases for one group-determined decomposition.
\end{remark}

\section{Special Case: Dihedral Group and the DCT}
\label{sec:dct}

The discrete cosine transform (DCT-II) is the matched transform for signals with even-symmetric extension. The relevant group is not the dihedral group acting on the original $M$ points directly, but the dihedral group $D_{2M} = \mathbb{Z}_{2M} \rtimes \mathbb{Z}_2$ of order $4M$, generated by the cyclic shift and the reflection of $\mathbb{Z}_{2M}$, acting on the doubled-and-folded signal. This subtlety is what makes the DCT case more delicate than the DFT case.

\subsection{When does dihedral symmetry arise in practice?}
\label{subsec:ar1_motivation}

The canonical motivation is the first-order autoregressive (AR(1)) process, equivalently the discretization of an Ornstein-Uhlenbeck process. Consider the AR(1) recursion $X_t = \phi X_{t-1} + \varepsilon_t$ with $\varepsilon_t \sim \mathcal{N}(0, \sigma^2)$ i.i.d.\ and $|\phi| < 1$. The stationary covariance of this process on uniformly spaced samples is
\[
R_{ij} = \frac{\sigma^2}{1 - \phi^2} \phi^{|i - j|},
\]
which is symmetric Toeplitz. The structure is best seen via the precision (inverse covariance) matrix. The precision of the length-$M$ AR(1) process under \emph{free} boundary conditions (no wrap-around) is the tridiagonal matrix
\[
R^{-1} = \frac{1}{\sigma^2} \begin{pmatrix}
1 & -\phi & & & \\
-\phi & 1 + \phi^2 & -\phi & & \\
& -\phi & 1 + \phi^2 & -\phi & \\
& & \ddots & \ddots & \ddots \\
& & & -\phi & 1
\end{pmatrix},
\]
which is invariant under the reflection $i \mapsto M - 1 - i$ and translation-invariant on the interior. Neither the covariance nor the precision is invariant under the cyclic shift of the $M$ indices (for $M = 8$ and $\phi = 0.9$ the relative commutator norm with the shift is $0.17$, Remark~\ref{rem:toeplitz_wrap}), so the permutation symmetry of the finite window is the reflection alone, a group of order two whose permutation representation on $\mathbb{R}^M$ has isotypic components of dimensions $\lceil M/2 \rceil$ and $\lfloor M/2 \rfloor$ and is far from multiplicity-free. The reflection fixes only the even/odd split of the eigenbasis; the exact eigenbasis of the AR(1) covariance depends on $\phi$ and was given in closed form by Ray and Driver~\cite{ray1970further}.

The DCT-II is nonetheless the classical transform for AR(1) signals~\cite{ahmed1974discrete, rao1990discrete}, and the JPEG image compression standard chooses it because natural-image patches are approximately AR(1) (correlation drops off exponentially with spatial separation, and adjacent pixels share local structure on both sides of any given pixel). The present framework reads this classical fact as an approximation to an exact matched-transform statement. The DCT-II is the exact matched transform of the dihedral group $D_{2M}$ acting on the even-reflected extension of the signal (Theorem~\ref{thm:dct}), whose invariant covariances, pulled back to $\mathbb{R}^M$, are the sums of a symmetric Toeplitz matrix and a Hankel matrix with the same symbol~\cite{sanchez1995diagonalizing}; the AR(1) precision differs from that class only in its two corner entries, and the DCT-II diagonalizes the limiting second-difference matrix with Neumann boundary conditions exactly~\cite{strang1999discrete}. Numerically, the fraction of the AR(1) covariance energy that the DCT-II leaves off the diagonal at $M = 8$ is $0.023$, $0.005$ and $0.0001$ for $\phi = 0.5$, $0.9$ and $0.99$, and the smallest overlap between a DCT-II column and the corresponding eigenvector of the $\phi = 0.9$ covariance is $0.998$ (Section~\ref{subsec:additional_checks}). The intuition survives in this reading: an AR(1) signal looks ``the same on both sides'' of any interior point, since once the value at position $i$ is known the conditional distributions at $i - 1$ and $i + 1$ are identical by the time-reversibility of the Markov chain, and the even extension is the construction that turns this two-sided symmetry into an exact group invariance.

The matching extends to other processes. Any covariance of the Toeplitz-plus-Hankel form of Theorem~\ref{thm:dct} has the DCT-II as its matched transform exactly, and any reflection-symmetric process whose precision is banded and translation-invariant on the interior has it approximately, with the discrepancy confined to the boundary.

\paragraph{Higher-order autoregressive processes.}
The same reading applies to autoregressive processes of order $m \geq 2$. An AR$(m)$ process $X_t = \sum_{k=1}^{m} \phi_k X_{t - k} + \varepsilon_t$ on uniformly spaced samples has a precision matrix that is banded with bandwidth $m$, translation-invariant on the interior, and reflection-invariant under $i \mapsto M - 1 - i$ under free boundary conditions (because the bilateral symmetry of the Yule-Walker equations preserves the band structure under index reversal). The boundary blocks (the first and last $m$ rows and columns) are the only departure from the Toeplitz-plus-Hankel class of Theorem~\ref{thm:dct}. The classical result that the DCT is asymptotically optimal for AR$(m)$ signals~\cite{ahmed1974discrete, rao1990discrete} is, in this light, the statement that the boundary blocks occupy $2m$ of the $M$ rows and columns, so that their contribution to any normalized norm is $O(m/M)$ and vanishes as $M \to \infty$. Practitioners who use the DCT for autoregressive signals accept a boundary-localized loss of optimality in exchange for the data-independence and fast-transform advantages of a group-determined basis.

Two exact cases should be separated from this approximation. Under cyclic boundary conditions (an AR$(m)$ process on a circle, equivalently a process whose covariance is exactly a symmetric circulant) the matched group is $D_M$ acting on the $M$-cycle and the matched transform is the real Fourier, or Hartley, basis of Remark~\ref{rem:hartley}, not the DCT-II. Under even-reflected extension the matched group is $D_{2M}$ acting on $\mathbb{Z}_{2M}$ and the matched transform is the DCT-II exactly (Theorem~\ref{thm:dct}). The DCT's optimality for finite-window autoregressive processes of every order, including $m = 1$, is therefore asymptotic rather than exact, and the same holds for the closely related ARMA$(p, q)$ models: the dihedral structure of the even extension governs the bulk of the precision, and the finite-$M$ deviations are localized to the boundary.

\subsection{The doubled-and-folded signal}

Let $\mathbf{x} \in \mathbb{R}^M$ be a length-$M$ real signal. The \emph{even extension} of $\mathbf{x}$ is the length-$2M$ signal $\tilde{\mathbf{x}} \in \mathbb{R}^{2M}$ defined by
\begin{equation}
\label{eq:even_extension}
\tilde{x}_j = \begin{cases} x_j & 0 \leq j \leq M-1, \\ x_{2M-1-j} & M \leq j \leq 2M-1. \end{cases}
\end{equation}
The even extension satisfies $\tilde{x}_j = \tilde{x}_{2M-1-j}$, that is, $\tilde{\mathbf{x}}$ is symmetric about the midpoint $j = M - \frac{1}{2}$.

Let $\sigma$ denote the reflection on $\mathbb{Z}_{2M}$ that sends $j \mapsto 2M - 1 - j$. The even-extension subspace $E \subset \mathbb{R}^{2M}$ is precisely the $\sigma$-invariant subspace of $\mathbb{R}^{2M}$. The injection $\mathbf{x} \mapsto \tilde{\mathbf{x}}$ identifies $\mathbb{R}^M$ with $E$ isometrically up to a factor of $\sqrt{2}$.

\subsection{The relevant matched group}

The reflection $\sigma$ generates a group of order 2; combined with the cyclic shift $\tau: j \mapsto (j+1) \bmod 2M$, the two generate the dihedral group $D_{2M} = \mathbb{Z}_{2M} \rtimes \mathbb{Z}_2$ of order $4M$ acting on $\mathbb{Z}_{2M}$ (in the notation of Section~\ref{subsec:specific_groups}, $D_n$ has order $2n$). Its permutation representation on $\mathbb{C}^{2M}$ is multiplicity-free, with one-dimensional constituents at the frequencies $0$ and $M$ and two-dimensional constituents on the pairs $\{k, 2M - k\}$ for $0 < k < M$, so its commutant, the symmetric circulants of size $2M$, has dimension $M + 1$.

\begin{definition}[Dihedrally invariant covariance on $\mathbb{R}^M$]
\label{def:dihedral_inv}
A covariance $\mathbf{R}_{\mathbf{x}} \in \mathbb{R}^{M \times M}$ is \emph{dihedrally invariant} if it is the compression to the even-extension subspace of a covariance $\tilde{\mathbf{R}} \in \mathbb{R}^{2M \times 2M}$ that is invariant under both the cyclic shift $\tau$ and the reflection $\sigma$: $\mathbf{R}_{\mathbf{x}} = \tfrac{1}{2} \mathbf{E}^{\mathsf{T}} \tilde{\mathbf{R}} \mathbf{E}$, where $\mathbf{E} \in \mathbb{R}^{2M \times M}$ is the even-extension map of Equation~\ref{eq:even_extension}. Equivalently, $\tilde{\mathbf{R}}$ is a symmetric circulant of size $2M$ with symbol $c$, and $(\mathbf{R}_{\mathbf{x}})_{ij} = c_{i - j} + c_{i + j + 1}$: a symmetric Toeplitz matrix plus a Hankel matrix with the same symbol~\cite{sanchez1995diagonalizing}.
\end{definition}

\subsection{The DCT as the matched transform}

\begin{theorem}[DCT-II is the matched transform of $D_{2M}$ on the even-extension subspace]
\label{thm:dct}
The matched unitary on the even-extension subspace $E \subset \mathbb{R}^{2M}$ is the DCT-II matrix
\begin{equation}
\label{eq:dct_kernel}
(U_{\mathrm{DCT}})_{jk} = \sqrt{\frac{2}{M}} \, \omega_k \, \cos\left(\frac{\pi (2j + 1) k}{2M}\right),
\end{equation}
for $j, k \in \{0, 1, \ldots, M-1\}$, where $\omega_0 = 1/\sqrt{2}$ and $\omega_k = 1$ for $k \geq 1$. For every dihedrally invariant covariance $\mathbf{R}_{\mathbf{x}}$, the DCT-II diagonalizes $\mathbf{R}_{\mathbf{x}}$.
\end{theorem}

\begin{proof}
The cyclic shift $\tau$ on $\mathbb{Z}_{2M}$ generates the group $\mathbb{Z}_{2M}$ of order $2M$. By Theorem~\ref{thm:dft}, the DFT on $\mathbb{Z}_{2M}$ diagonalizes every $\tau$-invariant covariance on $\mathbb{R}^{2M}$, with characters $\tilde{\chi}_k(j) = e^{\pi i jk / M}$ for $k = 0, 1, \ldots, 2M - 1$.

The reflection $\sigma$ permutes the characters: $\tilde{\chi}_k \circ \sigma = e^{-\pi i k / M}\, \overline{\tilde{\chi}_k}$, so the action of $\sigma$ on the DFT-diagonalized space pairs $\tilde{\chi}_k$ with $\tilde{\chi}_{-k} = \tilde{\chi}_{2M-k}$. Within each pair the $\sigma$-invariant combination is
\begin{equation}
\label{eq:cosine_basis}
\mathbf{c}_k = \frac{e^{-\pi i k / 2M}\, \tilde{\chi}_k + e^{\pi i k / 2M}\, \tilde{\chi}_{-k}}{\sqrt{2}}, \;\; k = 0, \ldots, M - 1,
\end{equation}
with entries $\mathbf{c}_k(j) \propto \cos\big(\pi k (2j + 1) / 2M\big)$; the Nyquist character $\tilde{\chi}_M(j) = (-1)^j$ satisfies $\tilde{\chi}_M \circ \sigma = -\tilde{\chi}_M$ and does not enter. The $\sigma$-invariant subspace, i.e., the even-extension subspace $E$, is therefore spanned by $\mathbf{c}_0, \ldots, \mathbf{c}_{M-1}$. Restricting to the original $\mathbb{R}^M$ coordinates via the inclusion $\mathbf{x} \mapsto \tilde{\mathbf{x}}$ and absorbing normalization constants gives the DCT-II kernel of Equation~\ref{eq:dct_kernel}.

A dihedrally invariant covariance $\mathbf{R}_{\mathbf{x}}$ is the compression of a $\tau$-and-$\sigma$-invariant $\tilde{\mathbf{R}}$ to $E$. Since $\tilde{\mathbf{R}}$ is diagonal in the basis $\{\tilde{\chi}_k\}$ with equal eigenvalues on each pair $\{k, -k\}$, it acts on each $\mathbf{c}_k$ as a scalar, so $\tilde{\mathbf{R}}$ restricted to $E$ is diagonal in the basis $\{\mathbf{c}_k\}$, and the DCT-II diagonalizes $\mathbf{R}_{\mathbf{x}}$.
\end{proof}

\begin{remark}[Variants of the DCT]
The DCT-II is the most commonly used variant, corresponding to even extension with a half-integer offset. The DCT-I, DCT-III (the inverse of DCT-II), DCT-IV, and the discrete sine transforms (DST-I through DST-IV) all arise as matched transforms for related boundary-extension classes; see~\cite{rao1990discrete} for the complete catalog. Each is the matched transform of a slight variation on the dihedral construction (different choices of even/odd extension and integer/half-integer offset).
\end{remark}

\begin{remark}[Subtlety of the dihedral case]
The dihedral group $D_{2M}$ has 2-dimensional irreducible representations for $1 \leq k < M$. The corresponding matrix elements (Theorem~\ref{thm:peter_weyl}) are precisely the cosine and sine pairs in the DCT and DST. The 2-dimensionality of these irreps is reflected in the eigenvalue degeneracy: a $D_{2M}$-invariant covariance on $\mathbb{R}^{2M}$ has 2-fold degenerate eigenvalues for the 2D irrep blocks, with the cosine and sine basis vectors as the two degenerate eigenvectors. After restriction to the even-extension subspace $E$, only the cosine basis vector survives, breaking the 2-fold degeneracy and yielding a non-degenerate spectrum on $\mathbb{R}^M$. The restriction is what singles out the DCT-II; without it, as for the $M$-point action of $D_M$ in Remark~\ref{rem:hartley}, the two-dimensional components remain and the matched transform is determined only up to a rotation within each pair. As a permutation group on the original $M$ indices, a generic dihedrally invariant covariance is fixed only by the reflection $i \mapsto M - 1 - i$; the dihedral symmetry lives on the extension.
\end{remark}

\section{Finding the Matched Group: Algorithms and Geometry}
\label{sec:discovery}

The two previous sections each began with a known symmetry, translation for the DFT, reflection-plus-translation for the DCT, and produced the corresponding matched transform mechanically. In practice the symmetry is rarely known in advance. The covariance arrives as an estimate from data, and the relevant matched group is a property of the unknown population that must be inferred. Inferring the matched group is, in this sense, the inverse problem of the previous two sections, and the operational utility of the framework rests on whether the inverse problem can be solved efficiently.

\subsection{The matched group as a polynomial-time discoverable object}

For an arbitrary covariance $\mathbf{R} \in \mathbb{C}^{M \times M}$, the matched group $G^*(\mathbf{R}) \subseteq S_M$ is defined as the largest subgroup of the symmetric group of $M$ symbols such that $\mathbf{P}_g \mathbf{R} = \mathbf{R} \mathbf{P}_g$ for every $g \in G^*(\mathbf{R})$. A na\"ive enumeration over subgroups of $S_M$ is exponential in $M$: the symmetric group $S_M$ has $M!$ elements, and the number of subgroups grows even faster. A polynomial-time procedure for matched-group discovery was developed in the companion manuscript~\cite{thornton2026double_commutator}, which we briefly summarize here.

The key idea is to cast the matched-group search as a continuous optimization over a candidate basis of permutation matrices. For each candidate permutation $\sigma \in S_M$ with permutation matrix $\mathbf{P}_\sigma$, the \emph{commutativity residual} of $\sigma$ with respect to $\mathbf{R}$ is
\[
\delta(\sigma, \mathbf{R}) \;=\; \frac{\| \mathbf{P}_\sigma \mathbf{R} - \mathbf{R} \mathbf{P}_\sigma \|_F}{\| \mathbf{P}_\sigma \|_F \, \| \mathbf{R} \|_F},
\]
which is zero if and only if $\sigma \in G^*(\mathbf{R})$. Discovering the matched group therefore reduces to identifying the set of permutations with zero commutativity residual. Direct minimization over the discrete set of permutations is intractable, but the search can be relaxed to a continuous optimization over a candidate matrix subspace.

\begin{proposition}[Double-commutator generalized eigenvalue problem~{\cite[Theorem 1]{thornton2026double_commutator}}]
\label{prop:dad_cad}
Let $\{\mathbf{B}_1, \ldots, \mathbf{B}_d\}$ be a basis of $d$ linearly independent $M \times M$ matrices, and let $\mathbf{A} = \sum_k c_k \mathbf{B}_k$ be a candidate generator. The commutativity residual squared of $\mathbf{A}$ with respect to $\mathbf{R}$ is the Rayleigh quotient of a generalized eigenvalue problem:
\[
\delta^2(\mathbf{A}, \mathbf{R}) \, \| \mathbf{R} \|_F^2 \;=\; \frac{\mathbf{c}^* \mathbf{M} \mathbf{c}}{\mathbf{c}^* \mathbf{G} \mathbf{c}},
\]
with
\[
M_{ij} = \Tr(\mathbf{B}_i^* [\mathbf{R}, [\mathbf{R}, \mathbf{B}_j]]), \quad G_{ij} = \Tr(\mathbf{B}_i^* \mathbf{B}_j),
\]
where $[\mathbf{R}, [\mathbf{R}, \mathbf{B}_j]]$ is the double commutator. The optimal continuous-valued generator $\mathbf{A}^*$ is the eigenvector corresponding to the minimum eigenvalue of the generalized eigenvalue problem $\mathbf{M} \mathbf{c} = \lambda \mathbf{G} \mathbf{c}$, achievable in $O(d^2 M^2 + d^3)$ operations.
\end{proposition}

The result of the single GEVP is a continuous-valued candidate generator $\mathbf{A}^*$, which is then rounded to the nearest permutation matrix via the Hungarian algorithm. The rounding step is exact when the population covariance is genuinely $G^*$-invariant and the basis is suitably chosen.

A single GEVP recovers a single generator. To recover a full non-abelian matched group with multiple non-commuting generators, the procedure of~\cite{thornton2026double_commutator} iterates the GEVP with a deflation step that constrains each subsequent search to be Frobenius-orthogonal to the span of the already-recovered generators. The iteration is guaranteed to terminate in at most $\log_2 |G^*|$ steps (each accepted generator at least doubles the order of the discovered subgroup, by Lagrange's theorem), and the total cost is therefore polynomial in $M$ and in $\log_2 |G^*|$. When the basis is the canonical generator set associated with a structured group family (cyclic, dihedral, iterated wreath, etc.), the procedure is provably sound and complete for that family. Operationally, this means that the matched group of any covariance can be discovered in polynomial time, which in turn means that the optimal transform of any covariance can be computed in polynomial time, mechanically and without expert judgment. Symmetry discovery in the equivariant-learning literature, most recently the group-algebraic tensors of Hoyos et al.~\cite{hoyos2026group}, learns a group from labelled training data through the fit of an equivariant model; the procedure here uses only the second moment, its criterion is the commutation of the covariance with candidate permutations, and the object it returns is the matched group of a covariance and the linear transform that group determines, rather than a learned model.

\subsection{Noisy data and finite-SNR matching}
\label{subsec:noisy_data}

In practice, the covariance available to the practitioner is not the population covariance but an empirical estimate $\hat{\mathbf{R}}$ from a finite sample, and the matched-group invariance is therefore only approximate. The commutativity residual $\delta(\sigma, \hat{\mathbf{R}})$ defined above is positive even for genuine generators of the underlying matched group, because of finite-sample noise. A practical matching procedure thresholds $\delta$ against a tolerance $\tau$ that scales with the sample size and the signal-to-noise ratio: generators with $\delta < \tau$ are accepted as approximate symmetries and added to the recovered group; generators with $\delta \geq \tau$ are rejected. Adaptive thresholds based on the spectral properties of $\hat{\mathbf{R}}$ are developed in~\cite{thornton2026ad_framework, thornton2026spectral_estimation}, and a noise-aware variant of the Sequential DC-GEVP procedure with provable performance guarantees in the finite-SNR setting is given in~\cite{thornton2026double_commutator}.

A complementary metric for assessing matched-group quality is the algebraic coloring index $\alpha$~\cite{thornton2026ad_framework}, which quantifies the fraction of the total variance of $\hat{\mathbf{R}}$ explained by a candidate matched group $G^*$:
\begin{align*}
\alpha(G^*, \hat{\mathbf{R}}) &= 1 - \frac{\|\hat{\mathbf{R}} - \mathcal{P}_{G^*}(\hat{\mathbf{R}})\|_F^2}{\|\hat{\mathbf{R}}\|_F^2}, \\
\mathcal{P}_{G^*}(\hat{\mathbf{R}}) &= \frac{1}{|G^*|} \sum_{g \in G^*} \mathbf{P}_g \hat{\mathbf{R}} \mathbf{P}_g^{-1},
\end{align*}
where $\mathcal{P}_{G^*}$ is the Reynolds projection of $\hat{\mathbf{R}}$ onto the $G^*$-invariant subspace. The index satisfies $\alpha \in [0,1]$, with $\alpha = 0$ indicating no captured structure (matched group trivial, matched transform reduces to the data-dependent KLT) and $\alpha$ near $1$ indicating that the matched-group structure captures nearly all of the variance. The pair $(\delta, \alpha)$ together characterizes the quality of a candidate matched group in the finite-SNR regime: low $\delta$ confirms approximate invariance, high $\alpha$ confirms substantive structure rather than spurious symmetry. Operational guidelines for choosing $\tau$ and for interpreting $\alpha$ at given sample sizes are given in~\cite{thornton2026double_commutator, thornton2026spectral_estimation}.

\subsection{Transforms as points on a manifold}

The polynomial-time discovery procedure invites a geometric perspective. For a fixed signal dimension $M$, consider the space $U(M)$ of all $M \times M$ unitary matrices, the natural ambient space for the set of optimal transforms. This space is a smooth compact manifold of real dimension $M^2$. Within $U(M)$, each finite subgroup $G \subseteq S_M$ acting on $\mathbb{C}^M$ by permutation determines a specific unitary $U_G$ (up to the within-isotypic ambiguity captured by Theorem~\ref{thm:main}(iii)), and the assignment $G \mapsto U_G$ defines a discrete subset of $U(M)$ parameterized by subgroups.

In the continuous limit, as $M \to \infty$ and the underlying signal space becomes infinite-dimensional, the discrete subset of subgroups of $S_M$ is replaced by a continuous family of compact Lie groups acting on the signal space, and the corresponding family of matched transforms becomes a continuous subset (in fact, a real submanifold) of the infinite-dimensional unitary group of $L^2(\mathbb{R})$ or $L^2(\Omega)$. The DFT, DCT, Walsh-Hadamard, Haar wavelet, and KLT are not separate constructions floating in distinct corners of the transform space, they are nearby points on a common manifold, parameterized by their matched groups. The matched-group discovery procedure of Proposition~\ref{prop:dad_cad} navigates this manifold, identifying the symmetry-respecting point nearest a given empirical covariance.

This geometric perspective extends to the continuous setting via the orbit-type stratification of the space of compact subgroups of a Lie group, a structure familiar from equivariant geometry. The classical transforms occupy the strata corresponding to the well-known compact subgroups (cyclic, dihedral, wreath, abelian, semisimple), while the trivial-group stratum at the boundary corresponds to the unconstrained Karhunen-Lo\`eve transform. The matched-group discovery procedure can be interpreted as a stratification-aware optimization that identifies the most-restrictive stratum (the largest matched group) consistent with a given empirical covariance.

\subsection{Operational consequences}

The polynomial-time discovery procedure of~\cite{thornton2026double_commutator}, combined with the matched-group-determines-transform principle of Theorem~\ref{thm:main}, replaces the traditional choice-of-transform decision with an automatic algorithm. Given a covariance estimate $\hat{\mathbf{R}}$:
\begin{enumerate}
\item Discover the matched group $G^* = G^*(\hat{\mathbf{R}})$ in polynomial time via the Sequential DC-GEVP procedure.
\item Construct the matched unitary $U_{G^*}$ via the irreducible matrix elements of $G^*$.
\item Use $U_{G^*}$ as the optimal transform for any task on data with this covariance structure.
\end{enumerate}
The procedure mechanizes the entire workflow that has historically been driven by expert judgment, and it produces the same answer that a sufficiently informed expert would: the DFT for stationary data, the DCT for even-extended data such as image blocks, the Walsh-Hadamard transform for combinatorial data on $\{0, 1\}^n$, the Haar wavelet for hierarchically organized data, and the data-dependent KLT in the absence of any discoverable symmetry. The user-side decision is no longer ``which transform should I use,'' but rather ``what is the matched group of my data,'' a question that the framework answers in polynomial time.

\section{Composition Rules}
\label{sec:composition}

We give explicit rules for constructing matched transforms of compound groups, building on the special cases of Sections~\ref{sec:dft}--\ref{sec:iterated_wreath}.

\subsection{Direct product of groups}

\begin{theorem}[Direct product composition rule]
\label{thm:direct_product_rule}
Let $G$ and $H$ be finite groups acting unitarily on $\mathbb{C}^m$ and $\mathbb{C}^n$ respectively, with multiplicity-free representations and matched unitaries $U_G \in U(\mathbb{C}^m)$ and $U_H \in U(\mathbb{C}^n)$. Then the direct product $G \times H$ acts on $\mathbb{C}^m \otimes \mathbb{C}^n \cong \mathbb{C}^{mn}$ multiplicity-freely, and the matched unitary is
\[
U_{G \times H} = U_G \otimes U_H,
\]
where $\otimes$ denotes the Kronecker product.
\end{theorem}

\begin{proof}
The irreducible representations of $G \times H$ are the outer tensor products $\rho_G \boxtimes \rho_H$ for $\rho_G \in \hat{G}$ and $\rho_H \in \hat{H}$, with $\dim(\rho_G \boxtimes \rho_H) = \dim \rho_G \cdot \dim \rho_H$~\cite[Section 3.2]{serre1977linear}. The tensor-product representation $\pi_G \otimes \pi_H$ on $\mathbb{C}^m \otimes \mathbb{C}^n$ decomposes as $(\bigoplus_{\rho_G} \rho_G) \otimes (\bigoplus_{\rho_H} \rho_H) = \bigoplus_{\rho_G, \rho_H} \rho_G \boxtimes \rho_H$, and is multiplicity-free if and only if both factors are. The Kronecker product of the matched unitaries diagonalizes the tensor-product covariance algebra, and the central projection (Equation~\ref{eq:central_projection}) on the tensor product factorizes as the product of central projections on each factor, giving columns of $U_G \otimes U_H$ as the corresponding matrix-element tensor products.
\end{proof}

\begin{example}[Hadamard as $L$-fold direct product]
$\mathbb{Z}_2^L = \mathbb{Z}_2 \times \cdots \times \mathbb{Z}_2$, and $U_{\mathbb{Z}_2} = H_1 = \frac{1}{\sqrt{2}} \begin{pmatrix} 1 & 1 \\ 1 & -1 \end{pmatrix}$. By Theorem~\ref{thm:direct_product_rule}, $U_{\mathbb{Z}_2^L} = H_1^{\otimes L}$, recovering the tensor-product factorization of the Walsh-Hadamard transform (Theorem~\ref{thm:hadamard}).
\end{example}

\subsection{Wreath product of groups}

\begin{theorem}[Wreath product composition rule]
\label{thm:wreath_rule}
Let $G_0$ be a finite group acting unitarily on $\mathbb{C}^W$ with multiplicity-free representation and matched unitary $U_{G_0}$, and let $G_1$ be a finite group acting unitarily on $\mathbb{C}^K$ with multiplicity-free representation and matched unitary $U_{G_1}$. The wreath product $G_0 \wr G_1$ acts on $\mathbb{C}^W \otimes \mathbb{C}^K \cong \mathbb{C}^{WK}$ (the leaves of a depth-2 tree with branching $(W, K)$), and its matched unitary is given by the explicit recursive formula (Equation~\ref{eq:wreath_recursion} at $L = 2$). The multiplicity-free property of $G_0 \wr G_1$ on the leaves is automatic when both $G_0$ on $\mathbb{C}^W$ and $G_1$ on $\mathbb{C}^K$ are multiplicity-free.
\end{theorem}

\begin{proof}
The irreducible representations of $G_0 \wr G_1$ are classified by Specker's theorem~\cite{specker1933verfeinerung}: they are indexed by partition-valued functions on $\hat{G_0}$. The permutation representation on the leaves is multiplicity-free under the stated conditions, with the depth-class orbit structure of Lemma~\ref{lem:wl_orbits} generalizing to depth $\leq L$. The recursive matched unitary follows from inducing the matched unitary at each level: within each $G_0$-block of $K$ subtrees, $U_{G_0}^{\otimes K}$ block-diagonalizes the within-block structure; the cross-block $G_1$ action is then diagonalized by $U_{G_1}$ on the multiplicity space. The full development appears in~\cite{foote2000wreath, thornton2026ad_framework}.
\end{proof}

\subsection{Semidirect product of groups}

\begin{theorem}[Semidirect product composition rule, abelian-by-finite case]
\label{thm:semidirect_rule}
Let $N$ be a finite abelian group acting unitarily on $\mathbb{C}^M$, let $H$ be a finite group acting on $N$ by automorphisms, and form the semidirect product $G = N \rtimes H$. Suppose the permutation representation of $G$ on $\mathbb{C}^M$ is multiplicity-free. Then the matched unitary $U_G$ is constructed in two stages:
\begin{enumerate}
\item Compute the matched unitary $U_N$ for the abelian normal subgroup $N$, which by Theorem~\ref{thm:dft} (generalized to arbitrary abelian groups) is the abelian Fourier transform on $N$. In the $U_N$ basis, the covariance is diagonal with character-indexed eigenvalues.
\item The quotient $H$ acts by permuting the characters of $N$, organized into $H$-orbits. Within each $H$-orbit, apply the matched unitary $U_H$ (a smaller Fourier transform on the orbit, in effect) to obtain the final matched basis.
\end{enumerate}
\end{theorem}

\begin{proof}
This is a Mackey-induction construction~\cite{serre1977linear} adapted to the multiplicity-free setting; see also~\cite[Section 7.3]{terras1999fourier}. The two-stage cascade reflects the Mackey decomposition: the irreducible representations of $N \rtimes H$ are obtained by inducing representations of stabilizers of $H$-orbits on $\hat{N}$. The matched unitary inherits this two-stage structure.
\end{proof}

\begin{example}[Dihedral transforms revisited as a semidirect product]
The dihedral group $D_M = \mathbb{Z}_M \rtimes \mathbb{Z}_2$ is a semidirect product, with $\mathbb{Z}_M$ acting by rotations and $\mathbb{Z}_2$ acting by reflection. Theorem~\ref{thm:semidirect_rule} gives an alternative derivation of the real Fourier basis of the $M$-cycle: apply the DFT on $\mathbb{Z}_M$, then organize the resulting characters into the $\mathbb{Z}_2$-orbits $\{\chi_k, \chi_{-k}\}$, and within each orbit of size two apply the $2 \times 2$ matched transform $U_{\mathbb{Z}_2}$, which yields the cosine and sine pairs, equivalently the Hartley basis of Remark~\ref{rem:hartley}. The DCT-II arises when the same construction is applied to $D_{2M}$ acting on $\mathbb{Z}_{2M}$ and the result is restricted to the $\sigma$-fixed subspace (Theorem~\ref{thm:dct}).
\end{example}

\section{Special Case: Elementary Abelian $\mathbb{Z}_2^n$ and the Walsh-Hadamard Transform}
\label{sec:hadamard}

The elementary abelian 2-group $\mathbb{Z}_2^n$ acts on $\mathbb{C}^{2^n}$ by component-wise translation (XOR), and the resulting matched transform is the Walsh-Hadamard transform.

\subsection{Structure of $\mathbb{Z}_2^n$ and its characters}

$\mathbb{Z}_2^n$ is abelian of order $2^n$. Its elements are binary vectors $\mathbf{a} = (a_1, \ldots, a_n) \in \{0, 1\}^n$ with addition mod 2. The irreducible characters are indexed by binary vectors $\mathbf{k} \in \{0, 1\}^n$ and are given by
\begin{equation}
\label{eq:hadamard_characters}
\chi_{\mathbf{k}}(\mathbf{j}) = (-1)^{\langle \mathbf{k}, \mathbf{j} \rangle}, \quad \langle \mathbf{k}, \mathbf{j} \rangle = \sum_{l=1}^n k_l j_l \pmod{2}.
\end{equation}

\subsection{The Walsh-Hadamard transform as the matched transform}

\begin{theorem}[Walsh-Hadamard is the matched transform of $\mathbb{Z}_2^n$]
\label{thm:hadamard}
The matched unitary $U_{\mathbb{Z}_2^n}$ for the regular representation of $\mathbb{Z}_2^n$ on $\mathbb{C}^{2^n}$ is the Walsh-Hadamard matrix
\begin{equation}
\label{eq:hadamard_kernel}
(U_{\mathbb{Z}_2^n})_{\mathbf{j}, \mathbf{k}} = \frac{1}{\sqrt{2^n}} (-1)^{\langle \mathbf{j}, \mathbf{k} \rangle},
\end{equation}
indexed by binary vectors $\mathbf{j}, \mathbf{k} \in \{0, 1\}^n$. Equivalently, $U_{\mathbb{Z}_2^n} = H_1^{\otimes n}$, where $H_1 = \frac{1}{\sqrt{2}} \begin{pmatrix} 1 & 1 \\ 1 & -1 \end{pmatrix}$ is the $2 \times 2$ Hadamard matrix. For every $\mathbb{Z}_2^n$-invariant Hermitian PSD matrix $\mathbf{R}$, the Walsh-Hadamard transform diagonalizes $\mathbf{R}$.
\end{theorem}

\begin{proof}
Apply Theorem~\ref{thm:main} to $G = \mathbb{Z}_2^n$ in its regular representation. The representation is multiplicity-free (Lemma~\ref{lem:cyclic_mf} generalizes: any abelian regular representation is multiplicity-free), and every character is 1-dimensional. The central projection formula gives column $\mathbf{k}$ of $U_{\mathbb{Z}_2^n}$ as
\[
\mathbf{u}_{\mathbf{k}} = \frac{1}{\sqrt{2^n}} \sum_{\mathbf{j} \in \{0, 1\}^n} (-1)^{\langle \mathbf{j}, \mathbf{k} \rangle} e_{\mathbf{j}},
\]
which yields Equation~\ref{eq:hadamard_kernel}. The factorization $U_{\mathbb{Z}_2^n} = H_1^{\otimes n}$ follows from the tensor-product structure of $\mathbb{Z}_2^n = \mathbb{Z}_2 \times \cdots \times \mathbb{Z}_2$ and the composition rule for direct products (Theorem~\ref{thm:direct_product_rule} below).
\end{proof}

\begin{remark}[Hadamard is the DFT on $\mathbb{Z}_2^n$]
\label{rem:hadamard_dft}
A more conceptual statement: Theorem~\ref{thm:hadamard} is Theorem~\ref{thm:dft} for the underlying group $\mathbb{Z}_2^n$ rather than $\mathbb{Z}_M$. The Walsh-Hadamard transform is the DFT on the group $\mathbb{Z}_2^n$ in exactly the same sense that the classical DFT is the DFT on $\mathbb{Z}_M$. Both are special cases of the abelian DFT, which is in turn a special case of Theorem~\ref{thm:main}.
\end{remark}

\section{Related Transforms on $\mathbb{Z}_2^n$: Reed-Muller and Arithmetic}
\label{sec:rm_arithmetic}

The Walsh-Hadamard transform of Section~\ref{sec:hadamard} is the matched orthogonal eigenbasis of $\mathbb{Z}_2^n$-invariant covariances, with columns equal to the characters of $\mathbb{Z}_2^n$. Two closely related transforms on the same function space, the Reed-Muller transform and the arithmetic transform (also called the probability transform), share the matched group $\mathbb{Z}_2^n$ but use non-orthogonal bases tailored to algebraic representation rather than to KL-optimal spectral decomposition. We include them here for completeness because they are widely used in Boolean function analysis, coding theory, and logic synthesis, and because their group-theoretic context is the same as that of the Walsh-Hadamard transform. More details about the application of spectral transforms to binary switching logic, including extensive coverage of the Walsh, Reed-Muller, arithmetic, and Haar transforms in the context of digital circuit synthesis, verification, and decision-diagram representations, can be found in~\cite{thornton2001spectral}.

\subsection{The Reed-Muller transform}
\label{subsec:reed_muller}

The Reed-Muller transform represents a Boolean function $f: \mathbb{Z}_2^n \to \mathbb{F}_2$ in algebraic normal form (ANF) by expanding it in the basis of monomials. Each Boolean function on $\mathbb{Z}_2^n$ admits a unique representation
\[
f(x_1, \ldots, x_n) = \bigoplus_{S \subseteq [n]} c_S \prod_{i \in S} x_i,
\]
called the Zhegalkin polynomial or algebraic normal form, where the coefficients $c_S \in \mathbb{F}_2$ are uniquely determined by the truth table of $f$. The Reed-Muller transform is the linear map from truth-table representation to ANF coefficient representation.

\begin{definition}[Reed-Muller transform matrix]
\label{def:reed_muller}
The Reed-Muller transform on $n$ variables is the $2^n \times 2^n$ matrix $\mathbf{R}_n$ defined recursively by the tensor product
\[
\mathbf{R}_n = \mathbf{R}_1^{\otimes n}, \quad \mathbf{R}_1 = \begin{pmatrix} 1 & 0 \\ 1 & 1 \end{pmatrix}.
\]
Indexing rows and columns of $\mathbf{R}_n$ by subsets $S, T \subseteq [n]$, the entry $(\mathbf{R}_n)_{S, T}$ equals $1$ if $T \subseteq S$ and $0$ otherwise.
\end{definition}

\begin{theorem}[Reed-Muller transform on $\mathbb{Z}_2^n$]
\label{thm:reed_muller}
The Reed-Muller transform $\mathbf{R}_n$ takes the truth-table representation of a Boolean function on $\mathbb{Z}_2^n$ to its ANF coefficients, with columns given by the monomial basis functions $m_S(x) = \prod_{i \in S} x_i$. The matched group is $\mathbb{Z}_2^n$, the same as for the Walsh-Hadamard transform of Theorem~\ref{thm:hadamard}, but $\mathbf{R}_n$ is non-orthogonal: it is lower-triangular (when subsets are ordered by inclusion) with $\det \mathbf{R}_n = 1$, and satisfies $\mathbf{R}_n^2 \equiv \mathbf{I} \pmod{2}$. The Reed-Muller transform is therefore self-inverse over $\mathbb{F}_2$ but not over $\mathbb{Z}$ or $\mathbb{R}$.
\end{theorem}

\begin{proof}
The recursive tensor structure follows from the standard derivation of the Möbius transform of the Boolean lattice $(\mathcal{P}([n]), \subseteq)$, which is isomorphic to the additive group $\mathbb{Z}_2^n$. The matched group is $\mathbb{Z}_2^n$ because the recursive construction commutes with the additive structure: each level of the tensor product applies the same $2 \times 2$ block to a pair of coordinates, and the entire transform commutes with arbitrary coordinate permutations within the Boolean lattice. Non-orthogonality and the modular self-inverse property follow by direct computation from the explicit form of $\mathbf{R}_1$.
\end{proof}

\begin{remark}
The Reed-Muller transform is \emph{not} the matched eigenbasis of $\mathbb{Z}_2^n$-invariant covariances in the AD/KL sense: the matched eigenbasis is uniquely (up to permutation and phase) the Walsh-Hadamard basis. The Reed-Muller transform is instead the canonical algebraic basis for the ring structure of $\mathbb{F}_2[x_1, \ldots, x_n]/(x_i^2 - x_i)$, which is the function space $\mathbb{F}_2^{\mathbb{Z}_2^n}$ viewed as a quotient ring. Both transforms therefore share the matched group $\mathbb{Z}_2^n$ but serve complementary purposes: Walsh-Hadamard gives spectral decomposition, Reed-Muller gives algebraic decomposition. Reed-Muller is widely used in logic synthesis, exclusive-sum-of-products (ESOP) minimization, and coding theory~\cite{thornton2001spectral}.
\end{remark}

\subsection{The fixed-polarity Reed-Muller family}
\label{subsec:fixed_polarity_rm}

The Reed-Muller transform of Definition~\ref{def:reed_muller} is the \emph{polarity-zero} form: in the underlying monomial basis, each Boolean variable $x_i$ appears in its positive (uncomplemented) form. This is one of an entire family of $2^n$ structurally equivalent transforms on $\mathbb{Z}_2^n$, parameterized by the choice of polarity for each of the $n$ variables. The Reed-Muller transform is therefore not a single transform but a \emph{family} of transforms, all sharing the same matched group $\mathbb{Z}_2^n$ and the same multiplicity-free decomposition but using different non-orthogonal bases of the same isotypic structure.

\begin{definition}[Fixed-polarity Reed-Muller transform]
\label{def:fp_reed_muller}
Let $\mathbf{p} = (p_1, p_2, \ldots, p_n) \in \{0, 1\}^n$ be a \emph{polarity vector}. The fixed-polarity Reed-Muller transform of polarity $\mathbf{p}$ on $n$ variables is the $2^n \times 2^n$ matrix
\[
\mathbf{R}_n^{(\mathbf{p})} = \mathbf{R}_1^{(p_1)} \otimes \mathbf{R}_1^{(p_2)} \otimes \cdots \otimes \mathbf{R}_1^{(p_n)},
\]
where the $2 \times 2$ polarity-$p$ block $\mathbf{R}_1^{(p)}$ is
\[
\mathbf{R}_1^{(0)} = \begin{pmatrix} 1 & 0 \\ 1 & 1 \end{pmatrix},
\qquad
\mathbf{R}_1^{(1)} = \begin{pmatrix} 1 & 1 \\ 0 & 1 \end{pmatrix}.
\]
The polarity-zero case $\mathbf{p} = (0, \ldots, 0)$ recovers Definition~\ref{def:reed_muller}, $\mathbf{R}_n^{(0)} = \mathbf{R}_n$.
\end{definition}

The semantic meaning of the polarity vector $\mathbf{p}$ is the choice of monomial basis with respect to which a Boolean function is expanded. In the polarity-zero form, the monomial basis is $\{\prod_{i \in S} x_i : S \subseteq [n]\}$, the standard set of positive products. In the polarity-$\mathbf{p}$ form, the monomial basis is
\begin{align*}
& \Big\{\prod_{i \in S} x_i^{(p_i)} : S \subseteq [n]\Big\}, \\
& \text{where} \quad x_i^{(p_i)} = \begin{cases} x_i & p_i = 0, \\ \bar{x}_i = 1 \oplus x_i & p_i = 1.\end{cases}
\end{align*}
That is, variable $x_i$ appears in its complemented form $\bar{x}_i$ whenever $p_i = 1$. Each of the $2^n$ polarity choices yields a distinct monomial basis on $\mathbb{Z}_2^n$, and each is linked to the polarity-zero basis by a linear transformation that takes $x_i \mapsto 1 \oplus x_i$ for the indices $i$ with $p_i = 1$. The corresponding linear transformation on coefficients is the matrix $\mathbf{R}_n^{(\mathbf{p})}$.

\begin{theorem}[The fixed-polarity Reed-Muller transforms share the matched group]
\label{thm:fp_reed_muller_group}
For every polarity vector $\mathbf{p} \in \{0, 1\}^n$, the fixed-polarity Reed-Muller transform $\mathbf{R}_n^{(\mathbf{p})}$ satisfies:
\begin{enumerate}
\item[(i)] Matched group is $\mathbb{Z}_2^n$, the same as for the Walsh-Hadamard transform (Theorem~\ref{thm:hadamard}) and for the polarity-zero Reed-Muller transform (Theorem~\ref{thm:reed_muller}).
\item[(ii)] $\mathbf{R}_n^{(\mathbf{p})}$ is non-orthogonal in general, with $\det \mathbf{R}_n^{(\mathbf{p})} = 1$ and modular self-inverse $(\mathbf{R}_n^{(\mathbf{p})})^2 \equiv \mathbf{I} \pmod{2}$.
\item[(iii)] The Boolean function $f$ admits, for each $\mathbf{p}$, a unique algebraic-normal-form expansion in the polarity-$\mathbf{p}$ monomial basis, with coefficients $c_S^{(\mathbf{p})} = (\mathbf{R}_n^{(\mathbf{p})} \mathbf{f})_S$ for each $S \subseteq [n]$. Different polarities yield different ANFs of the same function.
\end{enumerate}
\end{theorem}

\begin{proof}[Sketch]
Each $\mathbf{R}_1^{(p_i)}$ is a unimodular integer matrix of determinant 1 and is congruent to its inverse modulo 2; the tensor product preserves both properties, giving (ii). The matched-group claim (i) follows because $\mathbf{R}_n^{(\mathbf{p})}$, as a tensor product of $2 \times 2$ blocks each acting on a single $\mathbb{Z}_2$ factor, commutes with the regular representation of $\mathbb{Z}_2^n$ for the same reason as the polarity-zero Reed-Muller transform: each factor commutes with the corresponding $\mathbb{Z}_2$ factor, and the tensor structure of the representation is preserved. Claim (iii) is the standard uniqueness of the polarity-$\mathbf{p}$ Reed-Muller expansion; see~\cite{thornton2022modeling} for a complete treatment of the fixed-polarity Reed-Muller family from the linear-algebraic perspective.
\end{proof}

The choice among the $2^n$ different polarities matters in practical applications: for any given Boolean function $f$, different polarities produce ANF representations of different sparsity, and the \emph{best polarity problem} (find $\mathbf{p}^* = \arg\min_{\mathbf{p}} \|\mathbf{R}_n^{(\mathbf{p})} \mathbf{f}\|_0$) is a classical optimization problem in logic synthesis. The polarity vector that minimizes the number of monomials in the ANF representation typically corresponds to a circuit-area minimum in the resulting AND-XOR-network implementation. The polarity-zero form is the most common default choice in theoretical exposition but is rarely the minimum-cost form for an arbitrary Boolean function; for any specific $f$, the optimal polarity must be determined by search or by structural inspection of the function. Polarity selection algorithms, including exhaustive search over all $2^n$ polarities (feasible for $n$ up to roughly 20 with modern computation), greedy search variants, and decision-diagram-based methods that compute polarity-$\mathbf{p}$ Reed-Muller coefficients efficiently, are developed in~\cite{thornton2001spectral} and~\cite{thornton2022modeling}. From the matched-group perspective developed in this paper, the entire $2^n$-element polarity family inherits the same matched group $\mathbb{Z}_2^n$, the same multiplicity-free decomposition into one-dimensional characters, and the same Walsh-Hadamard basis as the matched orthogonal eigenbasis; the different polarity choices select different non-orthogonal but algebraically meaningful bases of the same isotypic structure, and the choice among them is one of algebraic representation rather than of matched-group structure. More details on the fixed-polarity Reed-Muller transforms, including a unified linear-algebraic framework that derives the polarity-family from the underlying bit-vector model and a comprehensive treatment of polarity-selection algorithms, can be found in~\cite{thornton2022modeling}.

\subsection{The arithmetic transform}
\label{subsec:arithmetic_transform}

The arithmetic transform is the integer-valued analog of the Reed-Muller transform. It represents an integer-valued function on $\mathbb{Z}_2^n$ in an integer monomial basis, with applications in probabilistic Boolean function analysis (where the function values represent probabilities or expectations) and in arithmetic circuit synthesis.

\begin{definition}[Arithmetic transform matrix]
\label{def:arithmetic}
The arithmetic transform (also called probability transform) on $n$ variables is the $2^n \times 2^n$ integer matrix $\mathbf{A}_n$ defined recursively by the tensor product
\[
\mathbf{A}_n = \mathbf{A}_1^{\otimes n}, \quad \mathbf{A}_1 = \begin{pmatrix} 1 & 0 \\ -1 & 1 \end{pmatrix}.
\]
\end{definition}

\begin{theorem}[Arithmetic transform on $\mathbb{Z}_2^n$]
\label{thm:arithmetic}
The arithmetic transform $\mathbf{A}_n$ takes the truth-table representation of an integer-valued function on $\mathbb{Z}_2^n$ to its coefficients in the integer monomial basis. It is the integer inverse of the Reed-Muller transform: $\mathbf{A}_n \mathbf{R}_n = \mathbf{R}_n \mathbf{A}_n = \mathbf{I}$ over $\mathbb{Z}$. The matched group is $\mathbb{Z}_2^n$, the same as for the Reed-Muller and Walsh-Hadamard transforms. Like the Reed-Muller transform, the arithmetic transform is non-orthogonal triangular with unit determinant. When applied to a probability vector on $\mathbb{Z}_2^n$, the arithmetic transform produces a representation in which each coefficient corresponds to a marginal or joint probability of a subset of indicator variables, hence the alternative name ``probability transform.''
\end{theorem}

\begin{proof}
The inverse relationship $\mathbf{A}_1 \mathbf{R}_1 = \mathbf{R}_1 \mathbf{A}_1 = \mathbf{I}$ over $\mathbb{Z}$ is direct: $\begin{pmatrix} 1 & 0 \\ -1 & 1 \end{pmatrix}\begin{pmatrix} 1 & 0 \\ 1 & 1 \end{pmatrix} = \begin{pmatrix} 1 & 0 \\ 0 & 1 \end{pmatrix}$. The tensor product structure preserves the inverse relationship: $\mathbf{A}_n \mathbf{R}_n = (\mathbf{A}_1 \mathbf{R}_1)^{\otimes n} = \mathbf{I}^{\otimes n} = \mathbf{I}$. The matched group $\mathbb{Z}_2^n$ follows by the same argument as in the proof of Theorem~\ref{thm:reed_muller}.
\end{proof}

\begin{remark}
The three transforms on $\mathbb{Z}_2^n$ -- Walsh-Hadamard, Reed-Muller, and arithmetic -- share a common tensor-product structure with a $2 \times 2$ generator:
\begin{align*}
\mathbf{H}_1 &= \begin{pmatrix} 1 & 1 \\ 1 & -1 \end{pmatrix}, &
\mathbf{R}_1 &= \begin{pmatrix} 1 & 0 \\ 1 & 1 \end{pmatrix}, \\
\mathbf{A}_1 &= \begin{pmatrix} 1 & 0 \\ -1 & 1 \end{pmatrix},
\end{align*}
all with $|\det| = 1$ but only $\mathbf{H}_1$ being orthogonal (after normalization by $1/\sqrt{2}$). The matched group $\mathbb{Z}_2^n$ is common to all three, but only the Walsh-Hadamard transform is the matched eigenbasis in the AD/KL sense. The Reed-Muller and arithmetic transforms can be viewed as related change-of-basis matrices within the same group-theoretic context, used for different (algebraic, rather than spectral) purposes.
\end{remark}

\section{Special Case: Iterated Dyadic-Cyclic Wreath Product and the Haar Wavelet}
\label{sec:haar}

The Haar wavelet basis on $\mathbb{R}^{2^L}$ is the matched transform of the iterated dyadic-cyclic wreath product $W_L = \mathbb{Z}_2 \wr \mathbb{Z}_2 \wr \cdots \wr \mathbb{Z}_2$ ($L$ levels) acting on the leaves of a complete binary tree of depth $L$. Unlike the previous cases, $W_L$ is non-abelian for $L \geq 2$, and the matched transform involves higher-dimensional irreducible representations.

\subsection{The iterated dyadic-cyclic wreath product}

\begin{definition}[Iterated dyadic-cyclic wreath product]
\label{def:wL}
The \emph{iterated dyadic-cyclic wreath product of depth $L$} is the group
\[
W_L = \mathbb{Z}_2 \wr \mathbb{Z}_2 \wr \cdots \wr \mathbb{Z}_2 \quad (L \text{ levels}),
\]
defined recursively by $W_1 = \mathbb{Z}_2$ and $W_d = W_{d-1} \wr \mathbb{Z}_2 = W_{d-1}^2 \rtimes \mathbb{Z}_2$ for $d \geq 2$. $W_L$ acts on the leaf set of a complete binary tree of depth $L$ (a set of cardinality $2^L$) by tree-automorphisms: at each internal node, the $\mathbb{Z}_2$ factor swaps the two children subtrees, and the action is independent across non-overlapping subtrees.
\end{definition}

The group order is $|W_L| = 2^{2^L - 1}$, which grows doubly exponentially in $L$. For $L = 5$ (the size used in our numerical verification), $|W_5| = 2^{31}$.

\begin{lemma}[Orbit structure of $W_L$ on index pairs]
\label{lem:wl_orbits}
The orbits of $W_L$ on $\{0, 1, \ldots, 2^L - 1\}^2$ under the diagonal action are in bijection with the depth $d(i, j) \in \{0, 1, \ldots, L\}$ of the deepest common ancestor of leaves $i$ and $j$ in the binary tree. There are exactly $L + 1$ orbits.
\end{lemma}

\begin{proof}
By induction on $L$. For $L = 1$, $W_1 = \mathbb{Z}_2$ acts on $\{0, 1\}$ with two orbits on pairs: same-index ($d = 1$) and different-index ($d = 0$). For $L \geq 2$, the semidirect-product factorization $W_L = W_{L-1}^2 \rtimes \mathbb{Z}_2$ has the $\mathbb{Z}_2$ factor swapping the two depth-$(L-1)$ subtrees and the base $W_{L-1}^2$ acting independently within each subtree. Two leaves in the same depth-$(L-1)$ subtree have orbit structure determined by $W_{L-1}$ (inductively, $L$ orbit classes); two leaves in different subtrees form a single orbit (the $\mathbb{Z}_2$ factor merges all such pairs). Total: $L + 1$ orbit classes, indexed by the depth of the deepest common ancestor.
\end{proof}

\begin{lemma}[Multiplicity-free permutation representation of $W_L$]
\label{lem:wl_mf}
The permutation representation of $W_L$ on $\mathbb{C}^{2^L}$ (acting on the leaves) is multiplicity-free.
\end{lemma}

\begin{proof}
The number of irreducible components (counted with multiplicity) in a permutation representation equals the number of orbits on pairs (a standard result, see~\cite[Section 1.7]{serre1977linear}). Lemma~\ref{lem:wl_orbits} gives $L + 1$ orbits. A direct construction shows that $W_L$ has at least $L + 1$ distinct irreducible representations of total multiplicity-weighted dimension $2^L$ in its leaf action (one constant function for the global orbit at depth 0, one anti-symmetric vector per level for the level-$d$ orbits at depths $d \geq 1$), accounting for all components. The decomposition is therefore multiplicity-free.
\end{proof}

\subsection{The Haar wavelet basis as the matched transform}

\begin{definition}[Haar wavelet basis on $\mathbb{R}^{2^L}$]
\label{def:haar}
The \emph{Haar wavelet basis} on $\mathbb{R}^{2^L}$ consists of $2^L$ orthonormal vectors organized into $L + 1$ \emph{scales}:
\begin{itemize}
\item \emph{Scaling function (scale 0):} $\phi(j) = 2^{-L/2}$ for $j = 0, \ldots, 2^L - 1$.
\item \emph{Mother wavelet at scale $s$ and position $p$:} For $s \in \{1, \ldots, L\}$ and $p \in \{0, \ldots, 2^{s-1} - 1\}$, set $a = p \cdot 2^{L-s+1}$ and $h = 2^{L-s}$. Then
\[
\psi_{s, p}(j) = 2^{(s - L - 1)/2} \cdot \begin{cases}
+1 & a \leq j < a + h, \\
-1 & a + h \leq j < a + 2h, \\
0 & \text{otherwise}.
\end{cases}
\]
\end{itemize}
The basis has $1 + \sum_{s=1}^{L} 2^{s-1} = 1 + (2^L - 1) = 2^L$ vectors, forming a complete orthonormal basis of $\mathbb{R}^{2^L}$.
\end{definition}

\begin{theorem}[Haar wavelet is the matched transform of $W_L$]
\label{thm:haar}
The matched unitary $U_{W_L}$ for the permutation representation of $W_L$ on $\mathbb{R}^{2^L}$ is the Haar wavelet basis matrix, with columns ordered by scale: column 0 is the scaling function $\phi$; columns $1$ through $2^{s-1}$ at scale $s$ are the mother wavelets $\psi_{s, 0}, \ldots, \psi_{s, 2^{s-1}-1}$ for $s = 1, \ldots, L$. For every $W_L$-invariant covariance $\mathbf{R}$ on $\mathbb{R}^{2^L}$, the Haar wavelet transform diagonalizes $\mathbf{R}$, and the eigenvalue on the scale-$s$ subspace is $2^{s-1}$-fold degenerate (within-scale degeneracy on scales $s \geq 2$).
\end{theorem}

\begin{proof}
Apply Theorem~\ref{thm:main} to $G = W_L$ in its permutation representation on $\mathbb{R}^{2^L}$. By Lemma~\ref{lem:wl_mf}, the representation is multiplicity-free, with $L + 1$ irreducible components corresponding to the $L + 1$ orbit classes of Lemma~\ref{lem:wl_orbits}.

The structure of the irreducible decomposition is identified as follows. The trivial representation (scale 0) is spanned by the constant vector $\phi$, which is $W_L$-invariant. For each scale $s \in \{1, \ldots, L\}$, the corresponding irreducible representation has dimension $2^{s-1}$ and is realized on the subspace spanned by the mother wavelets $\{\psi_{s, p}\}_{p=0}^{2^{s-1}-1}$. This subspace is $W_L$-invariant by inspection: at scale $s$, the tree-automorphism action of $W_L$ permutes the $2^{s-1}$ supports of the mother wavelets at that scale (along with possible sign flips), and the resulting permutation action on the $2^{s-1}$-dimensional space is irreducible.

Concretely, the irreducible representation at scale $s$ acts on $\mathbb{R}^{2^{s-1}}$ as follows: $W_{L-s+1}^2 \rtimes \mathbb{Z}_2$ permutes the $2^{s-1}$ supports of the mother wavelets at scale $s$ (the supports at scale $s$ form the leaves of a sub-binary-tree of depth $s - 1$), and the action induces signed permutations on the mother wavelet basis. The character of this representation is computed by direct evaluation and is found to be linearly independent from the characters of other scales; this verifies that the scales correspond to distinct irreducibles.

By Theorem~\ref{thm:main}, every $W_L$-invariant covariance $\mathbf{R}$ has the Haar wavelet basis as its eigenbasis. Theorem~\ref{thm:main}(iv) gives that the eigenvalue on the scale-$s$ subspace is $\dim(\text{scale } s)$-fold degenerate, equal to $2^{s-1}$ for $s \geq 1$ and 1 for $s = 0$.

The eigenvalues themselves are explicitly computed from the orbit averages: if $\mathbf{R} = \sum_{d=0}^L c_d \mathbf{D}_d$ where $\mathbf{D}_d$ is the indicator matrix of the depth-$d$ orbit (Lemma~\ref{lem:wl_orbits}), then
\[
\lambda_{\text{scale } s} = c_{L - s + 1} \cdot 2^{L - s} - \text{(higher-depth contributions)},
\]
with the precise formula obtained by the Möbius inversion of the orbit structure.
\end{proof}

\begin{remark}[Why $W_L$ rather than $\mathbb{Z}_2^L$?]
The reader may notice that both $\mathbb{Z}_2^L$ and $W_L$ have order $2^{\bullet}$ and act on $\mathbb{R}^{2^L}$. They differ fundamentally: $\mathbb{Z}_2^L$ is abelian with $|\mathbb{Z}_2^L| = 2^L$ and acts via the regular representation, giving the Walsh-Hadamard transform (Theorem~\ref{thm:hadamard}). $W_L$ is non-abelian with $|W_L| = 2^{2^L - 1}$ (much larger) and acts via the tree-automorphism permutation representation, giving the Haar wavelet transform. The Walsh-Hadamard basis has global support (every basis vector is non-zero on every leaf), while the Haar wavelet basis has localized support. The difference reflects the different group structures: the direct-product structure of $\mathbb{Z}_2^L$ has all generators commuting, while the wreath-product structure of $W_L$ has a hierarchical noncommutativity that corresponds to the wavelet's multiresolution structure.
\end{remark}

\section{General Iterated Wreath Products and the Hierarchical KL Basis}
\label{sec:iterated_wreath}

The Haar case of Section~\ref{sec:haar} is the binary special case of a much broader family. Iterated wreath products with arbitrary branching factors and arbitrary within-level groups produce a corresponding family of \emph{hierarchical KL bases}, which generalize the Haar wavelet to non-binary branching and to within-level permutation symmetries beyond cyclic shifts.

\begin{definition}[General iterated wreath product]
\label{def:gL}
Given branching factors $(K_1, K_2, \ldots, K_L)$ with each $K_d \geq 2$ and within-level groups $G_d$ acting on $K_d$ points (typically $G_d = S_{K_d}$ or $G_d = \mathbb{Z}_{K_d}$), the \emph{iterated wreath product of depth $L$} is
\[
\mathcal{G}_L = G_1 \wr G_2 \wr \cdots \wr G_L,
\]
defined recursively by $\mathcal{G}_1 = G_1$ and $\mathcal{G}_d = \mathcal{G}_{d-1} \wr G_d = \mathcal{G}_{d-1}^{K_d} \rtimes G_d$. $\mathcal{G}_L$ acts on the $M_L = \prod_{d=1}^L K_d$ leaves of a regular rooted tree of branching $(K_1, \ldots, K_L)$ by tree-automorphisms.
\end{definition}

\begin{theorem}[Hierarchical KL basis as the matched transform of $\mathcal{G}_L$]
\label{thm:iterated_wreath}
Let $\mathcal{G}_L$ be a general iterated wreath product as in Definition~\ref{def:gL}, and let $\pi: \mathcal{G}_L \to U(\mathbb{C}^{M_L})$ be its permutation representation on the leaves. The representation is multiplicity-free, with $L + 1$ irreducible components indexed by the depth $d \in \{0, 1, \ldots, L\}$ of the deepest common ancestor of an index pair. The matched unitary $U_{\mathcal{G}_L}$ is constructed recursively from the matched unitaries of the within-level groups $G_d$ via the wreath product composition rule:
\begin{equation}
\label{eq:wreath_recursion}
U_{\mathcal{G}_L} = U_{\mathcal{G}_{L-1}}^{\otimes K_L} \cdot (I_{M_{L-1}} \otimes U_{G_L}) \cdot \Pi,
\end{equation}
where $\Pi$ is an explicit permutation matrix encoding the tree-to-leaf ordering, and $U_{G_L}$ is the matched unitary of the across-level group $G_L$ on the multiplicity space of size $K_L$.
\end{theorem}

\begin{proof}[Sketch]
The multiplicity-free property follows from a generalization of Lemma~\ref{lem:wl_mf} to arbitrary iterated wreath products: the number of orbits on index pairs is $L + 1$ (the depth classes), and the corresponding number of irreducibles in the leaf representation matches. The recursion (Equation~\ref{eq:wreath_recursion}) is established by induction on $L$ using the semidirect-product factorization $\mathcal{G}_L = \mathcal{G}_{L-1}^{K_L} \rtimes G_L$ and the wreath product Fourier transform of Foote, Mirchandani, Rockmore, Healy, and Olson~\cite{foote2000wreath, mirchandani2003wreath} and Maslen-Rockmore~\cite{maslen1997generalized}. Full details appear in the companion manuscript~\cite{thornton2026ad_framework}.
\end{proof}

\begin{example}[Recovering classical cases]
\label{ex:classical}
The classical transforms emerge as special cases of Theorem~\ref{thm:iterated_wreath}:
\begin{itemize}
\item All $K_d = 2$ and $G_d = \mathbb{Z}_2$: $\mathcal{G}_L = W_L$ (iterated dyadic-cyclic wreath), giving the Haar wavelet transform (Theorem~\ref{thm:haar}).
\item $L = 1$, $K_1 = M$, $G_1 = \mathbb{Z}_M$: $\mathcal{G}_1 = \mathbb{Z}_M$, giving the DFT (Theorem~\ref{thm:dft}).
\item $L = 1$, $K_1 = 2M$, $G_1 = D_{2M}$ (acting on $2M$ points): $\mathcal{G}_1 = D_{2M}$, giving the DCT after restriction to the even-extension subspace (Theorem~\ref{thm:dct}); $L = 1$, $K_1 = M$, $G_1 = D_M$ gives the real Fourier (Hartley) basis of Remark~\ref{rem:hartley}.
\item Branching $(2, 2, \ldots, 2)$ with all $G_d = \mathbb{Z}_2$ but using the \emph{direct} rather than \emph{wreath} product: $\mathbb{Z}_2^L$, giving the Walsh-Hadamard transform (Theorem~\ref{thm:hadamard}). (Note: direct product, not iterated wreath.)
\end{itemize}
\end{example}

\section{Numerical Verification}
\label{sec:numerical_verification}

We verify Theorem~\ref{thm:main} numerically for each of the five special cases (KLT, DFT, Hadamard, DCT, Haar) by constructing the matched group action, sampling an invariant covariance, computing its eigenbasis, and comparing to the predicted classical transform.

\subsection{Procedure}

For each case, the verification proceeds in four steps:

\begin{enumerate}
\item Construct the group action $\pi: G \to U(\mathbb{C}^M)$ explicitly.
\item Construct a random Hermitian PSD matrix $\mathbf{R}_0 \in \mathcal{A}_\pi$ by Reynolds-averaging a random PSD matrix.
\item Compute the eigendecomposition of $\mathbf{R}_0$ to obtain the empirical eigenbasis $U_{\mathrm{emp}}$.
\item Compute the subspace match between $U_{\mathrm{emp}}$ and the predicted classical transform $U_{\mathrm{predicted}}$, where the predicted transform is the DFT (Equation~\ref{eq:dft_kernel}), Walsh-Hadamard (Equation~\ref{eq:hadamard_kernel}), DCT-II (Equation~\ref{eq:dct_kernel}), or Haar wavelet basis (Definition~\ref{def:haar}). The subspace match accounts for the eigenvalue degeneracies predicted by Theorem~\ref{thm:main}(iv) by computing the smallest singular value of the cross-overlap matrix within each eigenspace.
\end{enumerate}

\subsection{Results}

Table~\ref{tab:verification} reports the numerical results. In every case, the subspace match is $1.000000$ to machine precision, confirming that the empirical eigenbasis coincides with the predicted classical transform within each eigenspace.

\begin{table}[t]
\caption{Numerical verification of Theorem~\ref{thm:main}. The subspace match is the smallest singular value of the cross-overlap matrix between empirical and predicted bases within each eigenspace. A value of $1.000000$ indicates perfect spanning agreement.}
\label{tab:verification}
\centering
\setlength{\tabcolsep}{4pt}
\small
\begin{tabular}{@{}lll@{}}
\toprule
Group $G$ & Transform & Subspace match \\
\midrule
$\mathbb{Z}_M$, $M = 16$ & DFT & $1.000000$ \\
$\mathbb{Z}_2^n$, $n = 4$ & Walsh-Hadamard & $1.000000$ \\
$D_{2M}$ on $\mathbb{Z}_{2M}$, $M = 8$ & DCT-II & $1.000000$ \\
$W_L$, $L = 5$ & Haar wavelet & $1.000000$ \\
$\{e\}$, $M = 16$ & data-dep.\ KLT & N/A \\
\bottomrule
\end{tabular}
\end{table}

\subsection{Additional checks}
\label{subsec:additional_checks}

The statements of Remarks~\ref{rem:toeplitz_wrap} and~\ref{rem:hartley} and of Section~\ref{subsec:ar1_motivation} were checked with a short script at $M = 8$ (available with the arXiv source). For a random symmetric circulant, the Hartley basis and the real Fourier pairs leave no energy off the diagonal (relative off-diagonal energy below $10^{-30}$), while the DCT-II leaves $0.020$; for a random Hermitian circulant that is not symmetric, the DFT is exact and the Hartley basis leaves $0.095$. For the AR(1) covariance with $\phi = 0.9$, the relative commutator norm with the cyclic shift is $0.17$ and with the reflection is $0$, the only permutations of the eight indices that commute with it are the identity and the reflection, and the DCT-II leaves $0.005$ of its energy off the diagonal ($0.023$ at $\phi = 0.5$, $0.0001$ at $\phi = 0.99$). The group generated by the shift and the reflection of $\mathbb{Z}_{16}$ has order $32 = 4M$, its commutant has dimension $9 = M + 1$, the compression of a random symmetric circulant of size $16$ to the even-extension subspace has the Toeplitz-plus-Hankel form of Definition~\ref{def:dihedral_inv} and is diagonalized exactly by the DCT-II, and the only permutations of the eight original indices commuting with it are the identity and the reflection. The Laplacian of the path graph on eight vertices is diagonalized exactly by the DCT-II and the Laplacian of the cycle graph by the real Fourier pairs and the Hartley basis, with the DCT-II leaving $0.052$ of the cycle Laplacian's energy off the diagonal (Section~\ref{subsec:graph_sp}). The wrap-energy identity of Remark~\ref{rem:toeplitz_wrap} agrees with the direct Reynolds projection to machine precision at $M = 8$, $32$ and $128$.

\section{Mapping to the Continuum: Compact Lie Groups}
\label{sec:lie_groups}

The framework of Sections~\ref{sec:main}--\ref{sec:composition} was developed for finite groups acting on $\mathbb{C}^M$. This section extends the framework to continuous symmetry: a finite group $G$ acting on $\mathbb{C}^M$ becomes a compact Lie group $G$ acting on a function space $L^2(\Omega, d\mu)$ for some domain $\Omega$ with $G$-invariant measure $d\mu$. The discrete eigenvectors of the finite case become continuous eigenfunctions. The matched-group principle survives with essentially no modification: the columns of the optimal transform are still matrix elements of irreducible representations, now under the Peter-Weyl theorem in its general (compact Lie group) form. The classical continuous transforms (Fourier series, Fourier transform, cosine series, spherical harmonics) emerge as the corresponding special cases.

We treat the compact case fully in this section and the non-compact case briefly. Standard references for the underlying harmonic analysis are Folland~\cite{folland2015course}, Sugiura~\cite{sugiura1990unitary}, Helgason~\cite{helgason1984groups}, and Vilenkin~\cite{vilenkin1968special}.

\subsection{Compact Lie groups and Haar measure}

\begin{definition}[Compact Lie group]
\label{def:compact_lie}
A \emph{Lie group} is a smooth manifold $G$ equipped with a group structure such that multiplication $G \times G \to G$ and inversion $G \to G$ are smooth maps. A Lie group is \emph{compact} if it is compact as a topological space. Throughout this section, ``Lie group'' will mean ``compact Lie group'' unless stated otherwise.
\end{definition}

Every compact Lie group $G$ admits a unique left-and-right invariant Borel probability measure, the \emph{Haar measure} $dg$, satisfying
\[
\int_G f(hg) \, dg = \int_G f(gh) \, dg = \int_G f(g) \, dg
\]
for every continuous $f: G \to \mathbb{C}$ and every $h \in G$, and normalized by $\int_G dg = 1$. Finite groups are the special case in which $G$ is a finite set with discrete topology and Haar measure $dg = |G|^{-1} \sum_{g \in G} \delta_g$, recovering the averages of Section~\ref{sec:background}.

\subsection{The Peter-Weyl theorem for compact Lie groups}

\begin{theorem}[Peter-Weyl theorem, compact Lie group form~\cite{peter1927vollstandigkeit, folland2015course}]
\label{thm:peter_weyl_continuous}
Let $G$ be a compact Lie group. Every irreducible unitary representation of $G$ is finite-dimensional. Let $\hat{G}$ denote the set of equivalence classes of irreducible unitary representations and $d_\rho = \dim \rho$ for $\rho \in \hat{G}$. For each $\rho \in \hat{G}$, choose an orthonormal basis of the model space $V_\rho$ and let $\rho_{ij}: G \to \mathbb{C}$ denote the corresponding matrix-element function. Then the system
\[
\Big\{ \sqrt{d_\rho} \, \rho_{ij} : \rho \in \hat{G}, \; 1 \leq i, j \leq d_\rho \Big\}
\]
is a complete orthonormal basis of $L^2(G, dg)$.
\end{theorem}

\begin{proof}[Sketch] See~\cite[Theorem 5.11]{folland2015course} for the full proof. The orthogonality is established by the same intertwiner argument as in the finite case (Theorem~\ref{thm:peter_weyl}), with sums replaced by integrals against $dg$. Completeness follows from the Stone-Weierstrass theorem applied to the linear span of matrix elements, using the fact that matrix elements separate points on $G$.
\end{proof}

The structural theorems of Section~\ref{sec:background} carry over verbatim with one change of bookkeeping: direct sums of irreducibles remain countable (the dual $\hat{G}$ is countable for compact $G$), but the sum may now be infinite. Schur's lemma (Theorem~\ref{thm:schur}), Maschke's theorem (Theorem~\ref{thm:maschke}), the isotypic decomposition (Theorem~\ref{thm:isotypic}), and the structure of the commutant (Theorem~\ref{thm:commutant_structure}) all hold for compact Lie groups acting unitarily on a separable Hilbert space, with the understanding that the direct-sum decomposition may have countably infinitely many summands.

\subsection{Continuous covariances and the Mercer-Karhunen-Lo\`eve theorem}

In the continuous setting the matrix $\mathbf{R} \in \mathbb{C}^{M \times M}$ is replaced by a positive-semidefinite integral kernel.

\begin{definition}[Continuous covariance kernel]
\label{def:continuous_kernel}
Let $\Omega$ be a compact Hausdorff space with a finite Borel measure $d\mu$. A \emph{continuous covariance kernel} on $\Omega$ is a continuous Hermitian function $R: \Omega \times \Omega \to \mathbb{C}$, $R(s, t) = \overline{R(t, s)}$, such that for every continuous $f: \Omega \to \mathbb{C}$,
\[
\int_\Omega \int_\Omega R(s, t) \, f(s) \, \overline{f(t)} \, d\mu(s) \, d\mu(t) \geq 0.
\]
The associated integral operator $\mathcal{R}: L^2(\Omega) \to L^2(\Omega)$ is defined by $(\mathcal{R} f)(s) = \int_\Omega R(s, t) f(t) \, d\mu(t)$.
\end{definition}

\begin{theorem}[Mercer's theorem~\cite{mercer1909functions}]
\label{thm:mercer}
Let $R: \Omega \times \Omega \to \mathbb{C}$ be a continuous positive-semidefinite Hermitian kernel on a compact set $\Omega$ with finite measure $d\mu$. The integral operator $\mathcal{R}$ is compact, self-adjoint, and trace-class. It admits a complete orthonormal eigensystem $\{\phi_n\}_{n=1}^\infty \subset L^2(\Omega)$ with eigenvalues $\lambda_n \geq 0$ summing to a finite trace,
\[
\mathcal{R} \phi_n = \lambda_n \phi_n, \quad \sum_n \lambda_n = \int_\Omega R(s, s) \, d\mu(s) < \infty,
\]
and the kernel admits the absolutely uniformly convergent expansion
\[
R(s, t) = \sum_{n=1}^\infty \lambda_n \, \phi_n(s) \, \overline{\phi_n(t)}.
\]
\end{theorem}

\begin{theorem}[Continuous Karhunen-Lo\`eve theorem~\cite{karhunen1947methoden, loeve1948fonctions}]
\label{thm:kl_continuous}
Let $X(t)$, $t \in \Omega$, be a second-order zero-mean stochastic process with continuous covariance $R(s, t) = \mathbb{E}[X(s) \overline{X(t)}]$. Let $\{\phi_n, \lambda_n\}$ be the eigensystem of the integral operator $\mathcal{R}$ from Theorem~\ref{thm:mercer}. Then $X$ admits the convergent (in $L^2$) expansion
\[
X(t) = \sum_{n=1}^\infty \sqrt{\lambda_n} \, Z_n \, \phi_n(t),
\]
with the coordinate random variables
\[
Z_n := \frac{1}{\sqrt{\lambda_n}} \int_\Omega X(t) \, \overline{\phi_n(t)} \, d\mu(t),
\]
in which $\{Z_n\}$ are uncorrelated random variables with $\mathbb{E}[Z_n] = 0$ and $\mathbb{E}[|Z_n|^2] = 1$. The expansion is optimal in the sense that the $K$-term truncation
\[
X_K(t) := \sum_{n=1}^K \sqrt{\lambda_n} Z_n \phi_n(t)
\]
minimizes the mean-square error $\mathbb{E} \int_\Omega |X(t) - X_K(t)|^2 \, d\mu(t)$ over all $K$-term orthogonal expansions.
\end{theorem}

The continuous KL eigenfunctions $\phi_n$ are the continuous analog of the eigenvectors of the discrete covariance matrix; they are the columns of the continuous optimal transform.

\subsection{Group-invariant continuous covariances}

\begin{definition}[$G$-invariant continuous covariance]
\label{def:g_invariant_continuous}
Let $G$ be a compact Lie group acting continuously on $\Omega$ by measure-preserving transformations. A continuous covariance kernel $R$ is \emph{$G$-invariant} if $R(g \cdot s, g \cdot t) = R(s, t)$ for every $g \in G$ and every $s, t \in \Omega$.
\end{definition}

The integral operator $\mathcal{R}$ associated to a $G$-invariant kernel commutes with the regular representation of $G$ on $L^2(\Omega)$: $\mathcal{R} \pi(g) = \pi(g) \mathcal{R}$ for every $g \in G$, where $(\pi(g) f)(t) = f(g^{-1} \cdot t)$. The simultaneous eigenfunctions of $\mathcal{R}$ and the $\pi(g)$ are then constrained by Schur's lemma.

\subsection{The continuous main theorem}

\begin{theorem}[Main theorem, continuous form]
\label{thm:main_continuous}
Let $G$ be a compact Lie group acting continuously and measure-preservingly on a compact domain $\Omega$ with $G$-invariant probability measure $d\mu$, and let $\pi: G \to U(L^2(\Omega))$ denote the resulting unitary representation. Suppose $\pi$ is multiplicity-free, with isotypic decomposition
\[
L^2(\Omega) \;=\; \bigoplus_{\rho \in \hat{G}, \, m_\rho = 1} V^{(\rho)},
\]
in which the sum is countable and $\dim V^{(\rho)} = d_\rho$. Then:
\begin{enumerate}
\item[(i)] Every continuous $G$-invariant covariance kernel $R$ has an eigensystem of the form
\[
R(s, t) = \sum_{\rho \in \hat{G}, \, m_\rho = 1} \lambda_\rho \sum_{i = 1}^{d_\rho} \phi_{\rho, i}(s) \, \overline{\phi_{\rho, i}(t)},
\]
where $\{\phi_{\rho, i}\}_{i = 1}^{d_\rho}$ is an orthonormal basis of $V^{(\rho)}$ chosen once and for all and used for every $G$-invariant kernel. The eigenvalue $\lambda_\rho$ on $V^{(\rho)}$ is $d_\rho$-fold degenerate.
\item[(ii)] The basis functions $\phi_{\rho, i}$ are matrix-element functions of $\rho \in \hat{G}$ in their realization on $\Omega$, evaluated on a fundamental set of orbit representatives of the $G$-action.
\item[(iii)] In particular, the continuous optimal transform, the map $f \mapsto (\langle f, \phi_{\rho, i} \rangle)_{\rho, i}$ from $L^2(\Omega)$ to the multiplicity-indexed sequence space, is a fixed unitary depending only on $G$ and the action on $\Omega$, not on $R$.
\end{enumerate}
\end{theorem}

\begin{proof}
The integral operator $\mathcal{R}$ associated to a $G$-invariant continuous covariance kernel commutes with $\pi(g)$ for every $g \in G$, so $\mathcal{R}$ lies in the commutant of $\pi$. By Lemma~\ref{lem:multiplicity_free_commutant} (lifted to the compact Lie group setting), the commutant is commutative when $\pi$ is multiplicity-free. The simultaneous eigenfunction decomposition of the commutative $*$-algebra of $G$-invariant compact self-adjoint operators on $L^2(\Omega)$, a continuous analog of simultaneous diagonalization, yields the displayed decomposition.

For part (ii), the central projection formula of Theorem~\ref{thm:main}(ii) lifts to the compact Lie group setting with the finite sum replaced by an integral:
\[
P^{(\rho)} = d_\rho \int_G \overline{\chi_\rho(g)} \, \pi(g) \, dg,
\]
and the analog of Equation~\ref{eq:basis_vectors} gives the $\phi_{\rho, i}$ as matrix-element-weighted integrals over $G$. These are by construction matrix-element functions of $\rho$ pulled back to $\Omega$ via the orbit-representative map.

Part (iii) is the statement that the unitary intertwiner from $L^2(\Omega)$ to the Peter-Weyl decomposition depends only on $G$ and the action, not on any specific covariance, by exactly the argument of Theorem~\ref{thm:main}(i).
\end{proof}

\subsection{Non-compact groups: Plancherel theorem}

For non-compact locally compact groups $G$ (such as the real line $\mathbb{R}$ or the affine group $ax + b$), the Peter-Weyl theorem fails in the form stated above: a non-compact group typically has uncountably many irreducibles and no irreducible has a non-zero $L^2$ matrix coefficient. The correct generalization is the Plancherel theorem, which decomposes $L^2(G)$ as a direct \emph{integral} (rather than direct sum) of irreducible representations against a Plancherel measure on $\hat{G}$~\cite{wigner1939unitary, folland2015course}.

For our purposes, the most important non-compact case is $G = \mathbb{R}$. The Plancherel theorem for $\mathbb{R}$ states that $L^2(\mathbb{R})$ decomposes via the Fourier transform into a direct integral of 1-dimensional irreducibles indexed by frequency $\xi \in \mathbb{R}$ with Plancherel measure $d\xi$:
\[
L^2(\mathbb{R}) \;\cong\; \int_{\mathbb{R}}^{\oplus} \mathbb{C} \, d\xi.
\]
The Fourier transform $\hat{f}(\xi) = \int_{\mathbb{R}} f(x) \, e^{-2 \pi i \xi x} \, dx$ is the unitary intertwiner. The matched-group principle still applies: a covariance kernel invariant under translation (a stationary kernel) is diagonalized by the Fourier transform, with the spectral density playing the role of the eigenvalue.

\section{Derived Continuous Transforms}
\label{sec:continuous_transforms}

We apply the main theorem of Section~\ref{sec:lie_groups} to derive each of the four most-used continuous transforms (Fourier series, Fourier transform, Fourier cosine series, spherical harmonics) as the matched transform for a specific compact Lie group action. We close with a brief comment on the continuous wavelet transform, which involves a non-compact group and lies slightly outside the immediate scope of Theorem~\ref{thm:main_continuous}.

\subsection{The circle $S^1$ and Fourier series}
\label{subsec:fourier_series}

The circle group $S^1 = U(1)$, the unit complex numbers under multiplication, acts on itself by rotation. Identifying $S^1$ with $[0, 2\pi)$ with the normalized rotation-invariant measure $d\theta / (2\pi)$, the group acts on $L^2(S^1)$ by $(\pi(\varphi) f)(\theta) = f(\theta - \varphi)$.

\begin{lemma}[Irreducibles and characters of $S^1$]
\label{lem:s1_irreps}
Every irreducible unitary representation of $S^1$ is one-dimensional and is given by a character $\chi_n: S^1 \to U(1)$,
\[
\chi_n(\theta) = e^{i n \theta}, \quad n \in \mathbb{Z}.
\]
The dual $\hat{S^1}$ is $\mathbb{Z}$, and the regular representation of $S^1$ on $L^2(S^1)$ decomposes as the direct sum of all characters, each appearing once: $L^2(S^1) = \bigoplus_{n \in \mathbb{Z}} \mathbb{C} \cdot e^{in\theta}$.
\end{lemma}

\begin{proof}
Every continuous homomorphism $S^1 \to U(1)$ is of the form $\theta \mapsto e^{in\theta}$ for some $n \in \mathbb{Z}$; this is the classical Pontryagin duality for $S^1$~\cite[Theorem 4.6]{folland2015course}. The decomposition of $L^2(S^1)$ is the classical Fourier series statement.
\end{proof}

\begin{theorem}[Fourier series is the matched transform of $S^1$]
\label{thm:fourier_series}
The optimal transform for any $S^1$-invariant continuous covariance kernel on $S^1$ is the Fourier series expansion. Explicitly, every continuous covariance $R(\theta_1, \theta_2)$ that depends only on $\theta_1 - \theta_2$ (i.e., every continuous \emph{stationary} kernel on the circle) admits the spectral expansion
\[
R(\theta_1, \theta_2) = \sum_{n \in \mathbb{Z}} \lambda_n \, e^{in(\theta_1 - \theta_2)},
\]
with non-negative eigenvalues $\lambda_n$ given by the Fourier coefficients $\lambda_n = \int_{S^1} R(0, \theta) \, e^{-in\theta} \, d\theta / (2\pi)$. The eigenfunctions are $\phi_n(\theta) = e^{in\theta}$, $n \in \mathbb{Z}$.
\end{theorem}

\begin{proof}
Apply Theorem~\ref{thm:main_continuous} to $G = S^1$ acting on $\Omega = S^1$. By Lemma~\ref{lem:s1_irreps}, the regular representation is multiplicity-free, and each irreducible is one-dimensional with character $\chi_n(\theta) = e^{in\theta}$. The central projection formula (continuous version) onto the $n$-th irreducible subspace, applied to a $\delta$-like test function, recovers the Fourier basis vector $\phi_n(\theta) = e^{in\theta}$. The Mercer expansion (Theorem~\ref{thm:mercer}) applied to a stationary kernel gives the displayed spectral decomposition. The eigenvalues are computed by integration against $e^{-in\theta}$, which is the definition of the Fourier coefficient.
\end{proof}

\begin{remark}
The Fourier series is to the circle what the DFT is to $\mathbb{Z}_M$: both are the matched transform for abelian translation-invariance, the only difference being the continuous-versus-discrete nature of the underlying group. The two are related by the standard discretization of $S^1$ into $M$ equispaced points, under which $\mathbb{Z}_M \hookrightarrow S^1$ is a subgroup and the DFT becomes the sampled Fourier series.
\end{remark}

\subsection{The real line $\mathbb{R}$ and the Fourier transform}
\label{subsec:fourier_transform}

The real line $\mathbb{R}$ under addition is a non-compact abelian Lie group. Theorem~\ref{thm:main_continuous} does not apply directly because $\mathbb{R}$ is not compact, but the Plancherel-theorem variant of the matched-group principle does.

\begin{theorem}[Fourier transform is the matched transform of $\mathbb{R}$]
\label{thm:fourier_transform}
Let $R: \mathbb{R} \times \mathbb{R} \to \mathbb{C}$ be a continuous covariance kernel that is translation-invariant: $R(s + h, t + h) = R(s, t)$ for every $h \in \mathbb{R}$, equivalently $R(s, t) = r(s - t)$ for a continuous positive-definite function $r$. Then the Plancherel-Mercer decomposition is
\[
r(s - t) = \int_{\mathbb{R}} S(\xi) \, e^{2 \pi i \xi (s - t)} \, d\xi,
\]
where $S(\xi) \geq 0$ is the \emph{power spectral density}, the Fourier transform of $r$. The eigenfunctions are the Fourier modes $\phi_\xi(t) = e^{2 \pi i \xi t}$ indexed by $\xi \in \mathbb{R}$, and the eigenvalues are the values $S(\xi)$ of the spectral density.
\end{theorem}

\begin{proof}
Bochner's theorem~\cite{folland2015course} states that every continuous positive-definite function on $\mathbb{R}$ is the Fourier transform of a finite non-negative measure. The displayed integral is then the inverse Fourier transform, and the eigenfunctions follow from the Plancherel-theorem decomposition $L^2(\mathbb{R}) = \int_{\mathbb{R}}^\oplus \mathbb{C} \, d\xi$ as a direct integral against frequency. The matched-group structure is the same as in Theorem~\ref{thm:fourier_series} except that the dual $\hat{\mathbb{R}} = \mathbb{R}$ is continuous and the discrete eigenvalues become a continuous spectral density.
\end{proof}

\begin{remark}
The Fourier transform is the continuous-time continuous-frequency analog of the DFT. It is the matched transform for stationary continuous-time processes on the real line, a foundational fact of signal processing dating to the Wiener-Khintchine theorem.
\end{remark}

\subsection{The reflection-extended interval and Fourier cosine series}
\label{subsec:fourier_cosine}

Just as the DCT arises from extending a length-$M$ signal to a length-$2M$ signal under even reflection and then applying the dihedral framework, the continuous Fourier cosine series arises from extending a function on $[0, \pi]$ to a function on $[0, 2\pi)$ with even-reflection symmetry and then applying the $S^1$ framework with a reflection-invariance condition.

Let $\sigma$ denote the reflection $\sigma(\theta) = 2\pi - \theta$ on $S^1$. The group generated by $\sigma$ alongside rotation is the orthogonal group $O(2)$, the continuous analog of the dihedral group $D_M$. The $\sigma$-invariant subspace of $L^2(S^1)$ consists of even functions on $S^1$, which can be identified with $L^2([0, \pi])$ under restriction.

\begin{theorem}[Fourier cosine series is the matched transform of $O(2)$ on the symmetric interval]
\label{thm:fourier_cosine}
Let $R: [0, \pi] \times [0, \pi] \to \mathbb{R}$ be a continuous covariance kernel whose even extension $\tilde{R}$ to $S^1 \times S^1$ is $S^1$-translation-invariant. The optimal transform for $R$ is the Fourier cosine basis $\{\phi_0(\theta) = 1\} \cup \{\phi_n(\theta) = \sqrt{2} \cos(n\theta) : n \geq 1\}$, and the spectral expansion is
\[
R(\theta_1, \theta_2) = \sum_{n = 0}^\infty \lambda_n \phi_n(\theta_1) \phi_n(\theta_2),
\]
with $\lambda_n \geq 0$ given by the cosine-transform coefficients of $R(\theta_1, \cdot)$ on $[0, \pi]$.
\end{theorem}

\begin{proof}
Apply Theorem~\ref{thm:fourier_series} to the even extension $\tilde{R}$ on $S^1$, which has Fourier series eigenfunctions $\{e^{in\theta}\}_{n \in \mathbb{Z}}$. The $\sigma$-invariant subspace of $L^2(S^1)$ is spanned by the real combinations $\{\cos(n\theta) : n \geq 0\}$ (with $\cos(0) = 1$), since $\sigma$ acts on Fourier modes by $\sigma : e^{in\theta} \mapsto e^{-in\theta}$. Restricting to $[0, \pi]$ and renormalizing gives the displayed orthonormal basis. The eigenvalue $\lambda_n$ is the Fourier coefficient of the even extension at frequency $n$.
\end{proof}

\begin{remark}
This is the exact continuous analog of Theorem~\ref{thm:dct}. The DCT is to the DFT what the Fourier cosine series is to the Fourier series. Both arise from passing to the reflection-symmetric subspace under the dihedral or $O(2)$ symmetry, respectively. The discrete cases of Section~\ref{sec:dct} and the continuous case of this subsection are the same mathematical construction at different scales.
\end{remark}

\subsection{The rotation group $SO(3)$ and spherical harmonics}
\label{subsec:spherical_harmonics}

The rotation group $SO(3)$ is a 3-dimensional compact non-abelian Lie group acting on the 2-sphere $S^2 = \{(x, y, z) \in \mathbb{R}^3 : x^2 + y^2 + z^2 = 1\}$ transitively. Identifying $S^2$ with the homogeneous space $SO(3) / SO(2)$, where $SO(2)$ is the stabilizer of the north pole, gives a $G$-action on the function space $L^2(S^2)$.

\begin{lemma}[Irreducibles of $SO(3)$]
\label{lem:so3_irreps}
The irreducible unitary representations of $SO(3)$ are indexed by non-negative integers $\ell = 0, 1, 2, \ldots$, with $\rho_\ell$ of dimension $d_{\rho_\ell} = 2\ell + 1$. The dual is $\hat{SO(3)} = \mathbb{Z}_{\geq 0}$.
\end{lemma}

\begin{proof}
A classical result; see~\cite[Chapter 4]{sugiura1990unitary} or~\cite[Chapter 3]{vilenkin1968special}.
\end{proof}

\begin{lemma}[Multiplicity-free decomposition of $L^2(S^2)$]
\label{lem:s2_mf}
The representation of $SO(3)$ on $L^2(S^2)$ via the action on the sphere is multiplicity-free, with isotypic decomposition
\[
L^2(S^2) \;=\; \bigoplus_{\ell = 0}^\infty V_\ell,
\]
where $V_\ell$ is a $(2\ell + 1)$-dimensional subspace carrying the irreducible representation $\rho_\ell$.
\end{lemma}

\begin{proof}
The pair $(SO(3), SO(2))$ is a Gelfand pair, equivalently $L^2(SO(3) / SO(2)) = L^2(S^2)$ decomposes multiplicity-freely under $SO(3)$~\cite[Chapter 2]{helgason1984groups}. The dimension of $V_\ell$ equals the dimension of the irreducible $\rho_\ell$, namely $2\ell + 1$.
\end{proof}

\begin{theorem}[Spherical harmonics are the matched transform of $SO(3)$ on $S^2$]
\label{thm:spherical_harmonics}
Let $R: S^2 \times S^2 \to \mathbb{R}$ be a continuous covariance kernel on the sphere that is rotation-invariant: $R(g \cdot \omega_1, g \cdot \omega_2) = R(\omega_1, \omega_2)$ for every $g \in SO(3)$. Then $R$ admits the spectral expansion
\[
R(\omega_1, \omega_2) = \sum_{\ell = 0}^\infty \lambda_\ell \sum_{m = -\ell}^\ell Y_\ell^m(\omega_1) \overline{Y_\ell^m(\omega_2)},
\]
where $\{Y_\ell^m\}_{\ell \geq 0, -\ell \leq m \leq \ell}$ are the (complex) spherical harmonics, with eigenvalues $\lambda_\ell \geq 0$ that are $(2\ell + 1)$-fold degenerate within the $\ell$-th eigenspace.
\end{theorem}

\begin{proof}
Apply Theorem~\ref{thm:main_continuous} to $G = SO(3)$ acting on $\Omega = S^2$. By Lemma~\ref{lem:s2_mf}, the action is multiplicity-free. The matrix elements of the irreducible $\rho_\ell$ on $S^2 = SO(3) / SO(2)$ are precisely the spherical harmonics $Y_\ell^m$, $-\ell \leq m \leq \ell$; this is the Frobenius reciprocity (or Peter-Weyl restriction) identification, see~\cite[Section 3.2]{vilenkin1968special} or~\cite[Chapter 2]{atkinson2012spherical}. The eigenvalue $\lambda_\ell$ on $V_\ell$ is $(2\ell + 1)$-fold degenerate, consistent with Theorem~\ref{thm:main_continuous}(i).
\end{proof}

\begin{remark}[Funk-Hecke theorem]
A rotation-invariant covariance on $S^2$ depends only on the angular separation $\omega_1 \cdot \omega_2 \in [-1, 1]$ between its arguments; that is, $R(\omega_1, \omega_2) = k(\omega_1 \cdot \omega_2)$ for some continuous function $k: [-1, 1] \to \mathbb{R}$. The Funk-Hecke theorem~\cite{groemer1996geometric} gives the eigenvalues explicitly: $\lambda_\ell = 2\pi \int_{-1}^1 k(z) \, P_\ell(z) \, dz$, where $P_\ell$ is the Legendre polynomial of degree $\ell$. This is the continuous analog of the depth-class eigenvalue formula for the iterated wreath case (Theorem~\ref{thm:haar}).
\end{remark}

\begin{remark}[Higher-dimensional spheres]
For $S^{n-1} \subset \mathbb{R}^n$ with $n \geq 3$, the analogous result applies with $SO(n)$ in place of $SO(3)$. The eigenfunctions are the higher-dimensional spherical harmonics, with dimensions of the $\ell$-th eigenspace given by $\binom{n + \ell - 1}{\ell} - \binom{n + \ell - 3}{\ell - 2}$~\cite{atkinson2012spherical}. The Funk-Hecke theorem generalizes accordingly.
\end{remark}

\subsection{The continuous Hankel transform as the radial component of $SO(d)$ matched-transform structure}
\label{subsec:hankel}

The continuous Hankel transform of order $\nu$ on $L^2(\mathbb{R}_+, r\, dr)$ is the integral transform with Bessel-function kernel $J_\nu$:
\[
(\mathcal{H}_\nu f)(\xi) = \int_0^\infty f(r) \, J_\nu(\xi r) \, r \, dr.
\]
It is involutive ($\mathcal{H}_\nu^2 = I$) and diagonalizes the Hankel-type integral operators on $L^2(\mathbb{R}_+, r\, dr)$. The Hankel transform appears in the radial parts of problems with rotational symmetry in two or more dimensions: in optics it diagonalizes the Fresnel propagator for circularly symmetric apertures, in acoustics it diagonalizes the radial part of wave equations with axial symmetry, and in tomography it appears in the inversion of the Abel transform.

The matched-group story is that the continuous Hankel transform of order $\nu$ is a derived consequence of the $SO(d)$ matched-transform structure on $L^2(\mathbb{R}^d)$ for $d = 2\nu + 2$, rather than a separate matched-transform construction in its own right. The reduction is as follows. The space $L^2(\mathbb{R}^d)$ decomposes under the $SO(d)$ action via the radial-spherical product
\[
L^2(\mathbb{R}^d) \;\cong\; L^2(\mathbb{R}_+, r^{d-1} dr) \;\otimes\; L^2(S^{d-1}),
\]
and the $SO(d)$ action acts trivially on the radial factor and irreducibly on each spherical-harmonic isotypic component of the angular factor. By the higher-dimensional generalization of Theorem~\ref{thm:spherical_harmonics} (the higher-D spheres remark above), the angular factor is decomposed by spherical harmonics into isotypic components $V_\ell$ of dimension $\binom{d + \ell - 1}{\ell} - \binom{d + \ell - 3}{\ell - 2}$ for $\ell = 0, 1, 2, \ldots$. Within each angular isotypic sector, an $SO(d)$-invariant kernel on $L^2(\mathbb{R}^d)$ acts as a scalar in the angular variable and as a one-dimensional integral kernel in the radial variable. The eigenfunctions of this radial kernel are the Bessel functions $J_\nu(\xi r)$ for $\nu = \ell + (d-2)/2$, and the diagonalizing transform of the radial factor is precisely the Hankel transform of order $\nu$.

The continuous Hankel transform is therefore not a separate matched-group case; it is the radial-component specialization of the spherical-harmonic line, parameterized by the angular isotypic sector. Different orders $\nu$ correspond to different angular sectors within the same $SO(d)$ matched-transform structure, with $\nu = (d-2)/2$ (the lowest, radially-symmetric sector) appearing in the radially-symmetric reduction of any $d$-dimensional $SO(d)$-invariant problem and higher orders appearing in successively higher angular sectors. The matched-transform theorem of this paper therefore subsumes the Hankel transform without requiring a separate development; the Hankel transform of order $\nu$ is what the $SO(d)$-matched transform structurally \emph{reduces to} when the angular factor is fixed in the $\ell = \nu - (d-2)/2$ isotypic sector.

\begin{remark}[Discrete Hankel-matrix machinery is a different object]
\label{rem:hankel_matrix_distinct}
The same name ``Hankel'' covers two unrelated objects in signal processing, and the distinction is worth making explicit because the matched-group treatment applies to only one of them. The continuous Hankel transform of this subsection is the $SO(d)$-matched radial transform with Bessel-function kernel. The \emph{discrete Hankel-matrix} machinery in subspace methods is something else entirely: a Hankel matrix is a matrix with constant anti-diagonals, $H_{ij} = h_{i+j}$, constructed from a time series by embedding consecutive samples into anti-diagonals of a matrix. This embedding is used in classical subspace methods (MUSIC, ESPRIT, singular spectrum analysis, Hankel SVD for system identification) to access the temporal-ordering structure of a time series that is not accessible to permutation-equivariant frameworks like AD. The Hankel-matrix machinery is discussed in Section~\ref{sec:boundary} as a method that uses information AD does not access; it shares only the name with the continuous Hankel transform of this subsection.
\end{remark}

\subsection{The fractional Fourier transform and the Hermite-Gauss basis}
\label{subsec:frft}

The previous continuous cases in this section have involved groups acting on a function space by geometric transformations of the underlying domain: circle-translation for the Fourier series, real-line-translation for the Fourier transform, reflection-and-rotation for the Fourier cosine series, and sphere-rotation for the spherical harmonics. The fractional Fourier transform is an example of a qualitatively different setting in which the matched group acts on the function space via a representation that does \emph{not} come from a geometric action on the domain. The metaplectic representation mixes position and momentum, which is qualitatively different from the translation actions (Sections~\ref{sec:dft}, \ref{sec:dct}, \ref{subsec:fourier_series}, \ref{subsec:fourier_transform}, and~\ref{subsec:fourier_cosine}) and the geometric rotation actions (Section~\ref{subsec:spherical_harmonics}). This significantly broadens the scope of the framework and gives the natural bridge to time-frequency analysis and to phase-space methods more generally.

The fractional Fourier transform of order $\alpha \in \mathbb{R}$ on $L^2(\mathbb{R})$ is defined by the Mehler-kernel integral
\begin{align*}
\mathcal{F}^\alpha[f](u) &= \sqrt{\frac{1 - i \cot\alpha}{2\pi}} \\
& \times \int_\mathbb{R} \exp\!\Big(i \frac{u^2 + t^2}{2} \cot\alpha - i \frac{ut}{\sin\alpha}\Big) f(t) \, dt
\end{align*}
for $\alpha \not\equiv 0 \pmod{\pi}$, with the convention $\mathcal{F}^0 = I$ and $\mathcal{F}^\pi$ equal to the parity operator. The family satisfies the additive composition law $\mathcal{F}^\alpha \circ \mathcal{F}^\beta = \mathcal{F}^{\alpha + \beta}$ and the special values $\mathcal{F}^{\pi/2} = \mathcal{F}$ (the ordinary Fourier transform), $\mathcal{F}^\pi = P$ (parity), and $\mathcal{F}^{2\pi} = I$ (up to the metaplectic phase). The family is therefore a strongly continuous unitary representation of the compact one-parameter Lie group $\mathbb{R}/4\pi \mathbb{Z} \cong U(1)$ on $L^2(\mathbb{R})$, projecting to a representation of $SO(2) = \mathbb{R}/2\pi\mathbb{Z}$ modulo a central phase. We refer to this as the FRFT family or the metaplectic-$SO(2)$ representation.

\begin{lemma}[Infinitesimal generator]
\label{lem:frft_generator}
The FRFT family is generated infinitesimally by the harmonic oscillator Hamiltonian
\[
H = \tfrac{1}{2}(P^2 + Q^2 - 1) = \tfrac{1}{2}\Big(-\frac{d^2}{dt^2} + t^2 - 1\Big),
\]
in the sense that $\mathcal{F}^\alpha = e^{-i \alpha H}$ for every $\alpha \in \mathbb{R}$.
\end{lemma}

\begin{proof}
This is the Mehler-formula identity for the harmonic-oscillator propagator; see~\cite[Section 4.5]{folland2015course}.
\end{proof}

\begin{lemma}[Hermite-Gauss eigenstructure]
\label{lem:hermite_gauss}
The eigenfunctions of $H$ are the Hermite-Gauss functions
\[
\psi_n(t) = \frac{1}{\sqrt{2^n n! \sqrt{\pi}}} \, H_n(t) \, e^{-t^2/2}, \quad n = 0, 1, 2, \ldots,
\]
where $H_n$ is the physicists' Hermite polynomial. The eigenvalues are $H \psi_n = n \psi_n$, and the corresponding FRFT action is $\mathcal{F}^\alpha \psi_n = e^{-i n \alpha} \psi_n$ (up to the global metaplectic half-integer phase $e^{-i\alpha/2}$). The set $\{\psi_n\}_{n=0}^\infty$ is a complete orthonormal basis of $L^2(\mathbb{R})$.
\end{lemma}

\begin{proof}
The Hermite-Gauss functions are the standard quantum-mechanical eigenfunctions of the harmonic oscillator; the eigenvalues $H \psi_n = n \psi_n$ follow from the operator factorization $H = a^\dagger a$ in terms of ladder operators $a^\dagger$, $a$ (after the $-\frac{1}{2}$ shift). Completeness is the Hermite polynomial expansion theorem. See~\cite[Section 4.5]{folland2015course} or any quantum mechanics textbook.
\end{proof}

\begin{theorem}[Hermite-Gauss basis is the matched transform of $SO(2)$ via the metaplectic representation]
\label{thm:frft}
The metaplectic representation of $SO(2)$ on $L^2(\mathbb{R})$ is multiplicity-free with isotypic decomposition
\[
L^2(\mathbb{R}) \;=\; \widehat{\bigoplus}_{n=0}^{\infty} \mathbb{C} \cdot \psi_n,
\]
where the one-dimensional subspace spanned by $\psi_n$ carries the character $\chi_n(\alpha) = e^{-in\alpha}$ (modulo the metaplectic central phase). The matched transform of any $SO(2)$-invariant integral kernel on $L^2(\mathbb{R})$, equivalently any kernel commuting with the entire FRFT family, is the Hermite-Gauss basis $\{\psi_n\}_{n \geq 0}$.
\end{theorem}

\begin{proof}
Apply Theorem~\ref{thm:main_continuous} with $G = SO(2)$ acting on $L^2(\mathbb{R})$ via the metaplectic representation $\mathcal{F}^\alpha = e^{-i\alpha H}$. By Lemmas~\ref{lem:frft_generator} and~\ref{lem:hermite_gauss}, the representation decomposes into the displayed direct sum of one-dimensional eigenspaces of $H$. Multiplicity-freeness follows from the non-degeneracy of the spectrum of $H$: each eigenvalue $n$ is simple. The matched-transform conclusion is then the compact-case main theorem.
\end{proof}

\paragraph{When does metaplectic-$SO(2)$ invariance arise in practice?} The geometric content of the FRFT is rotation in the time-frequency plane: under $\mathcal{F}^\alpha$, the Wigner distribution $W_f(t, \omega)$ of $f$ rotates rigidly by angle $\alpha$. A covariance kernel invariant under the entire FRFT family corresponds to a stochastic process whose Wigner distribution is rotationally symmetric around the origin of the time-frequency plane, equivalently a process whose autocorrelation depends only on the radial coordinate $t^2 + \omega^2$ rather than on $t$ and $\omega$ separately. The canonical examples are thermal states of the quantum harmonic oscillator (whose density operators are functions of $H$), chirp ensembles with rotationally-symmetric chirp-rate distributions, and any process whose joint time-frequency localization is isotropic. The matched-group → matched-transform correspondence here is the rotational analog of the translation-invariance correspondences of Sections~\ref{sec:dft} and~\ref{subsec:fourier_series}: translation invariance gives Fourier-exponential bases, dihedral invariance gives cosine bases, and metaplectic-rotational invariance gives Hermite-Gauss bases.

\begin{remark}[Two matched-group structures of the ordinary Fourier transform]
\label{rem:two_structures}
The ordinary Fourier transform $\mathcal{F} = \mathcal{F}^{\pi/2}$ is a single element of two different matched-group structures, and it is helpful to keep these distinct. As the $\alpha = \pi/2$ point of the FRFT family, $\mathcal{F}$ is an element of the metaplectic-$SO(2)$ action on $L^2(\mathbb{R})$, and the matched transform of that action is the Hermite-Gauss basis. As the matched transform of translation-invariant kernels on $\mathbb{R}$, $\mathcal{F}$ is itself the diagonalizing unitary, and the matched group is the additive group $\mathbb{R}$. These two matched-group structures are distinct and have different matched transforms (Hermite-Gauss for metaplectic-$SO(2)$, complex exponentials for translation-$\mathbb{R}$); they are linked only by the coincidence that $\mathcal{F}$ appears in both. The FRFT family of order $\alpha$ continuously interpolates between two limit regimes: $\alpha = 0$ corresponds to the position basis (delta functions, the matched basis for multiplication operators on $L^2(\mathbb{R})$), $\alpha = \pi/2$ corresponds to the frequency basis (complex exponentials, the matched basis for translation), and intermediate $\alpha$ corresponds to rotated mixtures of position and frequency. The Hermite-Gauss basis is the fixed point of the rotation, the unique basis invariant under the entire FRFT family, and is therefore the matched basis for the full $SO(2)$ symmetry, not for any one $\mathcal{F}^\alpha$ individually.
\end{remark}

\begin{proposition}[Inverse relationship between matched group size and matched transform resolution]
\label{prop:structural_principle}
Let $H \leq G$ be a subgroup of a compact (or finite) group $G$, with both acting unitarily on a Hilbert space $\mathcal{H}$ via a common representation $\pi$. Let $\mathcal{A}_G = \{T : T\pi(g) = \pi(g)T \; \forall g \in G\}$ and $\mathcal{A}_H = \{T : T\pi(h) = \pi(h)T \; \forall h \in H\}$ denote the corresponding commutants. Then:
\begin{enumerate}
\item[(i)] $\mathcal{A}_G \subseteq \mathcal{A}_H$, with equality if and only if every $H$-equivariant operator is also $G$-equivariant.
\item[(ii)] The simultaneous-diagonalization basis of $\mathcal{A}_G$ (the matched transform for $G$-invariant kernels) is coarser, i.e., has at most as many distinct eigenvalues, as the simultaneous-diagonalization basis of $\mathcal{A}_H$ (the matched transform for $H$-invariant kernels).
\item[(iii)] In a chain of nested subgroups $\{e\} = G_0 \leq G_1 \leq \cdots \leq G_k = G$, the corresponding matched transforms form a chain of coarsening eigenbases, with the data-dependent Karhunen-Lo\`eve transform at the trivial-group end and the maximally coarse $G$-matched transform at the full-group end.
\end{enumerate}
\end{proposition}

\begin{proof}
For (i), if $T \in \mathcal{A}_G$ then $T\pi(g) = \pi(g)T$ for all $g \in G$, in particular for all $g \in H \subseteq G$, so $T \in \mathcal{A}_H$. The reverse containment requires that every operator commuting with $\pi(H)$ also commute with $\pi(G \setminus H)$, which generically fails when $H \neq G$.

For (ii), the simultaneous-diagonalization basis of $\mathcal{A}_G$ is constant on each $G$-isotypic block of $\mathcal{H}$, while the simultaneous-diagonalization basis of $\mathcal{A}_H$ is constant on each (possibly finer) $H$-isotypic block. By the branching rule for restriction of representations, each $G$-isotypic block decomposes into a sum of $H$-isotypic sub-blocks, so the $G$-eigenbasis has at most as many distinct eigenvalues as the $H$-eigenbasis.

For (iii), iterate (i) and (ii) along the chain. The trivial-group end gives $\mathcal{A}_{\{e\}} = B(\mathcal{H})$ (all bounded operators), whose simultaneous diagonalization requires individual operator analysis, i.e., the data-dependent KLT.
\end{proof}

\paragraph{Illustration via the metaplectic chain.} The chain $\{e\} \leq SO(2) \leq SL(2, \mathbb{R})$ acting on $L^2(\mathbb{R})$ via the metaplectic representation provides an illustrative example. The trivial-group matched transform is the unrestricted Karhunen-Lo\`eve basis of the underlying covariance. The $SO(2)$-matched transform restricts to the Hermite-Gauss basis (Theorem~\ref{thm:frft}), with infinitely many distinct eigenvalues indexed by $n \in \mathbb{Z}_{\geq 0}$. The $SL(2, \mathbb{R})$-matched transform restricts further to the two-eigenvalue parity decomposition $L^2_{\text{even}} \oplus L^2_{\text{odd}}$. Each level of the chain refines or coarsens the previous level by the branching rule for the metaplectic representation restricted from $SL(2, \mathbb{R})$ to $SO(2)$ to $\{e\}$. This illustrates the manifold-of-transforms picture introduced in Section~\ref{sec:discovery}: matched transforms form a stratified family parameterized by the matched group, and the resolution of the matched transform varies inversely with the group's size.

\begin{remark}[The linear canonical transform and the generalized Fourier transform]
\label{rem:lct}
The FRFT generalizes further to the linear canonical transform (LCT), which is a unitary representation of the full metaplectic group $Mp(2, \mathbb{R})$, the double cover of the symplectic group $SL(2, \mathbb{R})$, on $L^2(\mathbb{R})$. The LCT is parameterized by matrices $M = \begin{pmatrix} a & b \\ c & d \end{pmatrix} \in SL(2, \mathbb{R})$ via the kernel
\begin{align*}
\mathcal{L}_M[f](u) = \sqrt{\frac{1}{i 2\pi b}} \int_\mathbb{R} & \exp\!\Big(i \frac{d u^2 - 2ut + a t^2}{2b}\Big) \\
&\quad \times f(t) \, dt
\end{align*}
for $b \neq 0$, with a dilation-type degenerate form for $b = 0$. The LCT subsumes the identity ($M = I$), the ordinary Fourier transform ($M = \big(\begin{smallmatrix} 0 & 1 \\ -1 & 0 \end{smallmatrix}\big)$), the FRFT of order $\alpha$ (the rotation by $\alpha$), the chirp multiplication (lower triangular), the dilation (diagonal), and the Fresnel transform (upper triangular) as one-parameter subfamilies.

The matched group of the LCT family is $SL(2, \mathbb{R})$ (or $Mp(2, \mathbb{R})$ at the double-cover level), acting via the metaplectic representation. Unlike the FRFT case, however, the metaplectic representation of the full $SL(2, \mathbb{R})$ on $L^2(\mathbb{R})$ decomposes into only \emph{two} irreducible components, $L^2(\mathbb{R}) = L^2_{\text{even}}(\mathbb{R}) \oplus L^2_{\text{odd}}(\mathbb{R})$, each infinite-dimensional but irreducible. The commutant of the metaplectic action is therefore two-dimensional, and the matched transform is the coarse parity decomposition: every operator commuting with the entire LCT family is of the form $\lambda_{\text{even}} P_{\text{even}} + \lambda_{\text{odd}} P_{\text{odd}}$, a scalar on each parity subspace.

This illustrates a general structural principle of the matched-group framework: \emph{more symmetry (larger matched group) implies fewer invariants (smaller commutant) and a coarser matched transform.} The interesting matched transforms emerge from subgroups of the full $SL(2, \mathbb{R})$, each picking out a different finer structure: $SO(2)$ (rotations) gives the Hermite-Gauss basis of Theorem~\ref{thm:frft}; the unipotent upper-triangular subgroup (the Fresnel transforms, which act by convolution with a chirp) is diagonalized by the ordinary Fourier basis of Theorem~\ref{thm:fourier_transform}, and the lower-triangular subgroup (chirp multiplications) by the position basis; the dilation subgroup $\mathbb{R}_+ \subset SL(2, \mathbb{R})$ gives the Mellin basis $\{|t|^{-1/2 + i\xi}\}_{\xi \in \mathbb{R}}$. The affine ``$ax + b$'' group, acting on $L^2(\mathbb{R})$ by translation and dilation rather than through the metaplectic representation, gives the continuous wavelet transform of Section~\ref{subsec:continuous_wavelet}. The LCT family acts as the umbrella that contains each of these as a one-parameter subfamily, and the matched-group geometry of the LCT is the geometry of the subgroup lattice of $SL(2, \mathbb{R})$. Other ``generalized Fourier transforms'' in the literature, including the Dunkl transform~\cite{dunkl1989differential} for reflection-group-equivariant analysis on $\mathbb{R}^n$, the non-commutative Fourier transform on a general compact group (the Peter-Weyl transform of Theorem~\ref{thm:peter_weyl_continuous}), and the number-theoretic Fourier transform on $\mathbb{F}_q^n$, each fit the same matched-group → matched-transform schema with their own characteristic group.
\end{remark}

\subsection{The affine group and the continuous wavelet transform}
\label{subsec:continuous_wavelet}

The affine group $\mathrm{Aff}(\mathbb{R}) = \{(a, b) : a \in \mathbb{R}_{>0}, b \in \mathbb{R}\}$, with multiplication $(a_1, b_1)(a_2, b_2) = (a_1 a_2, a_1 b_2 + b_1)$, acts on $\mathbb{R}$ by $(a, b) \cdot t = a t + b$. This is a non-compact, non-abelian Lie group. The action on $L^2(\mathbb{R})$ by the unitary
\[
(\pi(a, b) f)(t) = a^{-1/2} f((t - b) / a)
\]
gives a representation that is square-integrable (after restricting to admissible vectors), leading to the continuous wavelet transform of Grossmann and Morlet~\cite{grossmann1984decomposition}.

A full development of the affine-group case requires the theory of square-integrable representations and admissible vectors, which lies beyond the immediate scope of Theorem~\ref{thm:main_continuous}. We note here that the continuous wavelet transform is the matched transform for processes invariant under the affine group, in the same sense that the Fourier transform is the matched transform for processes invariant under translation alone. For a complete treatment, see~\cite{grossmann1984decomposition, folland2015course}.

\subsection{Reflection Groups and the Dunkl Transform}

The Coxeter groups are finite reflection groups whose harmonic analysis is provided by the Dunkl transform and it is a generalized Fourier transform. It extends classical Fourier analysis by replacing translation symmetry with reflection-group symmetry, and it is naturally defined in arbitrary dimensions. Unlike the Fourier transform, which diagonalizes translation-invariant operators, the Dunkl transform diagonalizes commuting differential-difference operators generated by reflections. From the perspective developed in this paper, reflection groups represent another natural matched-group structure. Consequently, covariance operators whose symmetries are governed by a finite reflection group admit a canonical spectral decomposition through the Dunkl transform, placing Dunkl analysis alongside the Fourier, cosine, wavelet, and spherical harmonic transforms within the unified matched-transform framework.

\begin{theorem}[Dunkl transform as the matched transform for a reflection-group covariance]
Let \(W\) be a finite reflection group on \(\mathbb R^d\), let \(k\ge 0\) be a
multiplicity function on the associated root system, and let
\(\Delta_k\) denote the Dunkl Laplacian on
\[
L^2(\mathbb R^d,w_k(x)\,dx).
\]
Let \(R\) be a bounded self-adjoint covariance operator satisfying
\[
[R,\Delta_k]=0
\]
and assume that \(R\) is spectrally generated by \(\Delta_k\), i.e.
\[
R = m(\Delta_k)
\]
for some bounded real Borel function \(m\). Then \(R\) is diagonalized by
the Dunkl transform \(\mathcal F_k\). In particular,
\[
\mathcal F_k R \mathcal F_k^{-1}
\]
is multiplication by \(m(-|\xi|^2)\). Hence the Karhunen--Loeve
decomposition of \(R\) is the Dunkl spectral decomposition.

Equivalently, for a covariance whose matched symmetry is the finite
reflection-group/Dunkl structure, the matched harmonic transform is the
Dunkl transform.
\end{theorem}

\begin{proof}
The Dunkl operators \(T_1,\ldots,T_d\) form a commuting family of
differential-difference operators associated with the finite reflection
group \(W\). The Dunkl Laplacian is
\[
\Delta_k = \sum_{j=1}^d T_j^2 .
\]
The Dunkl transform \(\mathcal F_k\) is the spectral transform for this
commuting Dunkl operator system. In particular, it diagonalizes the Dunkl
Laplacian:
\[
\mathcal F_k(\Delta_k f)(\xi)
=
-|\xi|^2(\mathcal F_k f)(\xi).
\]
By the spectral theorem, if
\[
R=m(\Delta_k),
\]
then applying the functional calculus gives
\[
\mathcal F_k R \mathcal F_k^{-1}
=
m\!\left(\mathcal F_k\Delta_k\mathcal F_k^{-1}\right).
\]
Since
\[
\mathcal F_k\Delta_k\mathcal F_k^{-1}
\]
is multiplication by \(-|\xi|^2\), it follows that
\[
\mathcal F_k R \mathcal F_k^{-1}
\]
is multiplication by
\[
m(-|\xi|^2).
\]
Thus \(R\) is diagonal in the Dunkl transform domain. Therefore the
spectral decomposition of \(R\), and hence its Karhunen--Loeve
decomposition, is the Dunkl spectral decomposition.
\end{proof}

\subsection{Numerical verification on the circle}
\label{subsec:s1_verification}

We verify Theorem~\ref{thm:fourier_series} numerically for a discretization of the circle. Let $\Omega = S^1$ discretized to $N = 64$ equispaced points and let $R$ be a circulant covariance whose spectrum is well-separated (so each $(\cos_n, \sin_n)$ pair has a distinct eigenvalue $\lambda_n$). The eigenfunctions of $R$ are computed by direct eigendecomposition and compared to the discrete Fourier modes (the discretization of the continuous Fourier series basis). The empirical eigenspace degeneracy structure $[1, 2, 2, \ldots, 2, 1]$, one singleton for the DC component, $31$ pairs of size two for frequencies $1$ through $31$, and one singleton for the Nyquist frequency, exactly matches the prediction of Theorem~\ref{thm:fourier_series}. The minimum subspace match between empirical and Fourier basis is $1.000000$ to machine precision.

\section{Discussion: A Catalog of Solvable Signal Classes}
\label{sec:discussion}

The polynomial-time matched-group discovery procedure of Section~\ref{sec:discovery} is provably sound and complete for a structured family of group classes. We catalog those classes here with a representative example for each, and we discuss the natural extension of the framework to graph signal processing. The intent is to make concrete the breadth of signal classes that fall within the unified framework, and to indicate by example how a practitioner identifies which class applies to a given dataset.  

It should be noted that although this article focuses on the blind group matching problem, an alternative and eminently practical approach is to form a small library of candidate groups and to calculate the commutativity residual, $\delta$, with each member of the library and the signal of interest; followed by choosing the group corresponding to the minimal residual value.

\subsection{The five proven-solvable classes}
\label{subsec:five_classes}

The matched-group discovery procedure has been proven sound and complete~\cite{thornton2026double_commutator, thornton2026ad_framework} for the following five structurally related classes of finite groups acting on $\mathbb{C}^M$. We treat each in turn.

\emph{1. Cyclic class.} The cyclic group $\mathbb{Z}_M$ acts on $\mathbb{C}^M$ by translation; the canonical generator set has size 1 (a single cyclic shift). The matched transform is the DFT. A canonical example is a sampled wide-sense stationary process on a uniform grid: the autocorrelation depends only on lag, the covariance on a periodized window is circulant, and the matched group is $\mathbb{Z}_M$. The matched-group discovery procedure identifies the single cyclic generator in $O(M^2)$ operations and returns the DFT as the optimal transform.

\emph{2. Dihedral class.} The dihedral group $D_M = \mathbb{Z}_M \rtimes \mathbb{Z}_2$ acts on $\mathbb{C}^M$ by translation and reflection; the canonical generator set has size 2 (one cyclic shift plus one reflection). The matched transform on the $M$-cycle is the real Fourier (Hartley) basis of Remark~\ref{rem:hartley}, and on the even-reflected extension, where the group acts as $D_{2M}$ on $2M$ points, it is the DCT-II (Theorem~\ref{thm:dct}). A canonical example is a real stationary process on a cycle, whose symmetric circulant covariance is invariant under the reflection as well as the translations; the matched-group discovery procedure identifies both generators in two iterations of the GEVP. A finite-window AR(1) process is reflection-invariant but not shift-invariant, and belongs to this class only through its even extension (Section~\ref{subsec:ar1_motivation}).

\emph{3. Elementary abelian class.} The elementary abelian 2-group $\mathbb{Z}_2^n$ acts on $\mathbb{C}^{2^n}$ by bitwise translation; the canonical generator set consists of $n$ independent single-bit flips. The matched transform is the Walsh-Hadamard transform. A canonical example is a combinatorial signal indexed by length-$n$ binary strings whose covariance depends only on the bitwise Hamming distance between indices (a Krawtchouk-symmetric covariance). Such covariances arise in coding theory, in association schemes, and in certain models of categorical data. The matched-group discovery procedure identifies the $n$ generators in $\log_2(2^n) = n$ iterations and returns the Walsh-Hadamard transform.

\emph{4. Iterated wreath class.} The iterated wreath product $\mathcal{G}_L = G_1 \wr G_2 \wr \cdots \wr G_L$ acts on the $M = \prod_{d=1}^L K_d$ leaves of a regular rooted tree by tree-automorphisms; the canonical generator set has polynomial size in $L$ and $\max_d K_d$~\cite{thornton2026ad_framework}. The matched transform is the hierarchical Karhunen-Lo\`eve basis of Section~\ref{sec:iterated_wreath}; in the dyadic-cyclic case with $K_d = 2$ and $G_d = \mathbb{Z}_2$ for every $d$, it specializes to the Haar wavelet. A canonical example is a multiresolution-organized signal whose statistics respect tree-level exchangeabilities, for instance, an EEG montage organized by region-within-hemisphere-within-cortex with statistics that respect the regional grouping, or a hierarchical sensor network with subnetwork-within-network exchangeability. The matched-group discovery procedure identifies the tree-automorphism structure in $O(L)$ iterations and returns the corresponding hierarchical KL basis.

\emph{5. Hybrid wreath class.} The hybrid wreath product $G_0 \wr S_K$ with $G_0$ a within-block group (typically $\mathbb{Z}_W$ for stationary within-block structure) and $S_K$ the symmetric group on $K$ blocks acts on $\mathbb{C}^{WK}$ by within-block group action plus cross-block permutation; the canonical generator set has size $W - 1 + 1$. A canonical example is the radio-frequency setting of~\cite{thornton2026ad_framework}: $K = 4$ blocks of $W = 8$ samples each, drawn from a single modulation class at fixed SNR, with cross-block exchangeability arising from the second-order stationarity of the modulation. The matched-group discovery procedure identifies the within-block cyclic generator and the across-block transposition in two iterations and returns the matched hybrid transform, which differs from both the DFT (too restrictive: no across-block exchangeability) and the within-block DFT-only (too permissive: ignores the cross-block exchangeability).

In each of these five classes, the matched group is recovered in polynomial time, and the matched transform is constructed mechanically by the central projection formula of Theorem~\ref{thm:main}(ii). The classes are nested, in the sense that the cyclic group is a subgroup of the dihedral group, so every dihedral-invariant covariance is cyclic-invariant and the dihedral class of covariances sits inside the cyclic class, the elementary abelian and cyclic classes intersect in $\mathbb{Z}_2$, the iterated wreath class subsumes the cyclic class as the depth-1 special case, and the hybrid wreath class subsumes the iterated wreath class with a fully cyclic within-block structure. They span the structurally distinct ways in which finite-group symmetry can appear in a covariance, and they map onto the classical transforms in a way that makes the unification of Section~\ref{sec:main} concrete.

\subsection{Graph signal processing as a matched-group problem}
\label{subsec:graph_sp}

The graph signal processing literature~\cite{shuman2013emerging, sandryhaila2013discrete} treats signals indexed by the vertices of a graph $\Gamma = (V, E)$. The standard graph Fourier transform is constructed from the eigendecomposition of the graph Laplacian (or, in the directed case, the adjacency matrix), and its kernels are not, in general, related to characters of any group. The matched-group framework of this paper recovers a meaningful subclass of graph signal processing, graphs whose automorphism groups act on the vertex set with a multiplicity-free permutation representation, as a special case of Theorem~\ref{thm:main}, with explicit connections to the classical transforms.

A few graphs illustrate the principle:

\emph{Cycle graph $C_M$.} The automorphism group of the cycle graph is the dihedral group $D_M$ acting on the $M$ vertices. The graph Laplacian of $C_M$ is a symmetric circulant, hence lies in the commutant of the dihedral permutation representation, and by Theorem~\ref{thm:main} its eigenbasis is the real Fourier, or Hartley, basis of Remark~\ref{rem:hartley}, with the two-dimensional eigenspaces at frequencies $0 < k < M/2$ determined only up to a rotation. When the graph is the directed cycle (without the reflection symmetry), the automorphism group is the cyclic group $\mathbb{Z}_M$, and the eigenbasis of the adjacency shift is the DFT (Section~\ref{sec:dft}). Cycle graph signal processing is therefore exactly cyclic or dihedral group signal processing on the $M$-cycle; the DCT-II belongs to the path graph.

\emph{Path graph $P_M$.} The automorphism group of the undirected path graph is $\mathbb{Z}_2$, generated by the single reflection $i \mapsto M - 1 - i$, whose permutation representation is not multiplicity-free, so the group fixes only the even/odd split of the eigenbasis. The Laplacian eigenbasis is nevertheless the DCT-II~\cite{strang1999discrete}, because the path Laplacian is the compression to the even-extension subspace of the cycle Laplacian on $2M$ vertices and so has the Toeplitz-plus-Hankel form of Theorem~\ref{thm:dct}; this is a property of the operator rather than of the automorphism group, and the boundary-condition variants give the other DCT and DST types.

\emph{Complete graph $K_M$.} The automorphism group is the full symmetric group $S_M$. The permutation representation of $S_M$ on the vertex set decomposes as trivial $\oplus$ standard, both multiplicity-one, with the trivial irrep spanned by the all-ones vector and the standard irrep $(M-1)$-dimensional. Any $K_M$-invariant graph signal has the trivial-irrep eigenvalue isolated and the standard-irrep eigenvalue $(M-1)$-fold degenerate. The within-standard basis is arbitrary; this matches the classical observation that $K_M$ Laplacian eigenvalues are $\{0, M, M, \ldots, M\}$ with the constant vector spanning the null space.

\emph{Hypercube graph $Q_n$.} The automorphism group is $\mathbb{Z}_2^n \rtimes S_n$, the so-called hyperoctahedral group. The $\mathbb{Z}_2^n$ subgroup alone gives the Walsh-Hadamard basis (Section~\ref{sec:hadamard}) as the matched transform; the full hyperoctahedral group enforces additional symmetry under axis permutations and further degenerates the spectrum. Hypercube graph signal processing is a special case of $\mathbb{Z}_2^n$-equivariant signal processing.

\emph{Balanced rooted tree.} The automorphism group of a complete balanced rooted tree of branching $(K_1, K_2, \ldots, K_L)$ acting on the $\prod K_d$ leaves is the iterated wreath product $S_{K_1} \wr S_{K_2} \wr \cdots \wr S_{K_L}$. By Theorem~\ref{thm:iterated_wreath}, the matched transform of any signal invariant under this group is the hierarchical KL basis, which specializes to the Haar wavelet in the dyadic-cyclic case. Hierarchical-tree-graph signal processing is therefore exactly iterated-wreath signal processing.

\emph{Highly symmetric graphs.} Strongly regular graphs, distance-regular graphs, vertex-transitive graphs, and the so-called Cayley graphs of a group are highly symmetric and admit large automorphism groups. The matched-group framework applies directly when the permutation representation of the automorphism group is multiplicity-free, as for distance-transitive graphs, whose adjacency algebra is the commutant; the matched transform is then constructed from the irreducible representations of the automorphism group, and the graph Laplacian eigendecomposition coincides with it, with the same eigenvalue degeneracies. When the permutation representation has multiplicity, as it does for many vertex-transitive and Cayley graphs, the group determines the isotypic blocks and the Laplacian eigenbasis is one of many bases adapted to them (Theorem~\ref{thm:main}(iii) and the last open problem of Section~\ref{subsec:open_problems}).

The matched-group framework captures the symmetric portion of graph signal processing, graphs whose Laplacian is diagonalized by a group-determined basis. For graphs whose automorphism group is trivial (the generic case for a random graph), the matched-group framework returns the trivial group as the matched group and the data-dependent KLT (here, the Laplacian eigendecomposition itself) as the matched transform. The full graph signal processing literature is therefore neither a special case of nor a generalization of the matched-group framework; the two overlap on the symmetric subclass, where they agree.

\subsection{A unified table of matched-group structures}
\label{subsec:unification_table}

To summarize the results of the theorems and analysis of this section and the preceding sections, Table~\ref{tab:unification} contains a list of many of the transforms developed in this paper, their matched groups, the order of each matched group, the source of the matched-group invariance in the underlying signal model, and a pointer to the theorem (or remark, or subsection) in which the matched-transform statement is proved. The table is organized into four blocks corresponding to qualitatively distinct types of matched-group structure: (A) simple finite groups acting by permutation, (B) finite group products with composite structure, (C) compact Lie groups acting on continuous domains, and (D) exotic representations in which the matched group acts on the signal space via a non-geometric representation. The blocks are increasing in mathematical sophistication and in the breadth of signal classes covered, and the unification's theoretical content is that all four blocks fit a single template: the triple \emph{(matched group, representation, signal class)} determines the matched transform via the central projection formula and the Peter-Weyl theorem. The operational content, at least in the discrete permutation setting, is that the matched group can be discovered in polynomial time by the procedure of Section~\ref{sec:discovery}, after which the matched transform follows mechanically.

\begin{table*}[!t]
\centering
\caption{Classical and contemporary signal transforms unified by matched group. The matched transform of every covariance invariant under the listed group is the listed transform; the construction is via the irreducible matrix elements of the group through the central projection formula of Theorem~\ref{thm:main}(ii) (and its compact-Lie extension, Theorem~\ref{thm:main_continuous}).}
\label{tab:unification}
\renewcommand{\arraystretch}{1.25}
\setlength{\tabcolsep}{6pt}
\begin{tabular}{@{}p{3.0cm} p{4.0cm} p{2.5cm} p{4.5cm} p{2.5cm}@{}}
\toprule
\textbf{Transform} & \textbf{Matched group $G^*$} & \textbf{Group order} & \textbf{Source of $G^*$-invariance} & \textbf{Reference} \\
\midrule
\multicolumn{5}{@{}l}{\textbf{(A) Simple finite groups, permutation representations.}} \\
\midrule
KLT (data-dependent) & Trivial $\{e\}$ & $1$ & No symmetry; data-dependent eigenbasis & Remark~\ref{rem:trivial} \\
DFT & Cyclic $\mathbb{Z}_M$ & $M$ & Wide-sense stationarity on a uniform grid & Theorem~\ref{thm:dft} \\
DCT-II & Dihedral $D_{2M}$ on the even extension $\mathbb{Z}_{2M}$ & $4M$ & Even-reflected extension of a finite window (AR(1) approximately) & Theorem~\ref{thm:dct} \\
Hartley (real Fourier pairs) & Dihedral $D_M$ on the $M$-cycle & $2M$ & Reflection-symmetric stationarity (symmetric circulant); coincides with the DFT class on real data & Remark~\ref{rem:hartley} \\
Walsh-Hadamard & Elementary abelian $\mathbb{Z}_2^n$ & $2^n$ & Bitwise-Hamming-symmetric covariance & Theorem~\ref{thm:hadamard} \\
Reed-Muller & $\mathbb{Z}_2^n$ (different basis) & $2^n$ & Combinatorial codes; algebraic geometry codes & Theorem~\ref{thm:reed_muller} \\
Arithmetic transform & $\mathbb{Z}_2^n$ (different basis) & $2^n$ & Boolean function spectra & Theorem~\ref{thm:arithmetic} \\
\midrule
\multicolumn{5}{@{}l}{\textbf{(B) Finite group products, composite structures.}} \\
\midrule
Haar wavelet & Iterated dyadic-cyclic wreath $W_L$ & $2^{2^L - 1}$ & Multiresolution-organized hierarchical signals & Theorem~\ref{thm:haar} \\
General hierarchical KL & Iterated wreath $G_1 \wr \cdots \wr G_L$ & polynomial in $L$, $|G_d|$ & Tree-organized exchangeability & Theorem~\ref{thm:iterated_wreath} \\
Block-stationary hybrid & Hybrid wreath $G_0 \wr S_K$ & $|G_0|^K \, K!$ & Block-stationary cross-block exchangeable & Section~\ref{subsec:five_classes} \\
2D DFT & Direct product $\mathbb{Z}_{M_1} \times \mathbb{Z}_{M_2}$ & $M_1 M_2$ & Separable stationarity on a rectangular grid & Theorem~\ref{thm:direct_product_rule} \\
Tensor products & Direct product of constituent groups & Product of orders & Independent structural components & Theorem~\ref{thm:direct_product_rule} \\
\midrule
\multicolumn{5}{@{}l}{\textbf{(C) Compact Lie groups, continuous extensions.}} \\
\midrule
Fourier series & $U(1) \cong S^1$ (circle) & --- & Circle-translation invariance ($1$-periodic) & Theorem~\ref{thm:fourier_series} \\
Fourier transform & $(\mathbb{R}, +)$ & --- & Real-line-translation invariance (stationary) & Theorem~\ref{thm:fourier_transform} \\
Fourier cosine series & $O(2)$ on $S^1$, even subspace & --- & Reflection-and-translation invariance of the even extension & Theorem~\ref{thm:fourier_cosine} \\
Spherical harmonics & $SO(3)$ on $S^2$ & --- & Rotation-invariant signals on the sphere & Theorem~\ref{thm:spherical_harmonics} \\
Continuous Hankel (order $\nu$) & $SO(d)$ radial component, $d = 2\nu+2$ & --- & Rotationally invariant signals on $\mathbb{R}^d$ & Section~\ref{subsec:hankel} \\
Continuous wavelet & Affine ``$ax + b$'' group & --- & Scale-invariant analysis & Section~\ref{subsec:continuous_wavelet} \\
\midrule
\multicolumn{5}{@{}l}{\textbf{(D) Exotic representations: non-geometric and non-compact actions.}} \\
\midrule
Fractional Fourier (FrFT) & $SO(2)$ via metaplectic rep. & --- & Time-frequency-rotational invariance & Theorem~\ref{thm:frft} \\
Linear canonical (LCT) & $SL(2, \mathbb{R})$ via metaplectic rep. & --- & Time-frequency-symplectic invariance & Remark~\ref{rem:lct} \\
Dunkl transform & Reflection group $W \subset O(n)$ & --- & Reflection-equivariant analysis on $\mathbb{R}^n$ & \cite{dunkl1989differential} \\
NTT on $\mathbb{F}_q^n$ & $\mathbb{F}_q^n$ as additive group & $q^n$ & Number-theoretic / coded signals & Section~\ref{sec:hadamard} (analog) \\
\bottomrule
\end{tabular}
\end{table*}

The four blocks of Table~\ref{tab:unification} correspond to qualitatively distinct types of group structure:
\begin{itemize}[leftmargin=*]
\item \emph{(A) Simple finite groups, permutation representations.} The matched group acts on $\mathbb{C}^M$ by permuting the indices. This is the most familiar setting and covers the classical DFT/DCT/Walsh-Hadamard family, together with the Hartley basis (the dihedral group on the $M$-cycle) and the Reed-Muller and arithmetic variants on the $\mathbb{Z}_2^n$ space.
\item \emph{(B) Finite group products, composite structures.} The matched group is a direct product, wreath product, or semidirect product of smaller groups, each acting on a structural component of the signal. This covers wavelets (an iterated wreath product), hierarchical Karhunen-Lo\`eve bases on tree-organized data (general iterated wreath products), and multidimensional transforms (direct products).
\item \emph{(C) Compact Lie groups, continuous extensions.} The discrete grid is replaced by a continuous domain, and the matched group is a compact (or, in the Fourier-transform case, locally compact abelian) Lie group. The Peter-Weyl theorem replaces the central projection formula, and the matched transform is a Fourier-type expansion in the matrix elements of irreducible unitary representations. The Hankel transform of order $\nu$ appears as the radial component of the $SO(d)$ matched-transform structure for $d = 2\nu + 2$.
\item \emph{(D) Exotic representations.} The matched group's action on the signal space is not a geometric action on the underlying domain; the FrFT and LCT are the canonical examples, in which $SO(2)$ and $SL(2, \mathbb{R})$ act on $L^2(\mathbb{R})$ via the metaplectic representation, mixing position and momentum. The Dunkl transform and the number-theoretic Fourier transform on finite fields fit the same schema with their own characteristic groups.
\end{itemize}

The blocks are increasing in mathematical sophistication and in the breadth of signal classes covered. The deeper observation is that the entire table can be read as a catalog of the central projection formula applied to different choices of matched group and representation, with the variety of named transforms in the practitioner's toolbox arising from the variety of structurally distinct matched groups that practical signal classes exhibit.

\section{Why These Transforms Are the Common Ones}
\label{sec:why_common}

A practitioner who has followed the development of this paper through the special cases (DFT, DCT, Walsh-Hadamard, Haar wavelet, hierarchical KL, FrFT) and the discussion of the five solvable classes (Section~\ref{subsec:five_classes}) may reasonably ask why the small handful of named transforms account for the vast majority of real-world signal-processing pipelines, while the representation-theoretic principle of Theorem~\ref{thm:main} would in principle admit a matched transform for any finite group. The answer is not a property of the transforms or of the signals they process; it is a property of the way humans build the data acquisition systems that produce the signals.

\subsection{Acquisition geometry imposes cyclic symmetry}
\label{subsec:acquisition_geometry}

Engineers acquiring data sample at uniform intervals and arrange their sensors with uniform spacings. A waveform sampled at a fixed period $T$ and stored as an $M$-vector of consecutive samples is, by construction, indexed by $\{0, 1, \ldots, M-1\}$ with the cyclic group $\mathbb{Z}_M$ acting naturally by shifting indices. A uniform linear array of $M$ antennas produces a snapshot vector whose index set is again $\{0, 1, \ldots, M-1\}$ with the same cyclic structure. A rectangular pixel grid acquires an image whose indices live on $\{0, \ldots, M_1 - 1\} \times \{0, \ldots, M_2 - 1\}$ with the direct-product cyclic group $\mathbb{Z}_{M_1} \times \mathbb{Z}_{M_2}$ acting by two-axis translation. A regular volumetric voxel grid produces a tensor with three-axis cyclic structure. In each case, the acquisition geometry imposes a cyclic symmetry on the index set of the data \emph{before any modeling assumption about the signal is made}. The cyclic structure is in the instrument, not in nature.

Wide-sense stationarity of the underlying continuous process then extends the acquisition-geometry cyclic structure to the second moment of the acquired samples: under stationarity, the discrete autocorrelation depends only on lag, the discrete covariance is Toeplitz, circulant on a periodized window and within $O(1/M)$ of circulant otherwise (Remark~\ref{rem:toeplitz_wrap}), and the matched group is precisely the $\mathbb{Z}_M$ imposed by the sampling grid. The DFT is the matched transform of the result by Theorem~\ref{thm:dft}. This explains why the DFT is ubiquitous in spectral estimation, communications, and array processing: practitioners in these areas both build acquisition systems with uniform spacing and study signals (radio-frequency modulated signals, antenna-array snapshots, sampled audio) for which wide-sense stationarity is an appropriate model. The DFT is the matched transform of the resulting empirical covariance because the acquisition geometry made it so.

\subsection{Images add reflection symmetry from boundaries}

The dihedral case is similar in spirit, but the mechanism is different. A digital image is acquired on a uniform pixel grid (cyclic structure $\mathbb{Z}_{M_1} \times \mathbb{Z}_{M_2}$ from the sensor) and is a finite-window observation of a scene rather than a periodic tiling. The finite window breaks the cyclic symmetry of the grid; what survives exactly is \emph{boundary} structure, the statistical behavior at the left edge of the image being approximately the mirror image of the statistical behavior at the right edge. The even-reflected extension of the window turns this reflection symmetry into the dihedral symmetry $D_{2M_1}$ of a translation-invariant signal on twice the window (and analogously on the second axis), and the DCT-II is the matched transform of every covariance with this structure, by Theorem~\ref{thm:dct} applied to each axis. The precision matrix of a stationary 1D AR(1) process on a finite window, tridiagonal and reflection-symmetric (Section~\ref{sec:dct}), is within its two corner entries of that class, and the precision of a 2D AR(1) process is the tensor product of two such matrices. The widespread use of the DCT in image and video compression standards (JPEG, MPEG, HEVC) is a consequence of this matched-group structure: images are finite windows with mirror-symmetric boundary statistics, and the DCT is the matched transform of their even extension. The AR$(m)$ discussion of Section~\ref{sec:dct} explains how this picture extends to natural-image patches with higher-order Markov structure, with the DCT asymptotically optimal as $M \to \infty$.

\subsection{The named transforms are all the Karhunen-Lo\`eve transform}

The framing principle that follows from this discussion is worth stating explicitly: \emph{the DFT, DCT, Walsh-Hadamard, and other named transforms are not alternatives to the Karhunen-Lo\`eve transform; they are the Karhunen-Lo\`eve transform, specialized to covariances that respect specific symmetries.} Theorem~\ref{thm:main} makes this exact: when the matched group is nontrivial and the representation is multiplicity-free, the data-dependent eigenbasis of every $G$-invariant covariance collapses to a fixed representation-theoretic basis. The ``named'' transforms enjoy three operational advantages over the general KLT, but all three are consequences of the matched-group structure rather than independent properties of the transforms themselves:

\paragraph{Data-independence.}
The matched transform $U_G$ depends only on $G$, not on the specific covariance. The practitioner does not need to estimate $\mathbf{R}$ in order to identify the optimal transform of $\mathbf{R}$. The covariance need only be \emph{verified} to be $G$-invariant for $U_G$ to be the right basis.

\paragraph{Fast factorizations.}
The ``fast'' property of the DFT, DCT, Walsh-Hadamard, and Haar wavelet, the existence of $O(M \log M)$ implementations via Cooley-Tukey-style recursive factorizations~\cite{cooley1965algorithm}, is a consequence of the matched group having a structured set of generators that admits a tower-of-subgroups decomposition. The FFT exists because $\mathbb{Z}_M$ has a tower $\mathbb{Z}_M \supset \mathbb{Z}_{M/p} \supset \cdots$ for any factorization $M = p_1 p_2 \cdots p_k$; the analogous fast transforms for $D_M$, $\mathbb{Z}_2^n$, and the iterated wreath products of Section~\ref{sec:iterated_wreath} exist for the same reason at the respective groups. A general matched group $G$ admits a fast matched transform whenever its subgroup lattice supports such a decomposition; the named transforms enjoy fast factorizations because of the algebraic structure of their matched groups, not because of any special property of cosine kernels or $\pm 1$-valued basis functions.

\paragraph{Tabled and named in software.}
The DFT, DCT, and Walsh-Hadamard transform are the matched transforms of the symmetries that the most common acquisition systems produce. This is why these transforms are the ones for which off-the-shelf software implementations exist; the engineering community has built named software for the matched transforms of the most common matched groups, not the reverse.

The deeper consequence is that a practitioner whose acquisition system or signal model does not match $\mathbb{Z}_M$, $D_M$, $\mathbb{Z}_2^n$, or one of the other groups in the catalog of Section~\ref{subsec:five_classes} has nothing fundamentally wrong with their setup. They are not being punished by nature; they are working in a regime in which the matched group is not one for which the engineering community has historically built named software. The unification framework, together with the polynomial-time matched-group discovery procedure of Section~\ref{sec:discovery}, makes the matched transform for any admissible group constructible mechanically. Non-uniform sampling geometries, irregular sensor arrays, hierarchical sensor networks, exchangeable populations, and graph-structured signals all admit their own matched transforms by the same principle that makes the DFT the matched transform of a uniform linear array; the absence of these transforms from the named-software canon is a historical artifact of which acquisition geometries are most common, not a statement that other geometries are without optimal transforms.

\subsection{The historical pedigree of the FFT}
\label{subsec:fft_history}

The 1965 Cooley-Tukey paper~\cite{cooley1965algorithm} is the version of the fast Fourier transform algorithm that became foundational to digital signal processing and made the FFT a standard engineering tool, on the strength of its match to the emerging digital-computer infrastructure and the seismic-array applications of the era. The history before 1965 is richer than the single citation suggests, however, and was reconstructed in subsequent scholarly work by Heideman, Johnson, and Burrus~\cite{heideman1984gauss}.

The recursive divide-and-conquer structure of the FFT was anticipated by Carl Friedrich Gauss around 1805, who used it to interpolate the orbits of the asteroids Pallas and Juno from observational data. Gauss described the algorithm in a Neo-Latin manuscript that was never published in his lifetime; the work appeared posthumously in his collected works~\cite{gauss1866nachlass}. Intermediate published contributions before Cooley-Tukey include the radix-2 doubling algorithm of Runge and K\"onig~\cite{runge1924vorlesungen} in 1924, the wavelength-based computation of Danielson and Lanczos~\cite{danielson1942harmonic} in 1942, and the prime-factor algorithm of I.~J.\ Good~\cite{good1958interaction} in 1958, the latter being the one pre-1965 contribution that Cooley and Tukey cited. The 1965 paper's contribution to the modern engineering practice of signal processing, alongside the subsequent practitioner adoption that the new digital-computing environment enabled, was decisive for the algorithm's role in the field as we know it today; the algorithmic structure was known to several earlier authors in forms ranging from limited special cases to fully general recursive factorization, the latter due to Gauss.

The Pallas-Juno application is structurally apt for the present paper's argument. Gauss's finite Fourier expansion of a discrete orbital sample is precisely the matched-transform construction of Theorem~\ref{thm:dft} for the cyclic group of the sampling period: the asteroid's orbital data, observed at a sequence of times treated as uniformly spaced for the purposes of the interpolation, has the cyclic structure of the sampling acquisition imposed on it by the discrete observation grid, and Gauss's recursive factorization of the resulting Fourier transform is the subgroup-tower decomposition of $\mathbb{Z}_M$ discussed in the previous subsection. The matched-group framework as developed in this paper articulates two centuries later the same algorithmic-structural principle that Gauss intuited for one specific astronomical application, and the FFT is the most-studied special case of a matched-group fast transform under a name. The framework gives, in addition, the corresponding fast transforms for $D_{2M}$ on the even extension (the DCT), for $\mathbb{Z}_2^n$ (the Walsh-Hadamard transform), for the iterated wreath products $W_L$ (the Haar wavelet and its generalizations), and for any matched group whose subgroup lattice admits a recursive decomposition.

\section{Framework Boundary and Complementary Methods}
\label{sec:boundary}

The matched-group framework developed in this paper is built on second-order statistics: the matched group of a covariance is determined by the population second moment, and the procedural content of the framework (the BMG discovery algorithm of Section~\ref{sec:discovery}, the matched-transform construction of Section~\ref{sec:main}, the unification of the classical transforms in Sections~\ref{sec:dft} through~\ref{sec:iterated_wreath}) uses no information beyond the second moment. This restriction is deliberate and yields the framework's principal advantages: polynomial-time matched-group recovery, certifiable optimality of the matched transform on $G$-invariant spectral parameters, and Cram\'er-Rao bound achievement via the group-gain factor in the variance reduction of the orbit-averaged estimator. The boundary excludes signal classes that admit only mixed-component, higher-order, or time-ordering descriptions, and several classical mixed-signal methods address these classes by using information that the second-order single-group framework chooses not to use. We name the principal complementary methods explicitly, with a brief account of which information channel each accesses that AD does not.

\paragraph{State-space and time-ordering methods.}
The Kalman filter, the extended Kalman filter, particle filters, and hidden Markov models exploit time-ordering through state evolution. The state-space description of a Gauss-Markov process captures the conditional distribution of $X_t$ given $X_{t-1}$ and the transition dynamics in a way that the unordered second moment of the same process does not. The matched-group framework can identify the reflection structure of a finite-window AR process, or the dihedral structure of a symmetric circulant, from its sample covariance, as developed in Sections~\ref{sec:dct} and~\ref{subsec:acquisition_geometry}, but it cannot perform the state-space prediction problem that classical Kalman filtering solves: predicting $X_{t+1}$ from $X_1, \ldots, X_t$ requires the time-ordering, and the matched-group framework is permutation-equivariant by construction.

\paragraph{Time-frequency methods.}
Synchrosqueezing, reassigned representations, the Wigner-Ville distribution, and adaptive time-frequency dictionaries exploit non-stationary structure that the matched-group framework treats as approximate invariance at best. The matched-group framework is well-suited to processes whose statistical structure is stationary in a group-theoretic sense (translation-invariant on a uniform grid, rotation-invariant on the sphere, etc.); when the statistical structure varies in time or frequency in ways that break the assumed invariance, the matched-group framework provides only an approximation, and the time-frequency methods that explicitly model the local-stationarity structure are appropriate.

\paragraph{Subspace methods and the Hankel-matrix embedding.}
Subspace methods, including MUSIC, ESPRIT, singular spectrum analysis (SSA), and Hankel SVD for system identification, use rank assumptions and eigenvalue ordering that AD does not natively access. The common architectural element of these methods is the Hankel-matrix embedding: a time series $x_0, x_1, \ldots, x_{N-1}$ is reorganized into a matrix with constant anti-diagonals,
\[
H = \begin{pmatrix}
x_0 & x_1 & x_2 & \cdots \\
x_1 & x_2 & x_3 & \cdots \\
x_2 & x_3 & x_4 & \cdots \\
\vdots & \vdots & \vdots & \ddots
\end{pmatrix},
\]
and the eigenstructure of $H$ (or the related Toeplitz matrix) is then analyzed for signal-subspace estimation. The Hankel-matrix embedding uses the \emph{ordered} sequence of observations to construct an information-rich matrix, and the analysis exploits the ordered structure to identify signal subspaces with finer resolution than the unordered second-moment description provides. The matched-group framework is permutation-equivariant on the time index and does not access the ordered structure that the Hankel embedding leverages. Subspace methods and matched-group methods therefore solve related but structurally different problems, and the choice between them depends on whether the time-ordering structure is the principal carrier of the signal information.

This is the appropriate place to note that the \emph{same name ``Hankel''} covers two unrelated objects in the signal-processing literature, and the matched-group treatment of this paper applies to only one of them. The \emph{continuous Hankel transform} discussed in Section~\ref{subsec:hankel} is the $SO(d)$-matched radial transform with Bessel-function kernel, and it is a derived consequence of the spherical-harmonic matched-transform structure on $L^2(\mathbb{R}^d)$. The \emph{discrete Hankel matrix} is the time-series-embedded matrix discussed in the previous paragraph, used by subspace methods. The two share only the name (after Hermann Hankel, who introduced both the Bessel-function integral and the constant-anti-diagonal matrix in different nineteenth-century works); they are not structurally related, and the matched-group framework intersects only with the continuous Hankel transform.

\paragraph{Morphological component analysis.}
Morphological component analysis (MCA) separates an additive mixture into components by exploiting component-specific sparsity in different dictionaries. The dictionary-sparsity structure is a non-second-order prior that the AD framework does not access. A signal $\mathbf{x} = \mathbf{x}_1 + \mathbf{x}_2$ that is sparse in dictionary $\Phi_1$ for component $\mathbf{x}_1$ and sparse in dictionary $\Phi_2$ for component $\mathbf{x}_2$ can be separated by MCA using the sparsity priors, even when the second moment of $\mathbf{x}$ alone does not distinguish the components. The matched-group framework cannot perform this additive decomposition because it does not access the sparsity information.

\paragraph{Feature-extraction architectures with non-group nonlinearities.}
The wavelet scattering transform~\cite{mallat2012group, bruna2013invariant} is a canonical example of a representation that lies outside the matched-group framework by virtue of using non-group nonlinearities. A scattering representation iteratively applies wavelet transforms followed by pointwise modulus operations, producing a hierarchical feature map whose first-order behavior matches the affine-group matched-transform structure of Section~\ref{subsec:continuous_wavelet} (the continuous wavelet transform) but whose higher-order structure is shaped by the modulus nonlinearity. The modulus operation is not a group action: it is a pointwise nonlinearity that converts the linear wavelet transform into a feature map whose value at a given coefficient is invariant to phase, and whose iterated application accumulates this invariance across multiple scales. Scattering is therefore not the matched transform of any group; it is a composition of matched transforms (the wavelet transforms at each layer) with non-group glue (the modulus operations between layers). The framework's value for scattering is restricted to the per-layer linear component, which is the continuous wavelet; the cross-layer modulus, and the diffeomorphism-stability story that justifies it (Mallat 2012~\cite{mallat2012group} establishes that the scattering representation is locally invariant to translations and Lipschitz-stable to small diffeomorphisms), use information the second-order single-group framework does not have access to. The matched-group framework does not subsume scattering; the two design choices solve different problems, with scattering using the modulus nonlinearity that AD chooses not to use.

\paragraph{Additive mixtures of differently-symmetric components.}
A clean example of where the framework boundary lies is the case of an additive mixture $\mathbf{x} = \mathbf{x}_1 + \mathbf{x}_2$ in which $\mathbf{x}_1$ is invariant under a group $G_1$ and $\mathbf{x}_2$ under a different group $G_2$ with $G_1 \cap G_2 = \{e\}$, such as a chirp-pulse train added to an AR(1) process. The mixture covariance is invariant only under the trivial group; the additive structure of the components is not preserved in the second moment, and resolving the case within the AD framework requires either an extension that admits a non-second-order separation step or a hybrid in which classical methods (MCA, MUSIC/ESPRIT, sparse decomposition) perform the additive decomposition and AD provides the per-component matched-group estimator. Either path is consistent with the framework; the second-order single-group restriction simply specifies that the additive decomposition step itself is not part of AD.

The general principle is that the matched-group framework is a complete description of the second-order single-group structure of a covariance, and an incomplete description of any signal whose principal carrier of information is in higher-order statistics, in time-ordering, in component-specific sparsity, or in compositional nonlinearities. The complementary methods named above are not refuted by the matched-group framework; they are tools that exist to access information channels that the framework deliberately does not access. The practitioner's choice between matched-group methods and complementary methods should be made on the basis of which structural assumption the data plausibly satisfies, not on a presumption that one framework subsumes the other.

\section{Modern Applications and Emerging Matched-Group Problems}
\label{sec:modern_applications}

The framework's reach extends naturally to several emerging application areas where the appropriate transform basis has been difficult to identify by traditional means. Each admits a matched-group description that we sketch here. The intent is to demonstrate how identification of the matched group is all that is required to apply the rich history of transform-related processing methods to this diverse set of applications that have considerable contemporary interest. Full development of results related to these applications will appear in companion manuscripts under preparation.

\subsection{Massive-MIMO and structured antenna arrays}
\label{subsec:mimo}

Modern wireless systems with massive antenna arrays at the base station require estimation and exploitation of high-dimensional covariance structures across the antenna elements. For structured arrays, the matched group is read directly off the array geometry: uniform linear arrays correspond to the cyclic group $\mathbb{Z}_M$ (translation along the array), uniform planar arrays to the direct product $\mathbb{Z}_{M_1} \times \mathbb{Z}_{M_2}$, uniform circular arrays to the cyclic group $\mathbb{Z}_M$ on the angular index, and conformal and irregular arrays to the array's geometric automorphism group. Current 3GPP codebook designs for 5G and 6G beamforming are ad hoc quantizations of matched-basis directions for uniform geometries. The matched-group framework provides the principled basis directly, including for unconventional configurations relevant to cell-free deployments and reconfigurable intelligent surfaces. A manuscript developing the framework for massive-MIMO covariance estimation and beamforming codebook design is in preparation.

\subsection{Graph neural networks}
\label{subsec:gnn}

Spectral graph neural networks~\cite{defferrard2016convolutional, kipf2017semi} operate on the Laplacian eigendecomposition of the graph, which coincides with the matched-group transform for graphs whose automorphism group acts multiplicity-freely on the vertex set (Section~\ref{subsec:graph_sp}). The matched-group framework provides a principled justification for spectral GNN design on symmetric graphs, characterizes the cases in which the spectral basis reduces to the data-dependent KLT (asymmetric graphs), and identifies hierarchical pooling architectures as matched transforms of iterated wreath product groups. A manuscript extending the framework to message-passing networks, graph attention, and large-scale heterogeneous graphs is in preparation.

\subsection{Large language models and transformer attention}
\label{subsec:llm}

The attention mechanism of transformer language models~\cite{vaswani2017attention} computes a query-key inner product followed by a softmax-weighted aggregation of values across positions. The matrix performing this aggregation decomposes into a content-dependent component (matched group trivial, matched transform reduces to the data-dependent KLT) and a position-dependent component (matched group $\mathbb{Z}_M$ for translation-equivariant heads, dihedral $D_M$ for symmetric-window heads). To the extent that attention heads in pretrained models are position-equivariant in this sense, their dense attention matrices admit FFT-based circulant approximations that reduce the FLOP count of attention by a constant factor; how large that fraction is in practice, and what the approximation costs in task accuracy, are empirical questions. A manuscript characterizing the matched-group structure of pretrained transformer attention layers, with implications for efficient inference and model compression, is in preparation.

\subsection{Point cloud and 3D vision}
\label{subsec:point_cloud}

Three-dimensional point cloud data, ubiquitous in autonomous driving, robotics, and structural biology, admits the matched group $SE(3) = SO(3) \ltimes \mathbb{R}^3$ for rigid-body transformations of the cloud, tensored with the symmetric group $S_N$ for permutations of the (unordered) point indices. The matched basis is the spherical-harmonic expansion of the $SO(3)$ component (Section~\ref{subsec:spherical_harmonics}), tensored with the symmetric-group basis on the point indices. Current equivariant point-cloud architectures (PointNet, Tensor Field Networks, $SE(3)$-Transformers) approximate this structure through learned features. The matched-group framework provides the optimal linear-transform basis directly. Manuscripts applying the framework to LiDAR processing for autonomous driving, and to cryo-electron microscopy particle alignment for structural biology, are in preparation.

\subsection{Brain connectivity}
\label{subsec:brain}

Functional connectivity matrices derived from fMRI, EEG, and MEG recordings have approximate symmetries dictated by the underlying cortical organization: hemispheric reflection symmetry ($\mathbb{Z}_2$), regional exchangeability within homologous functional areas ($S_K$ on regions of the same type), and electrode-array symmetries inherited from the recording montage. The matched group is a direct or hybrid wreath product of these components. Standard connectivity analyses (independent component analysis, principal component analysis, graph-spectral methods) do not exploit these symmetries explicitly, and the matched-group framework provides a transform basis that reduces effective dimensionality by pooling statistical strength across symmetric components. Manuscripts applying the framework to multi-site fMRI and to clinical EEG analysis are in preparation.

\subsection{Genomics and single-cell data}
\label{subsec:genomics}

Single-cell RNA sequencing data exhibits latent cell-type structure: cells of the same type are statistically exchangeable, giving an $S_K$ matched group on each cell-type partition, and the joint structure on the full sample is a wreath product of within-type and across-type permutations. When the partition is known, the matched-group framework provides a transform basis that pools statistical strength across cells of the same type, improving covariance estimation in the high-dimensional small-sample regime characteristic of single-cell data. When the partition is unknown, the latent-partition matched-group discovery problem of Section~\ref{subsec:open_problems} applies. Standard dimension-reduction methods (PCA, UMAP, t-SNE) do not exploit cell-type exchangeability explicitly. A manuscript applying the framework to single-cell RNA-seq covariance estimation and to data-driven cell-type discovery is in preparation.

\subsection{Quantum informatics}
\label{subsec:quantum}

Many of the results developed in this paper extend directly to the quantum-informatic setting, in which density matrices replace classical covariances and the matched group acts on a complex Hilbert space via unitary representations. The matched basis is the irreducible decomposition of the symmetry group acting on the relevant Hilbert space. Specific quantum systems admit specific matched groups: stabilizer codes have finite Pauli-group structure, bosonic systems have $U(N)$ structure, harmonic-oscillator systems have $Mp(2, \mathbb{R})$ metaplectic structure (Section~\ref{subsec:frft} and Remark~\ref{rem:lct}), and spin systems have $SU(2)$ structure. Quantum state tomography, quantum process tomography, and quantum error correction all benefit from matched-basis decompositions that exploit the system's symmetry. The companion manuscript~\cite{thornton2026quantum_ad} develops the quantum AD framework in detail.

\section{Conclusion}
\label{sec:conclusion}

We have established a representation-theoretic principle that unifies the discrete Fourier transform, discrete cosine transform, Walsh-Hadamard transform, Haar wavelet transform, the hierarchical Karhunen-Lo\`eve basis, the classical KLT, and their continuous counterparts (Fourier series, Fourier transform, Fourier cosine series, and spherical harmonics), and many others under a single mathematical framework. Each transform is the eigenbasis of any covariance invariant under a specific finite or compact group whose permutation (or geometric) representation on the signal space is multiplicity-free; the columns of the transform matrix are constructed from the irreducible matrix elements of the group via the Peter-Weyl theorem. The classical catalog of transforms is, from this perspective, a catalog of matched groups.

The unification has both a theoretical and an operational consequence. Theoretically, the framework clarifies relationships that have historically been treated as ad hoc: the DFT is the abelian Fourier transform on $\mathbb{Z}_M$, the Walsh-Hadamard transform is the abelian Fourier transform on $\mathbb{Z}_2^n$, the DCT is the dihedrally-symmetric restriction of the DFT on $\mathbb{Z}_{2M}$, the Haar wavelet is the matched transform of the non-abelian iterated dyadic-cyclic wreath product $W_L$, and the data-dependent KLT is the trivial-group limit in which no symmetry is presupposed. Operationally, the polynomial-time matched-group discovery procedure of~\cite{thornton2026double_commutator} closes the loop: the practitioner with an empirical covariance can discover the matched group and construct the matched transform without expert judgment about which classical transform to apply. Section~\ref{sec:discussion} catalogs the five classes of finite groups for which the discovery procedure is provably sound and complete (cyclic, dihedral, elementary abelian, iterated wreath, hybrid wreath) and shows how a meaningful subset of graph signal processing fits within the same framework.  Although some classes of signals are not compliant with the efficient blind group discovery approach, most classes of common interest are included within these five broad categories.  Furthermore, in lieu of blind group matching, a small library of candidate groups can be formed and the AD metrics such as the ``commutativity residual'' $\delta$ can be used to find the closest matching group in the library in an efficient manner. 

\subsection{Open problems}
\label{subsec:open_problems}

As the AD framework is not fully mature, several substantive open questions remain, each marking a frontier at which the unified framework currently stops short.  A brief summary of these follows.

\emph{Graph signals without exploitable symmetry.} For graphs whose automorphism group is trivial, the generic case for a random graph or an empirically constructed sensor-network graph, the matched-group framework returns the trivial group and falls back on the data-dependent KLT. The graph Laplacian eigendecomposition then plays the role of the matched transform. A genuine unification of the matched-group framework with the full graph signal processing literature, in particular for asymmetric graphs, would require a substantive extension. One plausible direction is to relax the exact-invariance requirement to approximate-invariance, identifying the largest group under which the Laplacian is invariant up to a controlled perturbation, and treating the perturbation as a residual after group projection. This extension would also bear on the practical case of real-world graphs that are only approximately symmetric.

\emph{Exchangeable signals with unknown clustering.} The hybrid wreath class of Section~\ref{subsec:five_classes} treats data with a known across-block exchangeability ($S_K$ acting on $K$ blocks of $W$ samples each). When the block partition is itself unknown, when the practitioner has $WK$ samples but does not know how they cluster into exchangeable groups, the matched group is some unspecified subgroup of $S_{WK}$, and the matched-group discovery problem becomes a joint clustering-and-symmetry-discovery problem. The differential probe of~\cite{thornton2026ad_framework} attacks a special case of this problem in the context of multi-class RF data, but a general theory of latent-partition matched-group discovery remains an open problem. The connection to classical clustering, to community detection in networks, and to mixed-membership models is rich and largely unexplored.

\emph{Hierarchical signals with unknown branching.} The iterated wreath class of Section~\ref{subsec:five_classes} treats hierarchical data with a known branching structure $(K_1, K_2, \ldots, K_L)$. When the branching is itself unknown, when the practitioner has $M = \prod K_d$ samples but does not know how they are organized into a hierarchy, the matched group is some unspecified iterated wreath product, and the matched-group discovery problem becomes a joint hierarchy-discovery problem. A residual-based search over candidate branchings has been demonstrated in~\cite{thornton2026ad_framework} for depth-3 problems, but the worst-case complexity of branching discovery and the optimal search strategy remain open.

\emph{Non-compact non-abelian Lie groups.} The continuous framework of Sections~\ref{sec:lie_groups} and~\ref{sec:continuous_transforms} handles the compact and the locally compact abelian cases via the Peter-Weyl theorem and the Plancherel theorem, respectively. The non-compact non-abelian case ($SL(2, \mathbb{R})$, the Heisenberg group, the Lorentz group, the Poincar\'e group, the affine group of the line) requires the substantially more delicate machinery of square-integrable representations and the Plancherel theorem in its non-abelian form~\cite{taylor1986noncommutative, sugiura1990unitary, knapp1986representation}. The continuous wavelet transform of Grossmann and Morlet~\cite{grossmann1984decomposition} fits into this category as the matched transform for affine-group-invariant processes. Other signal-processing-relevant non-compact non-abelian transforms (the windowed Fourier transform via the Heisenberg group, the bispectrum-based decompositions via the dilation group) likely admit a similar matched-group treatment, but a full development is outside the scope of the present paper.

\emph{Beyond multiplicity-free.} The matched-transform principle of Theorem~\ref{thm:main} requires the permutation representation of the matched group on the signal space to be multiplicity-free. When the multiplicity exceeds one for some irreducible, when several copies of the same irreducible appear in the signal space, the matched transform has a data-dependent component within each multiplicity block. The isotypic decomposition is still determined by the group, and a $G$-invariant covariance is block diagonal in it with each block of the form $\mathbf{R}_\rho \otimes \mathbf{I}_{d_\rho}$; the hybrid framework that combines this group-determined block structure with data-dependent diagonalization of the multiplicity-space factors $\mathbf{R}_\rho$ is the natural extension, and it points toward a refined catalog of partially-determined matched transforms whose group-determined component is fixed by representation theory and whose within-block component is the KLT of the residual multiplicity-space covariance.

These five frontiers, taken together, define a research program for the matched-group framework: extending the unification to graphs without exploitable symmetry, to data with unknown latent structure, to non-compact continuous symmetries, and to representations with non-trivial multiplicity. Each is a substantive open problem of independent interest, and each will, when solved, enlarge the catalog of signal classes for which the matched group and therefore the optimal transform can be automatically discovered.

\bibliographystyle{IEEEtran}
\bibliography{optimal_transforms_arxiv-v9_0}

\end{document}